\documentclass[letterpaper]{article} % DO NOT CHANGE THIS
\usepackage{aaai2026}  % DO NOT CHANGE THIS
\usepackage{times}  % DO NOT CHANGE THIS
\usepackage{helvet}  % DO NOT CHANGE THIS
\usepackage{courier}  % DO NOT CHANGE THIS
\usepackage[hyphens]{url}  % DO NOT CHANGE THIS
\usepackage{graphicx} % DO NOT CHANGE THIS
\usepackage{natbib}  % DO NOT CHANGE THIS AND DO NOT ADD ANY OPTIONS TO IT
\usepackage{caption} % DO NOT CHANGE THIS AND DO NOT ADD ANY OPTIONS TO IT
\usepackage{algorithm}
\usepackage{algorithmic}

\usepackage{newfloat}
\usepackage{listings}
\DeclareCaptionStyle{ruled}{labelfont=normalfont,labelsep=colon,strut=off} % DO NOT CHANGE THIS
\floatstyle{ruled}
\newfloat{listing}{tb}{lst}{}
\floatname{listing}{Listing}
\title{Exact Algorithms for Distance to Unique Vertex Cover}
\author {
    Foivos Fioravantes\textsuperscript{\rm 1},
    Dušan Knop\textsuperscript{\rm 1},
    Nikolaos Melissinos\textsuperscript{\rm 2},
    Michal Opler\textsuperscript{\rm 1},
    Manolis Vasilakis\textsuperscript{\rm 3}
}
\affiliations {
    \textsuperscript{\rm 1}Department of Theoretical Computer Science, Faculty of Information Technology, Czech Technical University in Prague, Czech Republic\\
    \textsuperscript{\rm 2}Computer Science Institue, Faculty of Mathematics and Physics, Charles University, Prague, Chezch Republic\\
    \textsuperscript{\rm 3}Universit\'{e} Paris-Dauphine, PSL University, CNRS UMR7243, LAMSADE, Paris, France\\
    foivos.fioravantes@fit.cvut.cz, Dusan.Knop@fit.cvut.cz,
    melissinos@iuuk.mff.cuni.cz,
    michal.opler@fit.cvut.cz,
    emmanouil.vasilakis@dauphine.eu
}

\newcommand{\defproblem}[3]{
    \vspace{3mm}
    \noindent\fbox{
        \begin{minipage}{.95\columnwidth}
            #1\newline
            \textbf{Input:} #2\\
            \textbf{Task:} #3
        \end{minipage}
    }
    \vspace{3mm}
}

\newif\iflong
\newif\ifshort

\longtrue

\iflong
\else
\shorttrue
\fi
\usepackage{tikz}
\usetikzlibrary {backgrounds}

\usepackage{amssymb}
\usepackage{xcolor}
\usepackage{mathtools}
\usepackage{cleveref}
\usepackage{comment}
\usepackage{complexity}
\usepackage{xspace}
\usepackage{subfig}

\newtheorem{theorem}{Theorem}
\newtheorem{corollary}{Corollary}[theorem]

\newcommand{\PAUVC}{\textsc{Pre-Assignment for Unification of Minimum Vertex Cover}\xspace}
\newcommand{\MUVC}{\textsc{Modulator to Unique Minimum Vertex Cover}\xspace}
\newcommand{\PAUVCShort}{\textsc{PAU-VC}\xspace}
\newcommand{\MUVCShort}{\textsc{MU-VC}\xspace}
\newcommand{\UQPSAT}{\textsc{UQ Planar 1-in-3 SAT}\xspace}

\newcommand{\true}{\mathtt{true}}
\newcommand{\false}{\mathtt{false}}

\DeclareMathOperator{\tw}{tw}
\DeclareMathOperator{\full}{full}
\DeclareMathOperator{\lab}{lab}
\DeclareMathOperator{\DPt}{DP}
\newcommand{\DPcw}{\DPt^\mathsf{cw}}
\newcommand{\DPtw}{\DPt^\mathsf{tw}}
\newcommand{\DPtree}{\DPt^\mathsf{t}}
\newcommand{\DPltree}{\DPt^\mathsf{lt}}

\newtheorem{observation}{Observation}
\newtheorem{claim}{Claim}
\newcommand{\qed}{\hfill $\square$ \medskip}
\newenvironment{proof}{\noindent{\em Proof}}{\qed}
\newenvironment{proofclaim}{\noindent{\em Proof of the claim.}}{\qedclaim}
\newcommand{\qedclaim}{\hfill $\diamond$ \medskip}
\newenvironment{sketch}{\noindent{\it Sketch of proof.}}{\qedclaim}

\begin{document}

\maketitle

\begin{abstract}
In their AAAI 2024 paper, Horiyama et al. studied the problem of generating graph instances that possess a unique minimum vertex cover under specific conditions. Their approach involved pre-assigning certain vertices to be part of the solution or excluding them from it. Notably, for the \textsc{Vertex Cover} problem, pre-assigning a vertex is equivalent to removing it from the graph. Horiyama et al.~focused on maintaining the size of the minimum vertex cover after these modifications. In this work, we extend their study by relaxing this constraint: 
our goal is to ensure a unique minimum vertex cover, even if the removal of a vertex may not incur a decrease on the size of said cover.

Surprisingly, our relaxation introduces significant theoretical challenges.
We observe that the problem is $\Sigma^2_P$-complete, and remains so even for planar graphs of maximum degree 5.
Nevertheless, we provide a linear time algorithm for trees, which is then further leveraged to show that MU-VC is in \textsf{FPT} when parameterized by the combination of treewidth and maximum degree. Finally, we show that MU-VC is in \textsf{XP} when parameterized by clique-width while it is fixed-parameter tractable (FPT) if we add the size of the solution as part of the parameter.
\end{abstract}

% Uncomment the following to link to your code, datasets, an extended version or similar.
% You must keep this block between (not within) the abstract and the main body of the paper.
% \begin{links}
%     \link{Code}{https://aaai.org/example/code}
%     \link{Datasets}{https://aaai.org/example/datasets}
%     \link{Extended version}{https://aaai.org/example/extended-version}
% \end{links}

\section{Introduction}

Addressing \NP-hard problems has long been a central challenge in computer science, driving advancements in algorithmic design and computational theory. These problems, being inherently computationally intractable, have inspired diverse approaches to developing efficient and scalable solutions. Algorithmic strategies for \NP-hard problems are critical to theoretical research and find applications in areas such as network optimization, scheduling, and data analysis.

A fundamental aspect of algorithmic research is the construction of robust benchmark datasets. These datasets serve for evaluating and comparing the effectiveness of different algorithms. By providing a standardized testing ground, traditional benchmarks such as TSPLIB~\cite{tsplib}, UCI Machine Learning Repository~\cite{asuncion2007uci}, SATLIB~\cite{hoos2000satlib}, MIPLIB~\cite{koch2011miplib}, LIBSVM~\cite{chang2011libsvm}, and NetworkX graph datasets~\cite{hagberg2008exploring} have laid the groundwork for progress in this field.

More recently, modern datasets tailored to AI-driven algorithm design have emerged.
Datasets like NPHardEval~\cite{nphardeval,nphardeval4v} provide a dynamic benchmark for assessing the reasoning capabilities of large language models through algorithmic questions, including \NP-hard problems.
GraphArena~\cite{grapharena} offers a comprehensive suite of real-world graph-based tasks, enabling the evaluation of AI algorithms in diverse computational challenges, from social networks to molecular structures.
Meanwhile, MaxCut-Bench~\cite{benchmarkmaximumcut} provides an open-source platform to benchmark heuristics and methods based on machine learning specifically for the \textsc{Maximum Cut} problem.
These benchmarks not only expand the scope of algorithmic research, but also bridge traditional approaches with cutting-edge AI innovations.

Recently, the \PAUVC (\PAUVCShort) problem was introduced in~\cite{horiyama2024pauvc}, and further studied in~\cite{an2024pre}.
This was motivated by the need to create challenging datasets for algorithmic evaluation.
Intuitively, ensuring a solution is unique adds significant complexity, as solvers have no margin for error in identifying the correct solution.
A set $S$ of vertices in a graph $G$ is called a \emph{vertex cover} if $S$ intersects all edges of~$G$.
The objective of \textsc{Minimum Vertex Cover} is to compute a vertex cover $S$ of $G$ with the smallest cardinality possible.
The \textsc{Unique Vertex Cover} problem extends this by ensuring $S$ is unique.
%that the input graph has a unique minimum vertex cover.
This guarantee imposes additional constraints, making the problem particularly challenging from both theoretical and practical perspectives.

It is important to note that, in the context of vertex cover, selecting a vertex for inclusion in the cover is equivalent to deleting it from the graph along with all its incident edges.
However, ensuring that a set of vertices belongs to the deletion set can inadvertently increase the size of the minimum vertex cover under this constraint, introducing further complexity.
This leads to the definition of the \MUVC (\MUVCShort for short) problem (refer to \Cref{fig:PAUVCvsMUVC} for an illustration of the difference in the nature of the two problems): given a graph $G=(V,E)$ and an integer $k$, find a set $S\subseteq V$ such that $G-S$ has a unique minimum vertex cover and $|S|\leq k$.

Unlike \PAUVCShort, where the solver must adapt to preselected vertices to enforce uniqueness, \MUVCShort simplifies this process by reverting to solving the \textsc{Minimum Vertex Cover} problem on $G - S$.
This distinction makes \MUVCShort particularly appealing, as it maintains the standard problem formulation while achieving uniqueness.
Although this property holds for \textsc{Vertex Cover}, it does not necessarily extend to other problems, such as \textsc{Dominating Set}.

\begin{figure}
\centering
\scalebox{0.7}{
    \begin{tikzpicture}[every node/.style={circle, draw}]
    \begin{scope}
        % Nodes
        \foreach [count=\i] \x/\y in {0/0, -1/1, -1/2, 0/1, 0/2, 1/1, 1/2, -1/-.3, 1/-.3} {
            \node (\i) at (\x, \y) {$\i$};
        }
        
        % Edges
        \foreach \source/\target in {1/2,1/4,1/6,1/8,1/9,2/3,4/5,6/7} {
            \draw (\source) -- (\target);
        }
    \end{scope}
    \begin{scope}[xshift = 4.1cm]
        % Nodes
        \foreach [count=\i] \x/\y in {0/0, -1/1, -1/2, 0/1, 0/2, 1/1, 1/2, -1/-.3, 1/-.3} {
            \node (\i) at (\x, \y) {$\i$};
        }
        
        % Edges
        \foreach \source/\target in {1/2,1/4,1/6,1/8,1/9,2/3,4/5,6/7} {
            \draw (\source) -- (\target);
        }

        % Preselected vertices
        \begin{scope}[on background layer]
            \foreach \name in {2,4,6} {
                \node[circle, fill=blue!30, inner sep=6pt] at (\name) {};
            }
        \end{scope}
    \end{scope}

    \begin{scope}[xshift = 8.2cm]
        % Nodes
        \foreach [count=\i] \x/\y in {0/0, -1/1, -1/2, 0/1, 0/2, 1/1, 1/2} {
            \node (\i) at (\x, \y) {$\i$};
        }
        
        % Edges
        \foreach \source/\target in {1/2,1/4,1/6,2/3,4/5,6/7} {
            \draw (\source) -- (\target);
        }

        % Preselected vertices
        \begin{scope}[on background layer]
            \foreach \name in {2,4,6} {
                \node[circle, fill=blue!30, inner sep=6pt] at (\name) {};
            }
        \end{scope}
    \end{scope}
    \end{tikzpicture}
}
    \caption{Left: a graph $G$ with minimum vertex cover of size 4.
    There are many such -- vertex $1$ and then any set of vertices of size $3$ that intersects the pairs $\{2,3\}$, $\{4,5\}$, and $\{6,7\}$.
    Middle: \PAUVCShort solution of size 3 (the only vertex missing at this point is $1$).
    Right: \MUVCShort solution (deleted) $8,9$, and the minimum size of a vertex cover is now 3 (i.e., 5 in total); namely, $2,4,6$ which is now unique.} 
    \label{fig:PAUVCvsMUVC}
\end{figure}
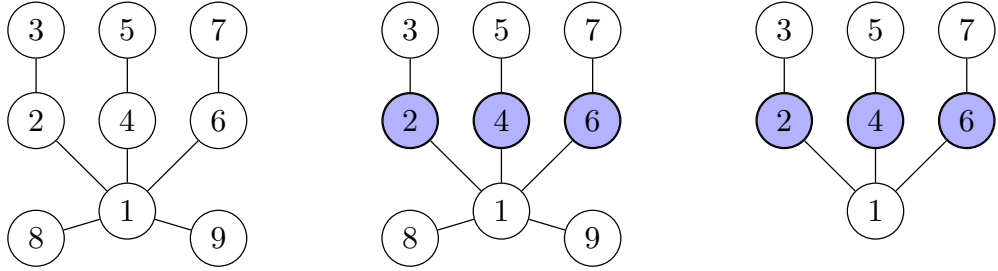

\paragraph{Our Contribution.}
It is easy to see that \MUVCShort is \NP-hard, as it generalizes the \textsc{Unique Optimal Vertex Cover} problem (for $k=0$),
which is at least as hard as \textsc{Unique SAT} which is \NP-hard~\cite{HudryL19}.
First, we precisely determine its complexity and show that \MUVCShort is actually $\Sigma_2^P$-complete,
as is the case for \PAUVCShort~\cite{horiyama2024pauvc};
in fact, in \Cref{thm:sigma_2P} we show that both problems remain so even for very restricted graph classes, namely planar graphs of maximum degree $5$ (which also improves the corresponding result in~\cite{horiyama2024pauvc}). 
Motivated by these negative results we proceed to examine \MUVCShort when the input graph $G$ has an even simpler structure, that of a tree. We initially present a polynomial-time algorithm for this case %in \Cref{subsec:trees_poly_algo}
, which is then improved into an (optimal) linear time algorithm in \Cref{thm:linear-trees}. We then tackle the problem through the parameterized complexity point of view. 
Given that the natural parameterization by $k$ cannot yield any positive results (as evidenced by the case $k=0$),
we proceed by taking into account the structure of $G$. We first consider the parameterization by treewidth, the most well-studied structural parameter. Building upon the algorithm of \Cref{thm:linear-trees}, in \Cref{thm:xp-tw} we employ DP to construct an {\XP} algorithm with running time $n^{O(2^w)}$, where $w$ is the treewidth of $G$, which is then improved into an {\FPT} algorithm in \Cref{thm:fpt-tw+Delta} when parameterized also by the maximum degree of $G$. Next, we consider the parameterization by clique-width and, in \Cref{thm:xp-cw}, we develop a DP algorithm showing that the problem belongs to {\XP}
and is solvable in $n^{O(2^d)}$ time, where $d$ denotes the clique-width of $G$.
Additionally parameterizing by the natural parameter $k$ lifts the problem to \FPT,
and in \Cref{thm:fpt-cw+k} we develop such an algorithm of running time $k^{O(2^d)} \cdot n$.

\section{Preliminaries}
Throughout the paper we use standard graph notation~\cite{Diestel17}.
All graphs considered are simple, undirected without loops.
Given a graph $G$ and a subset of its vertices $S \subseteq V(G)$, $G[S]$ denotes the subgraph induced by $S$, while $G - S$ denotes $G[V(G) \setminus S]$.
For $x, y \in \mathbb{Z}$, let $[x, y] = \{z \in \mathbb{Z} \mid x \leq z \leq y\}$, while $[x] = [1,x]$. Finally, a $k$-labeled graph~$G$ is a triple $(V, E, \lab_G)$ where $\lab_G \colon V \to [k]$.
Let us now formally define the two problems considered in this paper.

\defproblem{\PAUVC (\PAUVCShort)}
{Graph $G=(V,E)$, integer $k$.}
{Find a set $S\subseteq V$ of size $|S|\leq k$
s.t. there exists a unique minimum vertex cover of $G$ that contains $S$.}

\defproblem{\MUVC (\MUVCShort)}
{Graph $G=(V,E)$, integer $k$.}
{Find a set $S\subseteq V$, s.t. $G-S$ has a unique minimum vertex cover and $|S|\leq k$.}

Finally, recall that a decision problem is in $\Sigma_2^P$ if and only if it can be decided by a non-deterministic Turing machine with the added use of an \NP-oracle.

\iflong
Before moving on, allow us to further comment on the difference between \PAUVCShort and \MUVCShort.
In \Cref{fig:PAUVCvsMUVC} we exhibit a graph $G$ where the solutions of \PAUVCShort and \MUVCShort differ by one.
In fact, we can build a family of graphs such that this difference is arbitrarily large. 

\begin{theorem}
    For every $k \geq 3$, there exists a graph with
    minimum solutions $P$ and $M$ for \PAUVCShort and \MUVCShort respectively,
    where $|P| \ge k$ and $|M| = 2$.
\end{theorem}
\begin{proof}
    To ease the exposition, we denote by $P(G)$ ($M(G)$ resp.) a minimum \PAUVCShort (\MUVCShort resp.) solution of $G$.
    We will construct a graph $G_k$, for $k\geq 3$, such that $|P(G_k)|\geq k$ and $|M(G_k)|= 2$.
    For $k=3$, the graph $G_3$ coincides with the one illustrated in \Cref{fig:PAUVCvsMUVC};
    let us denote by $u$ the vertex labeled $1$ and by $v_1$ and $v_2$ the vertices labeled $8$ and $9$ respectively in that figure.
    Then, starting from the graph $G_{k-1}$, we construct the graph $G_k$ by introducing vertices $u_k,u'_k$ and adding the edges
    $u u_k$ and $u_k u'_k$.
    Before moving on, observe that for all $k \ge 3$, every minimum vertex cover of $G_k$ must contain the vertex $u$,
    as otherwise it would have to contain both the vertices $v_1$ and $v_2$. 

    We argue that for all $k \ge 3$, it holds that $|M(G_k)|= 2$.
    First observe that $|M(G_k)|\leq 2$.
    Indeed, the graph $G'=G_k-\{v_1,v_2\}$ is a star with $k$ leaves whose edges have been subdivided once.
    Thus, any minimum vertex cover of $G'$ has at least $k$ vertices.
    The only such vertex cover of $G'$ is the one that contains its $k$ vertices of degree $2$,
    as this is the only set that simultaneously covers all the $k$ edges that are incident to $u$ and the $k$ edges that are incident to leaves.
    It is also straightforward to observe that deleting any single vertex from $G_k$ does not result in a graph with a unique minimum vertex cover. 
    Thus, $|M(G_k)|=2$ for all $k \geq 3$. 
    
    We now argue that for all $k \ge 3$, it holds that $|P(G_k)| \geq k$.
    The proof of the statement is done by induction on $k$.
    For $k=3$ the statement is correct, as illustrated by \Cref{fig:PAUVCvsMUVC}.
    Assume now that the statement is correct for the graph $G_i$, for every $i \in [3,k-1]$;
    we will show it is also correct for $G_k$.
    Towards a contradiction, assume that $|P(G_k)|\leq k-1$. We distinguish two cases: 
    \begin{itemize}
        \item $P(G_k)\cap \{u_k,u'_k\}= \emptyset$.
        In other words, none of the vertices added to $G_{k-1}$ to build $G_k$ are in the \PAUVCShort solution of $G_k$.
        Then, any minimum vertex cover of $G_k$ will have to contain at least one of $u_k$ and $u'_k$.
        Since this vertex cover is also forced to contain $u$, both of these options are valid.
        This is a contradiction to the uniqueness of the vertex cover induced by $P(G_k)$.
        \item $P(G_k)\cap \{u_k,u'_k\}\neq \emptyset$. 
        Then, since $|P(G_k)|\leq k-1$, we have that $|P'|=|P(G_k)\cap V(G_{k-1})|\leq k-2$.
        This is a contradiction, as in this case $P'$ would be a \PAUVCShort solution of $G_{k-1}$ and
        we would have that $|P(G_{k-1})| \leq |P'| \leq k-2$, contradicting the induction hypothesis.
    \end{itemize}
    This completes the proof. 
\end{proof}
\fi
\paragraph*{Parameterized Complexity.}
The toolkit of parameterized complexity allows us to circumvent many of the limitations of classical measures of time (and space) complexity. This is achieved by considering additional measures that can affect the running time of an algorithm; these additional measures are exactly what we refer to as \textit{parameters}. The goal here is to construct exact algorithms that run in time $f(k)\cdot\operatorname{poly}(n)$, where $f$ is a computable function, $n$ is the size of the input and $k$ is the parameter; such algorithms are referred to as \textit{fixed-parameter tractable} (FPT). A problem admitting such an algorithm is said to belong in \FPT. Failing to achieve such an algorithm for a problem, we can instead try to show that it is \emph{slicewise polynomial}, i.e., that it can be determined in $n^{f(k)}$ time.
Such a problem then belongs to the class \XP.
%Similar to classical complexity theory, there is also a notion of infeasibility. A problem is presumably not in FPT{} if it is shown to be \W[1]-hard (by a parameterized reduction).
We refer the interested reader to now classical monographs~\cite{CyganFKLMPPS15,downey2012parameterized,FlumG06,Niedermeier06} for a more comprehensive introduction to this topic.

\paragraph*{Structural Parameters.}

%The most well-known structural parameter is that of \textit{treewidth}. Allow us to define it properly. 

A \emph{tree-decomposition} of $G$ is a pair $(T,\{B_x \mid x \in V(T) \})$, where~$T$ is a tree rooted at a node $r\in V(T)$, each node~$x$ in~$T$ is assigned a \emph{bag} $B_x$, and the following conditions hold:
\begin{itemize}
	\item for every edge $\{u,v\}\in E(G)$ there is a node $x\in V(T)$ such that $u,v\in B_x$, and
	\item for every vertex $v\in V$, the set of nodes $x$ with $v\in B_x$ induces a connected subtree of $T$.
\end{itemize}
The \emph{width} of a tree-decomposition $(T,\{B_x \mid x \in V(T) \})$ is $\max_{x\in V(T)} |B_x|-1$, and the treewidth $\tw(G)$ of a graph~$G$ is the minimum width of a tree-decomposition of~$G$.
It is known that computing a tree-decomposition of minimum width is in {\FPT} when parameterized by the treewidth~\cite{Bodlaender96,Kloks94}, and even more efficient algorithms exist for obtaining near-optimal tree-decompositions~\cite{KorhonenL23}.

A tree-decomposition $(T,\{B_x \mid x \in V(T) \})$ is \emph{nice} if every node $x$ in $T$ is exactly of one of the following four types: (i) Leaf: $x$ is a leaf of $T$ and $|B_x|=0$. (ii) Introduce: $x$ has a unique child $y$ and there exists $v\in V$ such that $B_x=B_y\cup \{v\}$. (iii) Forget: $x$ has a unique child $y$ and there exists $v\in V$ such that $B_y=B_x\cup \{v\}$. (iv) Join: $x$ has exactly two children $y,z$ and $B_x=B_y=B_z$.
Every graph $G=(V,E)$ admits a nice tree-decomposition of width $\tw(G)$~\cite{B98}, which is computed in linear time for graphs of bounded treewidth.

The \textit{clique-width} of a graph $G$ is an important parameter generalizing the treewidth of $G$~\cite{CourcelleO00}.
A graph of clique-width $d$ can be constructed through a sequence of the following operations on vertices that are labeled with at most $d$ different labels. 
We can use (1) introducing a single vertex~$v$ of an arbitrary label~$i$, denoted $i(v)$, (2) disjoint union of two labeled graphs, denoted $H_1 \oplus H_2$, (3) introducing edges between \emph{all} pairs of vertices of two distinct labels $i$ and $j$ in a labeled graph~$H$, denoted $\eta_{i,j}(H)$, and (4) changing the label of \emph{all} vertices of a given label~$i$ in a labeled graph~$H$ to a different label~$j$ (i.e., collapsing the pair of labels~$i$ and~$j$), denoted $\rho_{i \to j}(H)$.
An expression describes a graph~$G$ if $G$ is the final graph given by the expression (after we remove all the labels).
The \emph{width} of an expression is the number of different labels it uses.
The clique-width of a graph is the minimum width of an expression describing it.

\section{Trees}\label{sec:trees}

\begin{theorem} \label{thm:linear-trees}
The \MUVCShort problem can be solved in linear time $O(n)$ on trees.
\end{theorem}

We are working with \textit{rooted trees} $(T,r)$ such that~$T$ is a tree and $r \in V(T)$.
Similarly, a \emph{rooted forest} is just a pair $(F, r)$ such that $F$ is a forest and $r$ is its one designated vertex.
Let us define a way of representing rooted trees as algebraic terms similar to nice tree decompositions that allows us to moreover keep only one special vertex (the root) in our ``bag''.
First, we denote by $\mathsf{Leaf}(r)$ the singleton rooted tree~$(T, r)$, i.e., $V(T) = \{r\}$ and $E(T) = \emptyset$.
For a rooted tree~$(T', r')$ and $r \notin V(T)$, let $\mathsf{Extend}((T',r'),r)$ be the rooted tree~$(T, r)$ obtained from~$T'$ by adding~$r$ as the new root and joining it to~$r'$ via an edge, i.e., $V(T) = V(T') \cup \{r\}$ and $E(T) = E(T') \cup \{\{r',r\}\}$.

We remark that in terms of nice tree decompositions, the operation $\mathsf{Extend}$ corresponds to introducing the new root~$r$ and immediately afterwards forgetting the old root~$r'$.
For two rooted trees~$(T_1, r)$ and $(T_2, r)$ such that $V(T_1) \cap V(T_2) = \{r\}$, let $\mathsf{Join}((T_1,r), (T_2, r))$ be the rooted tree $(T,r)$ where $T$ is the union of $T_1$ and $T_2$, i.e., $V(T) = V(T_1) \cup V(T_2)$ and $E(T) = E(T_1) \cup E(T_2)$.
This corresponds to the usual join node in nice tree decompositions only restricted to graphs with bags of size exactly one.

We say that an expression built out of the operations $\mathsf{Leaf}$, $\mathsf{Extend}$, and $\mathsf{Join}$ is a \emph{neat tree decomposition}.
Furthermore, we say that a graph~$G$ \emph{admits a neat tree decomposition} if there is such a decomposition whose resulting in exactly~$G$.
It is easy to see that every rooted tree admits a neat tree decomposition that can be efficiently computed.

\begin{observation}\label{obs:rooted-trees}
Every rooted tree~$(T,r)$ admits a neat tree decomposition that can be computed in time $O(|V(T)|)$.
\end{observation}

\subsection{Polynomial Algorithm}\label{subsec:trees_poly_algo}
We say that a set~$M$ of vertices in a rooted forest $(F, r)$ is of type~1 if $r \in M$ and of type~0 otherwise, i.e., if $r \notin M$.
We define the \emph{reduced size of~$M$} to be~$|M \setminus \{r\}|$, i.e., we do not count the root~$r$.
A pair of functions~$(\alpha, \beta)$ where $\alpha \colon \{0,1\} \to [0, n]$ and $\beta \colon \{0,1\} \to [2]$ is a \emph{characteristic} of a set $S \subseteq V(T)\setminus \{r\}$ if for both $i \in \{0,1\}$,
\begin{itemize}
	\item $\alpha(i)$ is the reduced size of the smallest vertex cover of type~$i$ in the rooted forest~$(T - S, r)$, and
	\item $\beta(i) = 1$ if and only if there is a unique vertex cover of reduced size~$\alpha(i)$ and type~$i$ in $(T - S, r)$.
\end{itemize}
Notice that we only consider sets~$S$ that do not contain the root.
Moreover, observe that a vertex cover~$M$ of type~$0$ in~$(T-S, r)$ can be extended to a vertex cover~$M \cup \{r\}$ of the same reduced size and type~1.
This implies the following inequality for the characteristic of an arbitrary set~$S$.

\begin{observation}\label{obs:tree-alpha-positive}
For a rooted tree $(T,r)$ and a set $S \subseteq T\setminus \{r\}$ with characteristic~$(\alpha, \beta)$, we have $\alpha(1) \le \alpha(0)$.
\end{observation}

Let $G$ be the input tree and let $v \in V(G)$ be its arbitrary vertex.
By~\Cref{obs:rooted-trees}, the rooted tree~$(G,v)$ admits a neat tree decomposition~$\psi$.
The algorithm proceeds by dynamic programming along the neat tree decomposition.
Specifically, for each rooted tree~$(T,r)$ generated by a subexpression of~$\psi$, it stores a dynamic programming table $\DPtree_T[\alpha, \beta]$ such that $(\alpha, \beta)$ is a possible characteristic, and the value of $\DPtree_T[\alpha, \beta]$ contains the size of the smallest set of characteristic~$(\alpha, \beta)$ in~$(T,r)$, or~$\infty$ if no such set exists.

First, we observe that the characteristic of a given set~$S$ carries enough information to decide whether $S$ is a feasible solution to \MUVCShort.
\ifshort
%We omit the proof due to space limitations.
\fi

\begin{claim}\label{claim:tree-final}
A set~$S$ is a feasible solution to \MUVCShort in a rooted tree~$(T,r)$ if and only if one of the following holds
\begin{enumerate}
\item $r \notin S$ and  the set $S$ has characteristic $(\alpha, \beta)$ such that either $\alpha(0) < \alpha(1) + 1$ and $\beta(0) = 1$ or $\alpha(1) + 1 < \alpha(0)$ and $\beta(1) = 1$, or
\item $r \in S$ and the set $S \setminus\{r\}$ has characteristic $(\alpha, \beta)$ such that $\beta(1) = 1$.
\end{enumerate}
\end{claim}
\iflong
\begin{proofclaim}
First, suppose that $r \notin S$ and let $(\alpha, \beta)$ be the characteristic of~$S$ in~$T$.
If $S$ is a feasible solution to \MUVCShort, then there exists a unique minimum vertex cover~$M$ in $T-S$.
If~$M$ is of type~0, then $\alpha(0) = |M|$ and $\beta(0) = 1$.
Furthermore, any vertex cover of type~1 must have size strictly larger than~$|M|$ and thus, its reduced size is strictly larger than $|M| - 1$ and $\alpha(0) - 1 = |M| - 1 <\alpha(1)$.
An analogous argument proves that if $M$ is of type~1, then $\alpha(1) + 1 < \alpha(0)$ and $\beta(1) = 1$.
Conversely, suppose that $\alpha(0) < \alpha(1) + 1$ and $\beta(0) = 1$.
By definition, there is a unique minimum vertex cover~$M$ of size $\alpha(0)$ and type~0 in~$T-S$.
On the other hand, any vertex cover of type~1 has reduced size at least $\alpha(1) > \alpha(0) -1$ and thus, contains at least strictly more than~$\alpha(0) = |M|$ vertices.
It follows that $M$ is a unique minimum vertex cover in~$T-S$.
Again, an analogous argument proves that if $\alpha(1) + 1 < \alpha(0)$ and $\beta(1) = 1$, then there exists a unique minimum vertex cover in~$T-S$ of size $\alpha(1) + 1$.

Now, we deal with the case when $r \in S$.
Let $(\alpha, \beta)$ be the characteristic of $S \setminus \{r\}$ in~$T$.
Observe that a set $M \subseteq V(T)\setminus S$ is a vertex cover of~$T - S$ if and only if $M \cup \{r\}$ is a vertex cover of~$T -(S \setminus \{r\})$.
It follows that $T-S$ has a unique minimum vertex cover if and only if $T -(S \setminus \{r\})$ has a unique minimum vertex cover of type~1, that is, exactly if and only if $\beta(1)=1$.
\end{proofclaim}
\fi
%Observe that a set~$S \subseteq V(T) \setminus \{r\}$ with a characteristic~$(\alpha, \beta)$ is a feasible solution to \MUVCShort if and only if either $\alpha(0) < \alpha(1) + 1$ and $\beta(0) = 1$ or $\alpha(1) + 1 < \alpha(0)$ and $\beta(1) = 1$ because the unique minimum vertex cover is either of type~0 or type~1.
%Additionally a set $S$ that contains the root~$r$ is a feasible solution if and only if the characteristic of $S \setminus \{r\}$ is $(\alpha, \beta)$ where $\beta(1) = 1$.

Therefore, after we compute $\DPtree_G[\cdot, \cdot]$ for the input tree~$G$, the algorithm simply returns the minimum value out of (i) $\DPtree_G[\alpha, \beta]$ such that $\alpha(0) < \alpha(1) + 1$ and $\beta(0) = 1$ or $\alpha(1) + 1 < \alpha(0)$ and $\beta(1) = 1$; and (ii) $\DPtree_G[\alpha, \beta] + 1$ such that $\beta(1)=1$. Once we have received this value, we can then use the corresponding values of $\DPtree_G[\cdot, \cdot]$ to build the set $S$.  
%the following set
%\begin{multline*}
%\left\{\DPtree_G[\alpha, \beta] \mid
%\alpha(0) < \alpha(1) + 1 \text{ and } \beta(0) = 1\right\}\\
% \cup \left\{\DPtree_G[\alpha, \beta] \mid
%\alpha(0) > \alpha(1) + 1 \text{ and } \beta(1) = 1\right\} \\
% \cup \left\{\DPtree_G[\alpha, \beta] + 1 \mid
%\beta(1) = 1\right\}.
%\end{multline*}

Now, it remains to show how the individual operations act on the characteristic of a fixed set and how to use this to fill the dynamic programming tables.

\paragraph{Leaf node, $(T,r) = \mathsf{Leaf}(r)$.}
There is only a single possible choice of a set $S \subseteq V(T)\setminus \{r\}$ as the empty set and its characteristic is $(\alpha_{\mathsf{leaf}}, \beta_{\mathsf{leaf}})$ where $\alpha_{\mathsf{leaf}}(0) = \alpha_{\mathsf{leaf}}(1) = 0$ and $\beta_{\mathsf{leaf}}(0) = \beta_{\mathsf{leaf}}(1) = 1$ since there are unique vertex covers of both types with reduced size~0.
We set
\begin{equation}\label{eq:DPT-leaf}
\DPtree_{T}[\alpha, \beta] = \begin{cases}
	0 &\text{if $(\alpha, \beta) = (\alpha_{\mathsf{leaf}}, \beta_{\mathsf{leaf}})$, and}\\
	\infty &\text{otherwise.}
\end{cases}
\end{equation}

\paragraph{Join node, $(T,r) = \mathsf{Join}((T_1, r), (T_2, r))$.}
Let us define a function $f$ acting on pairs of characteristics $(\alpha_1, \beta_1), (\alpha_2, \beta_2)$ as $f((\alpha_1, \beta_1), (\alpha_2, \beta_2)) = (\alpha, \beta)$ where, for both~$i \in \{0,1\}$,
\begin{align}\label{eq:tree-join-char}
\begin{split}
	\alpha(i) &=  \alpha_1(i) + \alpha_2(i) \text{, and}\\
	\beta(i) &= \min(2,\, \beta_1(i) \cdot \beta_2(i)).
\end{split}
\end{align}

\begin{claim}\label{claim:tree-join-correctness}
Let $S_1 \subseteq V(T_1) \setminus \{r\}$ be an arbitrary set of characteristic~$(\alpha_1, \beta_1)$ in~$T_1$, let $S_2 \subseteq V(T_2) \setminus \{r\}$ be an arbitrary set of characteristic~$(\alpha_2, \beta_2)$ in~$T_2$, and let $(\alpha, \beta)$ be the image of $(\alpha_1, \beta_1), (\alpha_2, \beta_2)$ under $f$. Then $S_1 \cup S_2$ has characteristic $(\alpha, \beta)$ in~$T$.
\end{claim}
\iflong
\begin{proofclaim}
Let $M \subseteq V(T-S)$ and set $M_1 = M \cap V(T_1)$ and $M_2 = M \cap V(T_2)$.
Observe that $M$ is a vertex cover of~$T-S$ if and only if $M_1$ and $M_2$ are vertex covers of $T_1-S_1$ and $T_2-S_2$ respectively and moreover, the type of $M$ in~$T-S$ is the same as the types of $M_1$ and $M_2$ in $T_1-S_1$ and $T_2-S_2$ respectively.
It follows that for both $i \in \{0,1\}$, $M$ is a minimum vertex cover of type~$i$ in~$T-S$ if and only if $M_1$ a minimum vertex cover of type~$i$ in~$T_1-S_1$ and $M_2$ is a minimum vertex cover of type~$i$ in~$T_2-S_2$ and hence, $\alpha(i) = \alpha_1(i) + \alpha_2(i)$.
Here, it is important that $\alpha(i)$ is defined as the reduced size of the minimum vertex cover of type~$i$ and thus, we are not overcounting the root~$r$ for vertex covers of type~$1$.
Finally, there is a unique minimum vertex cover of type~$i$ in~$T-S$ if and only if there are unique minimum vertex covers of type~$i$ in both $T_1-S_1$ and $T_2-S_2$, that is, exactly if and only if $\beta_1(i) = \beta_2(i) = 1$.
\end{proofclaim}
\fi

The computation of $\DPtree_{T}[\cdot, \cdot]$ finds the smallest sum $\DPtree_{T_1}[\alpha_1, \beta_1] + \DPtree_{T_2}[\alpha_2, \beta_2]$ over the preimages of $(\alpha, \beta)$ under~$f$.
That is, we set
\begin{multline}\label{eq:DPT-join}
	\DPtree_{T}[\alpha,\beta] =\\ 	\min_{((\alpha_1, \beta_1),(\alpha_2,\beta_2)) \in f^{-1}(\alpha, \beta)} \DPtree_{T_1}[\alpha_1, \beta_1] + \DPtree_{T_2}[\alpha_2, \beta_2]
\end{multline}
where we take the minimum over the empty set to be~$\infty$.

\paragraph{Extend node, $(T,r) = \mathsf{Extend}((T', r'), r)$.}
In the extend operation, we have two possibilities depending on whether we add the old root~$r'$ to the set~$S$ or not.
Thus, we define two functions $g_{\mathsf{id}}$ and $g_{\mathsf{+}}$ acting on characteristics that describe how the characteristic of a fixed set $S \subseteq V(T') \setminus \{r'\}$ translates to the characteristics of the sets $S$ and $S \cup \{r'\}$ in $T$ respectively.
First, we define $g_{\mathsf{id}}(\alpha', \beta') = (\alpha, \beta)$\footnote{We omit the extra parentheses inside $g_{\mathsf{id}}(\cdot)$ for readability.} where
\begin{align}\label{eq:tree-extend-char-1}
\begin{split}
	\alpha(0) &=  \alpha'(1) + 1\\
\alpha(1) &= \min(\alpha'(0),\, \alpha'(1) + 1)\\
\beta(0) &= \beta'(1)\\
\beta(1) &= \begin{cases}
	\beta'(0) &\text{if $\alpha'(0) < \alpha'(1) + 1$,}\\
	\beta'(1) &\text{if $\alpha'(0) > \alpha'(1) + 1$, and}\\
	2 &\text{otherwise.}
\end{cases}
\end{split}
\end{align}
Now, we define $g_{+}(\alpha', \beta') = (\alpha, \beta)$ where
\begin{align}\label{eq:tree-extend-char-2}
\begin{split}
	\alpha(0) &= \alpha(1) =  \alpha'(1) \text{, and}\\
	\beta(0) &= \beta(1) =  \beta'(1).
\end{split}
\end{align}

\begin{claim}\label{claim:tree-extend-correctness}
Let $S \subseteq V(T') \setminus \{r'\}$ be an arbitrary set with characteristic~$(\alpha', \beta')$ in~$T'$. Then $S$ has characteristic $g_\mathsf{id}(\alpha', \beta')$ in~$T$ and $S \cup \{r'\}$ has characteristic $g_+(\alpha', \beta')$ in~$T$.
\end{claim}
\iflong
\begin{proofclaim}
First, let $(\alpha, \beta)$ be the characteristic of~$S$ in~$T$.
Observe that $M \subseteq V(T)\setminus S$ is a vertex cover of type~0 in~$T-S$ if and only if $M$ is a vertex cover of type~1 in~$T'-S$ since the edge $\{r, r'\}$ has to be covered by~$r'$.
It follows that $\alpha(0) = \alpha'(1) + 1$ and $\beta(0) = \beta'(1)$.
Note that although the size of a minimum vertex cover of type~0 in $T-S$ equals the size of a minimum vertex cover of type~1 in~$T'-S$, the plus~$1$ is due to the fact that $\alpha'(1)$ stores the \emph{reduced} size of the vertex cover, thus not accounting for vertex $r'$.
On the other hand,  $M \subseteq V(T)\setminus S$ is a vertex cover of type~1 in~$T-S$ if and only if $M \setminus \{r\}$ is a vertex cover in~$T'-S$ (of either type).
It follows that $\alpha(1)$ is equal to the smaller of the two values $\alpha'(0)$ and $\alpha'(1) + 1$ where the increase in the latter is again due to $r'$ no longer being the root in~$T$.
Moreover, the minimum vertex cover of type~1 in~$T-S$ is unique if and only if either (i) $\alpha(0) < \alpha(1) + 1$ and there is a unique minimum vertex cover of type~0 in $T'-S$, or (ii) $\alpha(1) + 1 < \alpha(0)$ and there is a unique minimum vertex cover of type~1 in $T'-S$.
This matches exactly the definition of $g_\mathsf{id}(\alpha', \beta')$.

Now, let us denote by $S_{+}$ the set $S \cup \{r'\}$ and let $(\alpha, \beta)$ be its characteristic in~$T$.
Observe that a set $M \subseteq V(T)\setminus S_{+}$ is a vertex cover of~$T - S_{+}$ if and only if $M \cup \{r'\}\setminus\{r\}$ is a vertex cover of~$T' - S$.
It follows that for both $i \in \{0,1\}$, $M \subseteq V(T)\setminus S_{+}$ is a minimum vertex cover of type~$i$ in~$T - S_{+}$ if and only if $M \cup \{r'\}\setminus\{r\}$ is a minimum vertex cover of type~1 in~$T'-S$.
Observe that the reduced size of $M$ in~$T - S_+$ is exactly the same as the reduced size of $M \cup \{r'\}\setminus\{r\}$ in~$T'-S$.
Moreover, the minimum vertex cover of both types is unique if and only if there is a unique minimum vertex cover of type~1 in~$T'-S$.
We obtain $\alpha(i) = \alpha'(1)$ and $\beta(i) = \beta'(1)$ for both $i \in \{0,1\}$ which matches the definition of $g_+(\alpha', \beta')$.
\end{proofclaim}
\fi

The computation of $\DPtree_{T}[\cdot, \cdot]$ finds the minimum between $\DPtree_{T'}[\alpha', \beta']$ over the preimages of $(\alpha, \beta)$ under~$g_\mathsf{id}$
and $\DPtree_{T'}[\alpha', \beta'] + 1$ over the preimages of $(\alpha, \beta)$ under~$g_+$.
That is, we set
\begin{multline}\label{eq:DPT-extend}
	\DPtree_{T}[\alpha,\beta] =	\\ \min \left( \begin{gathered}
	\min_{(\alpha', \beta') \in g^{-1}_\mathsf{id}(\alpha, \beta)} \DPtree_{T'}[\alpha', \beta'],\\
	\min_{(\alpha', \beta') \in g^{-1}_+(\alpha, \beta)} \DPtree_{T'}[\alpha', \beta'] + 1
	\end{gathered}   \right)
\end{multline}
where we take the minimum over empty set to be~$\infty$.

\ifshort
\paragraph{}
Let us now discuss the time complexity of the algorithm on an input tree~$G$ with $n$ vertices.
Any neat decomposition of~$G$ clearly has size $O(n)$ and the total number of possible characteristics is $O(n^2)$.
Therefore, a straightforward implementation of the computations in~\eqref{eq:DPT-leaf}, \eqref{eq:DPT-join} and \eqref{eq:DPT-extend} results in a polynomial-time algorithm.
Moreover, we remark that a more careful implementation achieves runtime $O(n^5)$. 
\fi

\iflong
\paragraph{Time complexity.}
Let us now discuss the time complexity of the algorithm on an input tree~$G$ with $n$ vertices.
Any neat decomposition of~$G$ has clearly size $O(n)$ and the total number of possible characteristics is $O(n^2)$.
Moreover, the functions $f$, $g_\mathsf{id}$, and $g_+$ are all evaluated on a single input in constant time.

The final computation and the computation in the leaves of the neat decomposition are both done straightforwardly in $O(n^2)$ time by a single traversal of the respective table~$\DPtree_T[\cdot, \cdot]$.
In an extend node,  a single entry $\DPtree_T[\alpha, \beta]$ is computed in time~$O(n^2)$ by enumerating over all possible characteristics~$(\alpha', \beta')$ (refer to~\eqref{eq:DPT-extend}).
This makes the computation of the whole table~$\DPtree_t[\alpha, \beta]$ finish in $O(n^4)$ time.
Finally, the bottleneck are the join nodes where a trivial implementation would take $O(n^4)$ time per entry (refer to~\eqref{eq:DPT-join}).
To improve over this, we start by initially setting every entry in the table to~$\infty$.
Then we iterate over all possible pairs of characteristics $(\alpha_1, \beta_1)$, $(\alpha_2, \beta_2)$.
For each pair, we first compute the value $f((\alpha_1, \beta_1), (\alpha_2, \beta_2))$, denoted by $(\alpha, \beta)$.
Afterwards, we update $\DPtree_T[\alpha, \beta]$ to $\DPtree_{T_1}[\alpha_1, \beta_1] + \DPtree_{T_2}[\alpha_2, \beta_2]$ but only if it is smaller than its current value.
This takes only $O(n^4)$ time in total.

Overall, the computation takes $O(n^4)$ time for each subexpression of the neat tree decomposition and thus, the whole algorithm terminates in $O(n^5)$ time.
\fi

\paragraph{}
The correctness of the algorithm follows straightforwardly from  Claims~\ref{claim:tree-join-correctness} and~\ref{claim:tree-extend-correctness} by a bottom-up induction over the neat tree decomposition.
\iflong
However, we chose to omit the full proof here and include one only for the more efficient algorithm presented in the subsection dedicated to trees.%\Cref{subsec:trees_linear_algo}.
\fi

\subsection{Linear Algorithm}\label{subsec:trees_linear_algo}
In order to speed up the previous algorithm, %of \Cref{subsec:trees_poly_algo}
we use two ideas that allow us to group the possible characteristics into constantly many classes.
First, the size of the minimum vertex cover is irrelevant in \MUVCShort. Thus, capturing $\alpha(0) - \alpha(1)$ suffices to see whether the minimum vertex covers of two possible types have the same size.
A reasonable implementation would already bring down the runtime to~$O(n^3)$.
However, it actually suffices to remember whether $\alpha(0) - \alpha(1)$ is equal to 0, 1, or at least~2.

Formally, a \emph{reduced characteristic} of a set $S \subseteq V(T)\setminus \{r\}$ with characteristic $(\alpha, \beta)$ in~$(T, r)$ is a pair~$(\delta, \beta)$ where $\delta = \min(2,\alpha(0) - \alpha(1))$.
By \Cref{obs:tree-alpha-positive}, we see that $\delta \in \{0,1,2\}$.
Moreover, $\beta$ is one of $4$ possible functions. Thus there are only $3 \cdot 4 = 12$ possible reduced characteristics. The algorithm follows the same scheme as before.
In particular for each rooted tree~$(T,r)$ generated by a subexpression of the neat tree decomposition, it fills a dynamic programming table $\DPltree_T[\cdot, \cdot]$ such that $\DPltree_T[\delta, \beta]$ contains the size of the smallest set $S \subseteq V(T) \setminus\{r\}$ with reduced characteristic~$(\delta, \beta)$.

We start by observing that the feasibility of a given set~$S$ as a solution to \MUVCShort can be deduced from its reduced characteristic.
\begin{claim}
	A set~$S$ is a feasible solution to \MUVCShort in a rooted tree~$(T,r)$ if and only if one of the following holds
	\begin{enumerate}
		\item $r \notin S$ and  the set $S$ has reduced characteristic $(\delta, \beta)$ such that either $\delta = 0$ and $\beta(0) = 1$, or $\delta = 2$ and $\beta(1) = 1$, or
		\item $r \in S$ and the set $S \setminus\{r\}$ has reduced characteristic $(\delta, \beta)$ such that $\beta(1) = 1$.
	\end{enumerate}
\end{claim}
\iflong
\begin{proofclaim}
It suffices to check that the conditions imposed on reduced characteristics exactly match the conditions imposed on full characteristics in \Cref{claim:tree-final}.
\end{proofclaim}
\fi

So after computing the whole table $\DPltree_G[\cdot,\cdot]$ for the input tree~$G$, the algorithm returns the minimum value out of
\ifshort
(i) $\DPltree_G[\delta, \beta]$ where $\delta = 0$ and $\beta(0) = 1$, or $\delta = 2$ and $\beta(1) = 1$; and (ii) $\DPltree_G[\delta, \beta]+1$ where $\beta(1) = 1$.
\fi\iflong
the set
\begin{multline*}
	\left\{\DPltree_G[\delta, \beta] \middle| 
	\begin{gathered}
\delta = 0 \text{ and } \beta(0) = 1, \text{or}\\
\delta = 2 \text{ and } \beta(1) = 1
	\end{gathered}\right\}\\
	\cup \left\{\DPltree_G[\delta, \beta] + 1 \mid
	\beta(1) = 1\right\}.
\end{multline*}
\fi

The most important part is to show that reduced characteristics still allow a similar scheme of computation to the full characteristics.
Namely, we show how the reduced characteristic of a set changes under the $\mathsf{Leaf}$, $\mathsf{Extend}$, and $\mathsf{Join}$ operations.
%As before, we omit the proofs of these transformations due to space limitations.

\paragraph{Leaf node, $(T,r) = \mathsf{Leaf}(r)$.}
We have seen that the characteristic of the only possible set~$S \subseteq V(T)\setminus\{r\}$ is $(\alpha_{\mathsf{leaf}}, \beta_{\mathsf{leaf}})$ and thus, its reduced characteristic is $(\delta_{\mathsf{leaf}}, \beta_{\mathsf{leaf}})$ where $\delta_{\mathsf{leaf}} = \alpha_{\mathsf{leaf}}(0) - \alpha_{\mathsf{leaf}}(1) = 0$.
Again, we set
\begin{equation}\label{eq:DPT-reduced-leaf}
\DPltree_T[\delta, \beta] = \begin{cases}
	0 &\text{if $(\delta, \beta) = (\delta_{\mathsf{leaf}}, \beta_{\mathsf{leaf}})$, and}\\
	\infty &\text{otherwise.}
\end{cases}
\end{equation}

\paragraph{Join node, $(T,r) = \mathsf{Join}((T_1, r), (T_2, r))$.}
Let us define a function $f'$ acting on pairs of reduced characteristics $(\delta_1, \beta_1), (\delta_2, \beta_2)$ as $f'((\delta_1, \beta_1), (\delta_2, \beta_2)) = (\delta, \beta)$ where
\begin{align}
\begin{split}
	\delta &=  \min(2,\, \delta_1 + \delta_2) \text{, and}\\
	\beta(i) &= \min(2,\, \beta_1(i) \cdot \beta_2(i)) \text{ for both $i \in \{0,1\}$}.
\end{split}
\end{align}

\begin{claim}\label{claim:tree-reduced-join-correctness}
	Let $S_1 \subseteq V(T_1) \setminus \{r\}$ be an arbitrary set of reduced characteristic~$(\delta_1, \beta_1)$ in~$T_1$, let $S_2 \subseteq V(T_2) \setminus \{r\}$ be an arbitrary set of reduced characteristic~$(\delta_2, \beta_2)$ in~$T_2$. Then $S_1 \cup S_2$ has reduced characteristic $f'((\delta_1, \beta_1),(\delta_2, \beta_2))$ in~$T$.
\end{claim}
\iflong
\begin{proofclaim}
For both $i \in [2]$, let $(\alpha_i, \beta_i)$ be the full characteristic of $S_i$ in~$T_i$ and moreover, let $(\alpha, \beta)$ and $(\delta, \beta)$ be the full and reduced characteristic of $S_1 \cup S_2$ in~$T$ respectively.
First, observe that the computation of~$\beta$ out of $\beta_1$ and~$\beta_2$ is exactly the same as in \eqref{eq:tree-join-char} and thus, its correctness follows directly from \Cref{claim:tree-join-correctness}.

We can express~$\delta$ as follows
\begin{align*}
\delta &= \min\left(2,\, \alpha(0)-\alpha(1)\right)\\
 &= \min\left(2,\, \alpha_1 (0) + \alpha_2(0) - \alpha_1(1) - \alpha_2(1)\right)\\
  &= \min\left(2,\, \left(\alpha_1 (0) - \alpha_1(1)\right) + \left(\alpha_2(0)  - \alpha_2(1)\right)\right)
\end{align*}
where the first equality is the definition and the second equality is due to \Cref{claim:tree-join-correctness}.

Now, if $\delta_1 = 2$ then $\alpha_1(0) - \alpha_1(1)$ is at least~2 and the same holds for the sum $(\alpha_1 (0) - \alpha_1(1)) + (\alpha_2(0)  - \alpha_2(1))$ since $\alpha_2(0)  - \alpha_2(1)$ is non-negative.
It follows that in this case $\delta = 2$ which matches the definition of $f'$.
The same holds analogously when $\delta_2 = 2$.

It remains to  consider the case when both $\delta_1$ and $\delta_2$ are at most~1.
In that case, we have $\delta_i = \alpha_i(0) - \alpha_i(1)$ for both~$i$ and as a consequence, $\delta = \min(2, \delta_1 + \delta_2)$ which again matches the definition of $f'$.
\end{proofclaim}
\fi
\ifshort
\fi

The computation of $\DPltree_{T}[\cdot, \cdot]$ follows the same scheme as in the previous case.
That is, we set
\begin{multline}\label{eq:DPT-reduced-join}
	\DPltree_{T}[\delta,\beta] =\\ 	\min_{((\delta_1, \beta_1),(\delta_2,\beta_2)) \in (f')^{-1}(\delta, \beta)} \DPltree_{T_1}[\delta_1, \beta_1] + \DPltree_{T_2}[\delta_2, \beta_2]
\end{multline}

where we take the minimum over the empty set to be~$\infty$.

\paragraph{Extend node, $(T,r) = \mathsf{Extend}((T', r'), r)$.}
\iflong
In the extend operation, we have two possibilities depending on whether we add the old root~$r'$ to the set~$S$ or not.
\fi
We again define two functions $g'_{\mathsf{id}}$ and $g'_{\mathsf{+}}$ describing how the reduced characteristics of the set $S$ and $S \cup \{r'\}$ in~$T$ depend on the reduced characteristic of a fixed set $S \subseteq V(T') \setminus \{r'\}$ in~$T'$.
First, we define~$g'_{\mathsf{id}}$ such that $g'_{\mathsf{id}}(\delta', \beta') = (\delta, \beta)$ where
\begin{align}\label{eq:tree-reduced-extend-char-1} 
\begin{split}
\delta &=  \begin{cases}
		0 &\text{if $\delta' \ge 1$,}\\
		1	 &\text{otherwise, i.e., if $\delta' = 0$.}
	\end{cases}\\
	\beta(0) &= \beta'(1), \text{ and}\\
	\beta(1) &= \begin{cases}
		\beta'(0) &\text{if $\delta' = 0$,}\\
		\beta'(1) &\text{if $\delta' \ge 2$, and}\\
		2 &\text{otherwise, i.e., if $\delta' = 1$.}
	\end{cases}
\end{split}
\end{align}
Now, we define~$g'_{+}$ such that $g'_{+}(\delta', \beta') = (\delta, \beta)$ where
\begin{equation}\label{eq:tree-reduced-extend-char-2} 
	\delta = 0, \qquad
	\beta(0) = \beta(1) =  \beta'(1).
\end{equation}

\begin{claim}\label{claim:tree-reduced-extend-correctness}
	Let $S \subseteq V(T') \setminus \{r'\}$ be an arbitrary set with reduced characteristic~$(\delta', \beta')$ in~$T'$. Then $S$ has reduced characteristic $g'_\mathsf{id}(\delta', \beta')$ in~$T$ and $S \cup \{r'\}$ has reduced characteristic $g'_+(\delta', \beta')$ in~$T$.
\end{claim}
\iflong
\begin{proofclaim}
To prove the claim, it suffices to syntactically verify that the reduced characteristic obtained by $g'_\mathsf{id}$ and $g'_+$ in~\eqref{eq:tree-reduced-extend-char-1} and~\eqref{eq:tree-reduced-extend-char-2} matches the definition of $g_\mathsf{id}$ and $g_+$ in~\eqref{eq:tree-extend-char-1} and~\eqref{eq:tree-extend-char-2}.
The claim then follows from \Cref{claim:tree-extend-correctness}.

Let $(\alpha', \beta')$ denote the full characteristic of~$S$ in~$T'$ and let $(\alpha, \beta)$ and $(\delta, \beta)$ denote the full  and reduced characteristic of $S$ in~$T$ respectively.
If $\delta' \ge 1$, then $\alpha'(0) - \alpha'(1) \ge 1$ by definition.
Therefore, we have $\alpha(0) = \alpha(1) = \alpha'(1) + 1$ by~\eqref{eq:tree-extend-char-1} and $\delta = 0$ which matches the definition of~$g'_\mathsf{id}$.
On the other hand if $\delta' = 0$, then $\alpha'(0) = \alpha'(1)$ and we have $\alpha(0) = \alpha'(1) + 1$ and $\alpha(1) = \alpha'(0) = \alpha'(1)$ by~\eqref{eq:tree-extend-char-1}.
Hence, we obtain $\delta = 1$ which again matches the definition of~$g'_\mathsf{id}$.
As for the computation of~$\beta$, it suffices to notice that the conditions `$\delta' = 0$', `$\delta' \ge 2$', and `$\delta' = 1$' are exactly equivalent to '$\alpha'(0) < \alpha'(1)+1$', `$\alpha'(0) > \alpha'(1) + 1$', and  `$\alpha'(0) = \alpha'(1)+1$' respectively.
As a result, the function~$g'_\mathsf{id}$ acts on the second coordinate of characteristics exactly in the same way as~$g_\mathsf{id}$.

Finally, it is straightforward to check that the full characteristic of $S \cup \{r'\}$ (as obtained in~\eqref{eq:tree-extend-char-2}) corresponds exactly to the reduced characteristic $g'_+(\delta', \beta')$ in~\eqref{eq:tree-reduced-extend-char-2}.
\end{proofclaim}
\fi

The computation of $\DPltree_{T}[\cdot, \cdot]$ is again analogous to the previous algorithm.
That is, we set
\begin{multline}\label{eq:DPT-reduced-extend}
	\DPltree_{T}[\delta,\beta] = 	\\ \min \left( \begin{gathered}
		\min_{(\alpha', \beta') \in {(g'_\mathsf{id})}^{-1}(\alpha, \beta)} \DPltree_{T'}[\alpha', \beta'],\\
		\min_{(\alpha', \beta') \in {(g'_+)}^{-1}(\alpha, \beta)} \DPltree_{T'}[\alpha', \beta'] + 1
	\end{gathered}   \right)
\end{multline}
where we take the minimum over empty set to be~$\infty$.

This finishes the description of the computation.
\ifshort
The correctness of the algorithm follows from  Claims~\ref{claim:tree-reduced-join-correctness} and~\ref{claim:tree-reduced-extend-correctness} by a bottom-up induction over the nice tree decomposition.
%We omit the full proof due to the space constraints.
\fi
\iflong
Now, we show the correctness of the algorithm in two separate claims.
First, we show that if there is a finite value stored in the dynamic programming table~$\DPltree_T[\cdot, \cdot]$ for a rooted tree $(T, r)$, there is a set $S \subseteq V(T) \setminus \{r\}$ with corresponding reduced characteristic and size. 

\begin{claim}\label{claim:tree-correct1}
Let $(T,r)$ be a rooted tree generated by a subexpression of the neat tree decomposition of~$G$ and let $(\delta, \beta)$ be arbitrary reduced characteristic. If $\DPltree_T[\delta, \beta] = s$ where $s \neq \infty$, then there exists a set $S \subseteq V(T) \setminus \{r\}$ of size~$s$ with reduced characteristic~$(\delta,\beta)$.
\end{claim}
\begin{proofclaim}
We prove the claim by a bottom-up induction on the neat tree decomposition of~$G$.

First, let $(T,r) = \mathsf{Leaf}(r)$ be a singleton rooted tree.
The only finite entry in the table $\DPltree_T[\cdot, \cdot]$ set in~\eqref{eq:DPT-reduced-leaf} is $\DPltree_T[\delta_\mathsf{leaf}, \beta_\mathsf{leaf}] = 0$.
As we already argued, the reduced characteristic of the empty set in~$(T,r)$ is exactly $(\delta_\mathsf{leaf}, \beta_\mathsf{leaf})$ and the claim holds for leaf nodes.

Now assume that $(T,r) = \mathsf{Join}((T_1, r),(T_2, r))$.
The value $\DPltree_T[\delta, \beta]$ was set to~$s$ in~\eqref{eq:DPT-reduced-join} and thus, there exist reduced characteristics~$(\delta_1, \beta_1)$ and $(\delta_2, \beta_2)$ such that $f((\delta_1, \beta_1), (\delta_2, \beta_2)) = (\delta, \beta)$ and $\DPltree_{T_1}[\delta_1, \beta_1] = s_1$, $\DPltree_{T_2}[\delta_2, \beta_2] = s_2$ with $s_1 + s_2 = s$.
Applying induction on $T_i$ for both $i \in [2]$, we see that there exists a set $S_i \subseteq V(T_i) \setminus \{r\}$ of size~$s_i$ and reduced characteristic $(\delta_i, \beta_i)$.
We conclude by \Cref{claim:tree-reduced-join-correctness} that $S_1 \cup S_2$ is a set of size~$s$ and reduced characteristic $(\delta, \beta)$ in~$H$.

Finally, suppose that $(T,r) = \mathsf{Extend}((T',r'),r)$.
The value $\DPltree_T[\delta, \beta]$ was set to be $s$ in~\eqref{eq:DPT-reduced-extend} and thus, there is a reduced characteristic~$(\delta', \beta')$ such that either
\begin{enumerate}
\item $g'_\mathsf{id}(\delta', \beta') = (\delta, \beta)$ and $\DPltree_{T'}[\delta', \beta'] = s$, or
\item $g'_+(\delta', \beta') = (\delta, \beta)$ and $\DPltree_{T'}[\delta', \beta'] = s - 1$.
\end{enumerate}
In the first case, there exists a set $S \subseteq V(T')  \setminus \{r'\}$ of size~$s$ and reduced characteristic~$(\delta', \beta')$ by induction on~$T'$.
The reduced characteristic of~$S$ in~$T$ is then precisely $(\delta, \beta)$ due to \Cref{claim:tree-reduced-extend-correctness}.
In the second case, there exists a set $S \subseteq V(T') \setminus \{r'\}$ of size~$s - 1$ and reduced characteristic~$(\delta', \beta')$ again by induction on~$T'$.
Therefore, the set~$S \cup \{r'\}$ has size exactly~$s$ and reduced characteristic~$(\delta, \beta)$ by \Cref{claim:tree-reduced-extend-correctness} in~$T$. 
\end{proofclaim}

Next, we show the opposite implication, that is, if there is a set~$S$ of a given reduced characteristic~$(\delta, \beta)$, then the computed value $\DPltree_T[\delta, \beta]$ is at most~$|S|$.

\begin{claim}\label{claim:tree-correct2}
Let $(T,r)$ be a rooted tree generated by a subexpression of the neat tree decomposition of~$G$ and let $(\delta, \beta)$ be an arbitrary reduced characteristic. If there exists a set $S \subseteq V(T)\setminus \{r\}$ of size~$s$ with reduced characteristic~$(\delta,\beta)$, then $\DPltree_T[\delta, \beta] \leq s$.
\end{claim}
\begin{proofclaim}
We prove the claim again by a bottom-up induction on the neat tree decomposition of~$G$.

First, let $(T,r) = \mathsf{Leaf}(r)$ be a singleton rooted tree.
The only possible choice for $S \subseteq V(T) \setminus \{r\}$ is the empty set.
We already know that the reduced characteristic of the empty set is~$(\delta_\mathsf{leaf}, \beta_\mathsf{leaf})$ and we have set $\DPltree_T[\delta_\mathsf{leaf}, \beta_\mathsf{leaf}] = 0$ in~\eqref{eq:DPT-reduced-leaf}.

Now assume that $(T,r) = \mathsf{Join}((T_1, r),(T_2, r))$ and for both $i \in [2]$, let $S_i$ be the restriction of $S$ to the vertices of~$T_i$ with a reduced characteristic $(\delta_i, \beta_i)$ in~$T_i$.
By applying induction on $T_i$ and $S_i$ for both $i \in [2]$, we see that $\DPltree_{T_i}[\delta_i, \beta_i] \le |S_i|$.
\Cref{claim:tree-reduced-join-correctness} implies that $f((\delta_1, \beta_1),(\delta_2,\beta_2)) = (\delta, \beta)$ and thus, we must have set $\DPltree_T[\delta, \beta] \le \DPltree_{T_1}[\delta_1, \beta_1] + \DPltree_{T_2}[\delta_2, \beta_2] \le |S_1| + |S_2| = s$ in~\eqref{eq:DPT-reduced-join}.

Finally, suppose that $(T,r) = \mathsf{Extend}((T',r'),r)$.
We consider separately two cases depending on the inclusion of $r'$ in $S$.
First suppose that $r' \notin S$ and let $(\delta', \beta')$ be the reduced characteristic of~$S$ in~$T'$.
By applying induction on~$T'$ and $S$, we see that $\DPltree_{T'}[\delta', \beta'] \le s$.
Moreover, we have $g'_\mathsf{id}(\delta', \beta') = (\delta, \beta)$ by \Cref{claim:tree-reduced-extend-correctness}.
It follows that we must have set $\DPltree_T[\delta, \beta] \le \DPltree_{T'}[\delta', \beta'] \le s$ in~\eqref{eq:DPT-reduced-extend}.
Now, suppose that $r' \in S$ and let $(\delta', \beta')$ be the reduced characteristic of~$S \setminus \{r'\}$ in~$T'$.
By applying induction on~$T'$ and $S\setminus\{r'\}$, we see that $\DPltree_{T'}[\delta', \beta'] \le s - 1$.
And since \Cref{claim:tree-reduced-extend-correctness} implies that $g'_+(\delta', \beta') = (\delta, \beta)$, we must have set $\DPltree_T[\delta, \beta] \le \DPltree_{T'}[\delta', \beta'] + 1 \le s$ in~\eqref{eq:DPT-reduced-extend}.
\end{proofclaim}
\fi

\ifshort
It remains to argue about the running time.
As already observed, the total number of possible reduced characteristics is~$12$.
It follows that a straightforward computation of the table~$\DPltree_T[\cdot, \cdot]$ in each node of the neat decomposition finishes in $O(1)$ time (refer to \eqref{eq:DPT-reduced-leaf}, \eqref{eq:DPT-reduced-join}, and~\eqref{eq:DPT-reduced-extend}).
Therefore, the overall running time is $O(n)$.
\fi

\iflong
It only remains to argue that the computation can be implemented to run in linear time.

\begin{claim}\label{claim:tree-reduced-runtime}
Given a tree~$G$ on input, the algorithm finishes in linear time.
\end{claim}
\begin{proofclaim}
The algorithm starts by computing a neat decomposition of the input tree in $O(n)$ time.
As already observed, the total number of possible reduced characteristics is~$12$.
It follows that the table $\DPltree_T[\cdot, \cdot]$ in each node of the neat decomposition is of constant size.
Moreover, a straightforward computation of a single entry also takes only constant time (refer to \eqref{eq:DPT-reduced-leaf}, \eqref{eq:DPT-reduced-join}, and~\eqref{eq:DPT-reduced-extend}).
It follows that the algorithm takes $O(1)$ time to compute the table $\DPltree_T[\cdot, \cdot]$ in every node, for a total of $O(n)$ time over the whole neat tree decomposition.
\end{proofclaim}
\fi

\section{Treewidth}

The algorithm for treewidth follows the same general scheme as the algorithms for trees in %\Cref{sec:trees}
the previous section: we define a suitable characteristic of any subset~$S$ of vertices such that (i) we can decide whether~$S$ is a feasible solution just from its characteristic, (ii) the way the characteristic of~$S$ changes in a node of a tree decomposition depends only on its previous characteristic, and (iii) the total number of characteristics is polynomial in the size of the input graph.
We then compute by a dynamic-programming scheme in every node of a nice tree decomposition the minimum size of a set~$S$ with each possible characteristic.
\iflong
Note that despite the high-level idea being similar, the details of the computation are quite intricate.
\fi

\begin{theorem} \label{thm:xp-tw}
	The \MUVCShort problem can be solved by an \XP-algorithm parameterized by the treewidth~$d$ of~$G$ in time $n^{O(2^d)}$.
\end{theorem}

\ifshort
\begin{sketch}
Due to space limitations, we only define the generalization of (reduced) characteristics. %and omit the description of their transformation in the nodes of a nice tree decomposition.
We first generalize rooted trees.
A \emph{terminal graph} is a pair $(G, X)$ where $G$ is a graph, $X \subseteq V(G)$ is a subset of its vertices, and $G[X]$ is an independent set, i.e., there are no edges between vertices of~$X$. \iflong 
Note that the last part of the definition is slightly nonstandard but we shall be working only with graphs having this property.\fi
For a terminal graph $(G,X)$, we say that a set $M \subseteq V(G)$ is of \emph{type $D$} if $M \cap X = D$.
Moreover, the \emph{reduced size} of a set $M \subseteq V(G)$ is defined as $|M \setminus X|$.

A pair of functions~$(\alpha, \beta)$ where $\alpha \colon 2^X \to \{0, \dots, n\}$ and $\beta \colon \{0,1\} \to \{1,2\}$ is a \emph{characteristic} of a set $S \subseteq V(G)\setminus X$ in a terminal graph~$(G,X)$ if for every $D \subseteq X$,
\begin{itemize}
	\item $\alpha(D)$ is the reduced size of the smallest vertex cover of type~$D$ in the terminal graph~$(G-S, X)$, and
	\item $\beta(D) = 1$ if and only if there is a unique vertex cover of reduced size~$\alpha(D)$ and type~$D$ in $(G - S, X)$.
\end{itemize}
%Notice that similarly to before we only consider sets~$S$ that do not contain any vertices of~$X$.
The characteristic is well defined: there always exists a vertex cover of any type~$D$, e.g., $V(G) \setminus X \cup D$.
Finally, remark that a rooted tree~$(T,r)$ can be interpreted as the terminal graph~$(T, \{r\})$ where the sets of type~0 and~1 in $(T, r)$ are exactly the sets of type~$\emptyset$ and~$\{r\}$ in $(T, \{r\})$ respectively.

We again define reduced characteristics.
A \emph{reduced characteristic} of a set $S \subseteq V(G)\setminus X$ with characteristic~$(\alpha,\beta)$ in a terminal graph~$(G,X)$ is a pair of functions~$(\delta, \beta)$ where $\delta \colon 2^X \to \{0, \dots, n\}$ such that $\delta(D) = \alpha(\emptyset) - \alpha(D)$ for every $D \subseteq X$.
%Note that $\alpha(\emptyset) - \alpha(D)$ is non-negative for every~$D$ because every vertex cover~$M$ of type~$\emptyset$ can be extended to a vertex cover~$M \cup D$ of type~$D$ and the same reduced size.

Let $x$ be a node of a nice tree decomposition~$(T, \{B_x \mid x\in V(T)\})$ of the graph~$G$.
Let $Y_x$ be the set of all vertices contained in the bags of the subtree rooted in~$t$.
We associate with the node~$x$ a terminal graph~$(G_x,B_x)$ such that $G_x$ is obtained from the graph~$G[Y_x]$ by removing all edges connecting two vertices of~$B_x$.
In other words, we consider an edge only when one of its endpoints is being forgotten.

The algorithm computes a nice tree decomposition of the graph~$G$.
Then it fills for each node~$x$ in the nice tree decomposition a dynamic programming table $\DPtw_x[\delta, \beta]$ such that $(\delta, \beta)$ is a reduced characteristic, and the entry $\DPtw_x[\delta, \beta]$ contains the size of the smallest set with reduced characteristic~$(\delta, \beta)$ in~$(G_x, B_x)$, or~$\infty$ if no such set exists.
\end{sketch}
\fi
\iflong
\begin{proof}
We first generalize rooted trees.
A \emph{terminal graph} is a pair $(G, X)$ where $G$ is a graph, $X \subseteq V(G)$ is a subset of its vertices, and $G[X]$ is an independent set, i.e., there are no edges between vertices of~$X$.
Note that the last part of the definition is slightly nonstandard but we shall be working only with graphs having this property.
For a terminal graph $(G,X)$, we say that a set $M \subseteq V(G)$ is of \emph{type $D$} if $M \cap X = D$.
Moreover, the \emph{reduced size} of a set $M \subseteq V(G)$ is defined as $|M \setminus X|$.
	
A pair of functions~$(\alpha, \beta)$ where $\alpha \colon 2^X \to [0, n]$ and $\beta \colon \{0,1\} \to [2]$ is a \emph{characteristic} of a set $S \subseteq V(G)\setminus X$ in a terminal graph~$(G,X)$ if for every $D \subseteq X$,
\begin{itemize}
	\item $\alpha(D)$ is the reduced size of the smallest vertex cover of type~$D$ in the terminal graph~$(G-S, X)$, and
	\item $\beta(D) = 1$ if and only if there is a unique vertex cover of reduced size~$\alpha(D)$ and type~$D$ in $(G - S, X)$.
\end{itemize}
Notice that similarly to before we only consider sets~$S$ that do not contain any vertices of~$X$.
Moreover, observe that the characteristic is well defined in the sense that there always exists a vertex cover of any type~$D$, e.g., $V(G) \setminus X \cup D$.
Finally, let us remark that a rooted tree~$(T,r)$ can be interpreted as the terminal graph~$(T, \{r\})$ where the sets of type~0 and~1 in $(T, r)$ are exactly the sets of type~$\emptyset$ and~$\{r\}$ in $(T, \{r\})$ respectively.
	
Since we do not care about the absolute sizes of the vertex covers, we again define reduced characteristics.
A \emph{reduced characteristic} of a set $S \subseteq V(G)\setminus X$ with characteristic~$(\alpha,\beta)$ in a terminal graph~$(G,X)$ is a pair of functions~$(\delta, \beta)$ where $\delta \colon 2^X \to [0,n]$ such that $\delta(D) = \alpha(\emptyset) - \alpha(D)$ for every $D \subseteq X$.
Note that $\alpha(\emptyset) - \alpha(D)$ is non-negative for every~$D$ because every vertex cover~$M$ of type~$\emptyset$ can be extended to a vertex cover~$M \cup D$ of type~$D$ and the same reduced size.
This can be seen as a direct generalization of \Cref{obs:tree-alpha-positive}.

Let $x$ be a node of a nice tree decomposition~$(T, \{B_x \mid x\in V(T)\})$ of the input graph~$G$.
Let $Y_x$ be the set of all vertices contained in the bags of the subtree rooted in~$t$.
We associate with the node~$x$ a terminal graph~$(G_x,B_x)$ such that $G_x$ is obtained from the graph~$G[Y_x]$ by removing all edges connecting two vertices of~$B_x$.
In other words, we consider an edge inside a bag only when one of its endpoints is being forgotten.

Now, we describe the high-level overview of the algorithm.
We first compute a nice tree decomposition of the input graph~$G$.
The algorithm proceeds by a dynamic programming along the nice tree decomposition.
Specifically for its each node~$x$, it stores a dynamic programming table $\DPtw_x[\delta, \beta]$ such that $(\delta, \beta)$ is a possible reduced characteristic, and the value of $\DPtw_x[\delta, \beta]$ contains the size of the smallest set with reduced characteristic~$(\delta, \beta)$ in~$(G_x, B_x)$, or~$\infty$ if no such set exists.

At this moment, it might seem pointless to focus on reduced characteristics since the full characteristics would result in an algorithm of similar efficiency and moreover, the computation would arguably be more straightforward.
However, we will show later that the number of reduced characteristics drops significantly in graphs of bounded degree and therefore, the same approach yields an \FPT-algorithm when we additionally parameterize by the maximum degree.
Note that we could alternatively first describe the computation for full characteristic and then derive the behavior of the reduced characteristics, similar to Subsections~\ref{subsec:trees_poly_algo} and~\ref{subsec:trees_linear_algo}.
We choose to present the dynamic programming for reduced characteristics straight away.

First, we observe that the reduced characteristic of a given set~$S$ in the terminal graph~$(G, \emptyset)$ allows to decide whether $S$ is a feasible solution to \MUVCShort.
\begin{observation}\label{obs:tw-final}
	A set~$S$ is a feasible solution to \MUVCShort in a graph~$G$ if and only if the reduced characteristic~$(\delta, \beta)$ of~$S$ in the terminal graph~$(G, \emptyset)$ satisfies $\beta(\emptyset) = 1$.
\end{observation}

For the root node~$r$ of the nice tree decomposition, we have $B_r = \emptyset$ and $G_r = G$ by definition.
So after computing the whole table $\DPtw_r[\cdot,\cdot]$ for the root node, the algorithm returns the smallest entry $\DPtw_r[\delta, \beta]$ such that $\beta(\emptyset) = 1$.

\paragraph{Leaf node.}
Let $x$ be a leaf node and thus, $B_x = \emptyset$.
There is only one choice for both~$S \subseteq V(G_x)$ and $D \subseteq B_x$, that is, the empty set. 
Therefore, the empty set has reduced characteristic~$(\delta_\emptyset, \beta_\emptyset)$ where $\delta_\emptyset(\emptyset) = 0$ and $\beta_\emptyset(\emptyset) = 1$.
We set
\begin{equation}\label{eq:tw-leaf}
	\DPtw_x[\delta, \beta] = \begin{cases}
		0 &\text{if $(\delta, \beta) = (\delta_{\emptyset}, \beta_{\emptyset})$, and}\\
		\infty &\text{otherwise.}
	\end{cases}
\end{equation}

\paragraph{Introduce node.}
Let $x$ be an introduce node introducing a vertex~$v$ and let~$y$ be its child.
The transformation of reduced characteristics in an introduce node is quite straightforward.
We define a function~$f$ acting on reduced characteristics such that $f(\delta_y, \beta_y) = (\delta_x, \beta_x)$ where for all $D \subseteq B_x$,
\begin{equation}
(\delta_x(D), \beta_x(D)) = (\delta_y(D \setminus \{v\}), \beta_y(D \setminus \{v\})).
\end{equation}

\begin{claim}\label{claim:tw-introduce-correctness}
	Let $S \subseteq V(G_y) \setminus B_y$ be an arbitrary set with reduced characteristic~$(\delta_y, \beta_y)$ in~$(G_y, B_y)$. Then $S$ has reduced characteristic $f(\delta_y, \beta_y)$ in~$(G_x, B_x)$.
\end{claim}
\begin{proofclaim}
Let $(\alpha_x, \beta_x)$ and $(\alpha_y, \beta_y)$ be the full characteristics of~$S$ in~$(G_x, B_x)$ and $(G_y, B_y)$ respectively.
We show that for every type $D \subseteq B_x$, we have $\alpha_x(D) = \alpha_y(D\setminus \{v\})$ and $\beta_x(D) = \beta_y(D\setminus \{v\})$.
The claim then follows since
\begin{multline*}
\delta_x(D) = \alpha_x(\emptyset) - \alpha_x(D) = \alpha_y(\emptyset \setminus \{v\}) - \alpha_y(D\setminus \{v\})\\ = \alpha_y(\emptyset) - \alpha_y(D\setminus \{v\}) =  \delta_y(D \setminus \{v\}).
\end{multline*}

We split the argument into two cases.
First, suppose that $v \in D$.
In this case, a set $M \subseteq V(G_y -S)$ is a vertex cover of type $D \setminus \{v\}$ in~$(G_y -S, B_y)$ if and only if $M \cup \{v\}$ is a vertex cover of type~$D$ in~$(G_x -S, B_x)$.
Moreover, the reduced size is preserved and thus, $\alpha_x(D) = \alpha_y(D\setminus \{v\})$ and $\beta_x(D) = \beta_y(D\setminus \{v\})$.

Now, assume that $v \notin D$.
If $M$ is a vertex cover of type~$D$ in~$(G_x-S, B_x)$, then $M$ is trivially also a vertex cover of type~$D$ in~$(G_y-S, B_y)$ because $G_y$ is a subgraph of~$G_x$.
On the other hand if $M$ is a vertex cover of type~$D$ in~$(G_y-S, B_y)$, then $M$ is still a vertex cover of type~$D$ in~$(G_x-S, B_x)$ because the vertex~$v$ is an isolated vertex in~$G_x$ and $E(G_x) = E(G_y)$.
It follows that $\alpha_x(D) = \alpha_y(D\setminus \{v\})$ and $\beta_x(D) = \beta_y(D\setminus \{v\})$ in this case as well.
\end{proofclaim}

The computation of $\DPtw_x[\cdot, \cdot]$ follows the now familiar scheme.
That is, we set
\begin{equation}\label{eq:tw-introduce}
	\DPtw_x[\delta,\beta] = 	\min_{(\delta_y, \beta_y) \in f^{-1}(\delta, \beta)} \DPtw_y[\delta_y, \beta_y]
\end{equation}
where we take the minimum over empty set to be~$\infty$.

\paragraph{Join node.}
Let $x$ be an join node and let~$y$ and $z$ be its children.
The way reduced characteristics are combined in join nodes remains basically the same as in the case of trees.
We define a function~$g$ acting on pairs reduced characteristics such that $g((\delta_y, \beta_y), (\delta_z, \beta_z)) = (\delta_x, \beta_x)$ where
\begin{align}\label{eq:tw-join-char}
	\begin{split}
		\delta_x(D) &=  \delta_y(D) + \delta_z(D) \text{, and}\\
		\beta_x(D) &= \min(2,\, \beta_y(D) \cdot \beta_z(D)).
	\end{split}
\end{align}

\begin{claim}\label{claim:tw-join-correctness}
Let $S_y \subseteq V(G_y) \setminus B_y$ be an arbitrary set of reduced characteristic~$(\delta_y, \beta_y)$ in~$(G_y, B_y)$ and let $S_z \subseteq V(G_z) \setminus B_z$ be an arbitrary set of reduced characteristic~$(\delta_z, \beta_z)$ in~$(G_z, B_z)$.
Then $S_y \cup S_z$ has reduced characteristic $g((\delta_y, \beta_y), (\delta_z, \beta_z))$ in~$(G_x, B_x)$.
\end{claim}
\begin{proofclaim}
Let $(\alpha_x, \beta_x)$ be the full characteristics of~$S_x \cup S_y$ in~$(G_x, B_x)$ and let $(\alpha_x, \beta_x)$ and $(\alpha_x, \beta_x)$ be the full characteristic of $S_y$ in~$(G_y, B_y)$ and $S_z$ in~$(G_z, B_z)$ respectively.
We show that for every type $D \subseteq B_x$, we have $\alpha_x(D) = \alpha_y(D) + \alpha_z(D)$ and $\beta_x(D) = \min(2,\, \beta_y(D) \cdot \beta_z(D))$.
The claim then follows since
\begin{multline*}
\delta_x(D) = \alpha_x(\emptyset) - \alpha_x(D) = \alpha_y(\emptyset) + \alpha_z(\emptyset) - \alpha_y(D) - \alpha_z(D)\\ = \alpha_y(\emptyset) - \alpha_y(D) =  \delta_y(D) + \delta_z(D).
\end{multline*}

The rest of the proof is analogous to the proof of \Cref{claim:tree-join-correctness}.
Let $M \subseteq V(G_x - (S_y \cup S_z))$ and set $M_y = M \cap V(G_y)$ and $M_z = M \cap V(G_z)$.
We observe that for any $D \subseteq B_x$, $M$ is a minimum vertex cover of type~$D$ in~$(G_x-S, B_x)$ if and only if $M_y$ is a minimum vertex cover of type~$D$ in~$(G_y-S_y, B_y)$ and $M_z$ is a minimum vertex cover of type~$D$ in~$(G_z-S_z, B_z)$ and hence, $\alpha_x(D) = \alpha_y(D) + \alpha_z(D)$.
It is again important here that we store the reduced sizes in characteristics because in terms of real sizes, we have $|M| = |M_y| + |M_z| - |D|$.
Finally, it follows that there is a unique minimum vertex cover of type~$D$ in~$(G_x -(S_y \cup S_z), B_x)$ if and only if there are unique minimum vertex covers of type~$D$ in both $(G_y -S_y, B_y)$ and $(G_z -S_z, B_z)$, that is, exactly if and only if $\beta_y(D) = \beta_z(D) = 1$.
\end{proofclaim}

The algorithm finds the minimum value $\DPtw_y[\delta_y, \beta_y] + \DPtw_z[\delta_z, \beta_z]$ over all preimages $(\delta_y, \beta_y)$ and $(\delta_z, \beta_z)$ under~$g$.
That is, we set
\begin{multline}\label{eq:tw-join}
	\DPtw_x[\delta, \beta] =\\ 	\min_{((\delta_x, \beta_x),(\delta_z,\beta_z)) \in g^{-1}(\delta, \beta)} \DPtw_y[\delta_y, \beta_y] + \DPtw_z[\delta_z, \beta_z]
\end{multline}
where we take the minimum over empty set to be~$\infty$.

\paragraph{Forget node.}
Let $x$ be a forget node introducing a vertex~$v$ and let~$y$ be its child.
As we already discussed, the computation in forget node is more complicated because we have to both (i) take care of the edges between~$v$ and the bag~$B_x$, and (ii) consider the situation when $v$ belongs to the set~$S$.
Thus, we define two functions acting on characteristics that describe how the reduced characteristic of a fixed set $S \subseteq V(G_y) \setminus B_y$ translates to the reduced characteristics of the set $S$ and $S \cup \{r'\}$ in~$(G_x, B_x)$.
In the first case, the characteristic also depends on the neighborhood of~$v$ within~$B_x$ and therefore, we define the transformation specifically for the vertex~$v$.
For any set $D \subseteq B_x$, let $D_+$ denote the set $D \cup \{v\}$.
We define~$h^v_{\mathsf{id}}$ such that $h^v_{\mathsf{id}}(\delta_y, \beta_y) = (\delta_x, \beta_x)$ such that
\begin{align}\label{eq:tw-forget-char-1}
\begin{split}
\delta_x(D) &=  \begin{cases}
	\delta_y(D) - \gamma &\begin{gathered}
	\text{if $N_G(v) \cap B_x \subseteq D$ and}\hfill\\\quad\text{$\delta_y(D_+) \le \delta_y(D) + 1$,}
	\end{gathered}\\
	\delta_y(D_+) - \gamma - 1 &\text{otherwise.}\\
\end{cases}\\
\beta_x(D) &= \begin{cases}
	\beta_y(D) &\begin{gathered}
		\text{if $N_G(v) \cap B_x \subseteq D$ and}\hfill\\\qquad\text{$\delta_y(D_+) < \delta_y(D) + 1$,}
	\end{gathered}\\
	\beta_y(D_+) &\begin{gathered}
		\text{if $N_G(v) \cap B_x \not\subseteq D$ or}\hfill\\\qquad\text{$\delta_y(D_+) > \delta_y(D) + 1$,}
	\end{gathered}\\
	2 &\text{otherwise.}
\end{cases}
\end{split}
\end{align}
where
\begin{equation}\label{eq:tw-forget-gamma} 
\gamma = \begin{cases}
	0 &\begin{gathered}
		\text{if $N_G(v) \cap B_x = \emptyset$ and}\\
		\hfill\text{$\delta_y(\{v\}) \le 1$, and}
	\end{gathered}\\
	\delta_y(\{v\}) - 1 &\text{otherwise.}
\end{cases}
\end{equation}
Now, we define~$h_{+}$ such that $h_{+}(\delta_y, \beta_y) = (\delta_x, \beta_x)$ where
\begin{align}\label{eq:tw-forget-char-2} 
		\begin{split}
			\delta_x(D) &= \delta_y(D \cup \{v\}) - \delta_y(\{v\}) \text{, and}\\
			\beta_x(D) &= \beta_y(D \cup \{v\}).
		\end{split}
\end{align}

\begin{claim}\label{claim:tw-forget-correctness}
	Let $S \subseteq V(G_y) \setminus B_y$ be an arbitrary set with reduced characteristic~$(\delta_y, \beta_y)$ in~$(G_y, B_y)$. Then the reduced characteristics of $S$ in~$(G_x, B_x)$ is $h^v_\mathsf{id}(\delta_y, \beta_y)$.
\end{claim}
\begin{proofclaim}
We first focus on the full characteristics.
To that end, let $(\alpha_x, \beta_x)$ and $(\alpha_y, \beta_y)$ be the full characteristic of~$S$ in~$(G_x, B_x)$ and in~$(G_y, B_y)$ respectively.
Fix a type~$D \subseteq B_x$ and let $D_+$ denote the set $D \cup \{v\}$ as before.
Our first goal is to show that
\begin{align}\label{eq:tw-forget-alpha}
\begin{split}
\alpha_x(D) &= \begin{cases}
\alpha_y(D_+) + 1 &\text{if $N_{G_x}(v) \cap B_x \not\subseteq D$,}\\
\min\left(\begin{gathered}
	\alpha_y(D),\\ \alpha_y(D_+) + 1
\end{gathered}\right) &\text{otherwise.}
\end{cases}\\
\beta_x(D) &= \begin{cases}
	\beta_y(D) &\begin{gathered}
		\text{if $N_{G_x}(v) \cap B_x \subseteq D$ and}\hfill\\\quad\text{$\alpha_y(D) < \alpha_y(D_+) + 1$,}
	\end{gathered}\\
	\beta_y(D_+) &\begin{gathered}
		\text{if $N_{G_x}(v) \cap B_x \not\subseteq D$ or}\hfill\\\quad\text{$\alpha_y(D) > \alpha_y(D_+) + 1$,}
	\end{gathered}\\
	2 &\text{otherwise.}
\end{cases}
\end{split}
\end{align}
We split the argument into the two cases depending on the neighborhood of~$v$ within~$B_x$.

First, suppose that $N_{G_x}(v) \not\subseteq D$, i.e., there exists a neighbor~$w$ of~$v$ in the bag~$B_x$ outside of~$D$.
In this case, every vertex cover of type~$D$ in~$(G_x -S, B_x)$ must contain the vertex~$v$ to cover the edge~$\{v,w\}$.
It follows that $M \subseteq V(G_x) \setminus S$ is a vertex cover of type~$D$ in~$(G_x -S, B_x)$ if and only if~$M$ is a vertex cover of type~$D_+$ in~$(G_y -S, B_y)$.
The reduced size of~$M$ increases when~$v$ leaves the bag and thus, we get $\alpha_x(D) = \alpha_y(D_+) + 1$.
Moreover, the uniqueness is carried over and $\beta_x(D) = \beta_y(D_+)$.
Thus, the expression~\eqref{eq:tw-forget-alpha} holds for this case.

Now, suppose that $N_{G_x}(v) \subseteq D$, i.e., all neighbors of~$v$ in the bag~$B_x$ are contained in~$D$.
In this case, a set $M \subseteq V(G_x) \setminus S$ is a vertex cover of type~$D$ in~$(G_x -S, B_x)$ if and only if $M$ is vertex cover of type~$D$ or~$D_+$ in~$(G_y -S, B_y)$.
It follows that a minimum vertex cover~$M$ of type~$D$ in $(G_x -S, B_x)$ has reduced size either~$\alpha_y(D)$ or $\alpha_y(D_+) + 1$ depending on which of the two is smaller.
Moreover, there is a unique minimum vertex cover of type~$D$ if and only if either (i) $\alpha_y(D) < \alpha_y(D_+) + 1$ and there is a unique minimum cover of type~$D$ in~$(G_y-S, B_y)$, or (ii) $\alpha_y(D_+) + 1 < \alpha_y(D)$ and there is a unique minimum cover of type~$D_+$ in~$(G_y-S, B_y)$.
This finishes proof of the expression~\eqref{eq:tw-forget-alpha}.

Our next step is to deduce the correctness of the computation in~$\eqref{eq:tw-forget-char-1}$ from~\eqref{eq:tw-forget-alpha}.
We start by showing that the conditions are equivalent.
By simple manipulation of terms, we get
\begin{align*}
\alpha_y(D) - \alpha_y(D_+) &= \alpha_y(D) - \alpha_y(D_+) + \alpha_y(\emptyset) - \alpha_y(\emptyset)\\
&= \delta_y(D_+) - \delta_y(D).
\end{align*}
It follows that $\alpha_y(D) < \alpha_y(D_+) + 1$ if and only if $\delta_y(D_+) < \delta_y(D) + 1$, and similarly $\alpha_y(D_+) + 1 < \alpha_y(D)$ if and only if $\delta_y(D_+) > \delta_y(D) + 1$.
This already shows that the computation of~$\beta_x$ in~\eqref{eq:tw-forget-char-1} is equivalent to~\eqref{eq:tw-forget-alpha}.

Next, we show that the value~$\gamma$ set in~\eqref{eq:tw-forget-gamma} is exactly equal to~$\alpha_y(\emptyset) - \alpha_x(\emptyset)$.
By~\eqref{eq:tw-forget-alpha}, we have $\alpha_x(\emptyset) = \alpha_y(\emptyset)$ if and only if~$N_{G_x}(v)\cap B_x  = \emptyset$ and $\alpha_y(\emptyset) \le \alpha_y(\{v\}) + 1$.
As we have seen, the condition~$\alpha_y(\emptyset) \le \alpha_y(\{v\}) + 1$ exactly equivalent to $\delta_y(\{v\}) \le \delta_y(\emptyset) + 1 = 1$ and we  indeed get $\gamma = 0 = \alpha_y(\emptyset) - \alpha_x(\emptyset)$ in this case.
Otherwise, we have $\alpha_x(\emptyset) = \alpha_y(\{v\}) + 1$ again by~\eqref{eq:tw-forget-alpha} and we obtain
\[\alpha_y(\emptyset) - \alpha_x(\emptyset) = \alpha_y(\emptyset) - \alpha_y(\{v\}) - 1 = \delta_y(\{v\}) - 1.\]

Plugging this into $\delta_y(D) - \gamma$ for arbitrary~$D$, we get
\begin{align*}
\delta_y(D) - \gamma &= \alpha_y(\emptyset) - \alpha_y(D) - \alpha_y(\emptyset) + \alpha_x(\emptyset)\\
&= \alpha_x(\emptyset) - \alpha_y(D).
\end{align*}
It follows that if $N_{G_x}(v) \cap B_x \subseteq D$ and $\alpha_y(D) \le \alpha_y(D_+) + 1$, then $\alpha_x(D) = \alpha_y(D)$ by~\eqref{eq:tw-forget-alpha} and thus, we get
\[\delta_y(D) - \gamma = \alpha_x(\emptyset) - \alpha_y(D) = \alpha_x(\emptyset) - \alpha_x(D) = \delta_x(D).\]
Otherwise, we have $\alpha_x(D) = \alpha_y(D_+) + 1$ by~\eqref{eq:tw-forget-alpha} and we see that
\begin{multline*}
\delta_y(D_+) - \gamma - 1 = \alpha_x(\emptyset) - \alpha_y(D_+) - 1 \\= \alpha_x(\emptyset) - \alpha_x(D) = \delta_x(D).
\end{multline*}

Either way, we see that both cases exactly match the computation of~$\delta_x$ in~\eqref{eq:tw-forget-char-1} which finishes the proof.
\end{proofclaim}

\begin{claim}\label{claim:tw-forget-correctness-2}
	Let $S \subseteq V(G_y) \setminus B_y$ be an arbitrary set with reduced characteristic~$(\delta_y, \beta_y)$ in~$(G_y, B_y)$. Then the reduced characteristics of $S \cup \{v\}$ in~$(G_x, B_x)$ is $h_+(\delta_y, \beta_y)$.
\end{claim}
\begin{proofclaim}
Let $S_+$ denote the set~$S \cup \{v\}$ and let $(\alpha_x, \beta_x)$ and $(\alpha_y, \beta_y)$ be the full characteristics of $S_+$ in~$(G_x,B_x)$ and $S$ in~$(G_y, B_y)$ respectively.

Observe that $M$ is a vertex cover in~$(G_x - S_+, B_x)$ if and only if $M \cup \{v\}$ is a vertex cover in~$(G_y - S, B_y)$.
Moreover, their reduced sizes are the same.
It follows that~$M$ is a minimum vertex cover of type~$D$ in $(G_x - S_+, B_x)$ if and only if $M \cup \{v\}$ is a minimum vertex cover in~$(G_y - S, B_y)$.
Therefore, we get that $\alpha_x(D) = \alpha_y(D_+)$ and $\beta_x(D) = \beta_y(D_+)$.
This already matches the computation of~$\beta_x$ in~\eqref{eq:tw-forget-char-2}.
Finally, wee see that	
\begin{multline*}
\delta_x(D) = \alpha_x(\emptyset) - \alpha_x(D)= \alpha_y(\{v\}) - \alpha_y(D_+)\\
 = \alpha_y(\{v\}) - \alpha_y(D_+) + \alpha_y(\emptyset) - \alpha_y(\emptyset)\\
 = \delta_y(D_+) - \delta_y(\{v\})
\end{multline*}
which again matches the computation in~~\eqref{eq:tw-forget-char-2}.
This finishes the proof of this claim.
\end{proofclaim}

The computation of $\DPtw_r[\cdot, \cdot]$ is analogous to the computation of $\DPltree[\cdot, \cdot]$ in extend nodes.
We set
\begin{multline}\label{eq:tw-forget}
	\DPtw_x[\delta,\beta] =\\ 	\min \left( \begin{gathered}
		\min_{(\delta_y, \beta_y) \in h^{-1}_\mathsf{id}(\delta, \beta)} \DPtw_y[\delta_y, \beta_y],\\
		\min_{(\delta_y, \beta_y) \in h^{-1}_+(\delta, \beta)} \DPtw_z[\delta_y, \beta_y] + 1
	\end{gathered}   \right)
\end{multline}
where we take the minimum over empty set to be~$\infty$.

We can finally prove the correctness of the algorithm, again split into two separate claims.

\begin{claim}\label{claim:tw-correct1}
Let $x$ be a node of the nice tree decomposition of~$G$ and let $(\delta, \beta)$ be an arbitrary reduced characteristic. If $\DPtw_x[\delta, \beta] = s$ where $s \neq \infty$, then there exists a set $S \subseteq V(G_x) \setminus B_x$ of size~$s$ with reduced characteristic~$(\delta,\beta)$.
\end{claim}
\begin{proofclaim}
We prove the claim by by a bottom-up induction on the nice tree decomposition of~$G$.
	
First, let $x$ be a leaf node.
The only finite entry in the table $\DPtw_x[\cdot, \cdot]$ set in~\eqref{eq:tw-leaf} is $\DPtw_T[\delta_\emptyset, \beta_\emptyset] = 0$.
As we already argued, the reduced characteristic of the empty set in~$(G_x,\emptyset)$ is exactly $(\delta_\emptyset, \beta_\emptyset)$ and the claim holds for leaf nodes.
	
Now assume that $x$ is an introduce node with a child~$y$.
The value $\DPtw_x[\delta, \beta]$ was set to~$s$ in~\eqref{eq:tw-introduce} and hence, there exists reduced characteristic~$(\delta_y, \beta_y)$ such that $f(\delta_y, \beta_y) = (\delta, \beta)$ and $\DPtw_y[\delta_y, \beta_y] = s$.
Applying induction on $y$ and~$(\delta_y, \beta_y)$, we see that there exists a set $S \subseteq V(G_y) \setminus B_y$ of size $s$ and reduced characteristic $(\delta_y, \beta_y)$.
We conclude by \Cref{claim:tw-introduce-correctness} that $S$ is a set of size~$s$ and reduced characteristic $(\delta, \beta)$ in~$(G_x, B_x)$.

Next, assume that $x$ is a join node with children~$y$ and~$z$.
The value $\DPtw_x[\delta, \beta]$ was set to~$s$ in~\eqref{eq:tw-join} and hence, there exists reduced characteristics~$(\delta_y, \beta_y)$ and~$(\delta_z, \beta_z)$ such that $g((\delta_y, \beta_y),(\delta_z, \beta_z)) = (\delta, \beta)$ and we have $\DPtw_y[\delta_y, \beta_y] = s_y$, $\DPtw_z[\delta_z, \beta_z] = s_z$ with $s_y + s_z  = s$.
Applying induction on both $y$ and~$z$, we see that there exist sets $S_y \subseteq V(G_y) \setminus B_y$ of size $s_y$ and reduced characteristic $(\delta_y, \beta_y)$ and $S_z \subseteq V(G_z) \setminus B_z$ of size $s_z$ and reduced characteristic $(\delta_z, \beta_z)$ respectively. 
\Cref{claim:tw-join-correctness} implies that $Sy \cup S_z$ is a set of size~$s$ and reduced characteristic $(\delta, \beta)$ in~$(G_x, B_x)$.

Finally, suppose that $x$ is a forget node forgetting a vertex~$v$, and lets $y$ its child.
The value $\DPtw_T[\delta, \beta]$ was set to be $s$ in~\eqref{eq:tw-forget} and therefore, there is a reduced characteristic~$(\delta_y, \beta_y)$ such that either
\begin{enumerate}
	\item $h_\mathsf{id}(\delta_y, \beta_y) = (\delta, \beta)$ and $\DPtw_y[\delta_y, \beta_y] = s$, or
	\item $h^v_+(\delta_y, \beta_y) = (\delta, \beta)$ and $\DPtw_y[\delta_y, \beta_y] = s - 1$.
\end{enumerate}
In the first case, there exists a set $S \subseteq V(G_y)  \setminus B_y$ of size~$s$ and reduced characteristic~$(\delta_y, \beta_y)$ by induction on~$y$.
The reduced characteristic of~$S$ in~$(G_x, B_x)$ is then precisely $(\delta, \beta)$ due to \Cref{claim:tw-forget-correctness}.
In the second case, there exists a set $S \subseteq V(G_y) \setminus B_y$ of size~$s - 1$ and reduced characteristic~$(\delta_y, \beta_y)$ again by induction on~$y$.
Therefore, the set~$S \cup \{v\}$ has size exactly~$s$ and reduced characteristic~$(\delta, \beta)$ by \Cref{claim:tw-forget-correctness} in~$(G_x, B_x)$. 
\end{proofclaim}

\begin{claim}\label{claim:tw-correct2}
Let $x$ be a node of the nice tree decomposition of~$G$ and let $(\delta, \beta)$ be an arbitrary reduced characteristic. If there exists a set $S \subseteq V(G_x)\setminus B_x$ of size~$s$ with reduced characteristic~$(\delta,\beta)$, then $\DPtw_x[\delta, \beta] \leq s$.
\end{claim}
\begin{proofclaim}
We again proceed by a bottom-up induction on the nice tree decomposition of~$G$.
	
First, let $x$ be a leaf node.
The only possible choice for~$S$ is the empty set.
We already argued that the reduced characteristic of the empty set is~$(\delta_\emptyset, \beta_\emptyset)$ and we have set $\DPtw_x[\delta_\emptyset, \beta_\emptyset] = 0$ in~\eqref{eq:tw-leaf}.

Now assume that $x$ is an introduce node with a child~$y$.
We have $S \subseteq V(G_y) \setminus B_y$ and let $(\delta_y, \beta_y)$ be its reduced characteristic in~$(G_y, B_y)$.
By applying induction on $y$ and~$S$, we see that $\DPtw_y[\delta_y, \beta_y] \le s$.
It follows by \Cref{claim:tw-introduce-correctness} that $f(\delta_y,\beta_y) = (\delta, \beta)$ and thus, we must have set $\DPtw_x[\delta, \beta] \le \DPtw_y[\delta_y, \beta_y] = s$ in~\eqref{eq:tw-introduce}.
	
Now assume that $x$ is a join node with children~$y$ and~$z$.
Let $S_y$ be the restriction of $S$ to the vertices of~$G_y$ with a reduced characteristic $(\delta_y, \beta_y)$ in~$(G_y, B_y)$ and similarly, let $S_z$ be the restriction of $S$ to the vertices of~$G_z$ with a reduced characteristic $(\delta_z, \beta_z)$ in~$(G_z, B_z)$.
By applying induction on both $y$ with $S_y$ and $z$ with~$S_z$, we see that $\DPtw_y[\delta_y, \beta_y] \le |S_y|$ and $\DPtw_z[\delta_z, \beta_z] \le |S_z|$.
\Cref{claim:tw-join-correctness} implies that $g((\delta_y, \beta_y),(\delta_z,\beta_z)) = (\delta, \beta)$ and thus, we must have set $\DPtw_x[\delta, \beta] \le \DPtw_y[\delta_y, \beta_y] + \DPtw_z[\delta_z, \beta_z] \le |S_y| + |S_z| = s$ in~\eqref{eq:tw-join}.
	
Finally, suppose that $s$ is a forget node forgetting a vertex~$v$, and let $y$ be its child.
First, assume that $v \notin S$ and let $(\delta_y, \beta_y)$ be the reduced characteristic of~$S$ in~$(G_y, B_y)$.
By applying induction on~$y$ and $S$, we see that $\DPtw_y[\delta_y, \beta_y] \le s$.
Moreover by \Cref{claim:tw-forget-correctness}, we have $h_\mathsf{id}(\delta_y, \beta_y) = (\delta, \beta)$.
It follows that we must have set $\DPtw_x[\delta, \beta] \le \DPtw_y[\delta_y, \beta_y] \le s$ in~\eqref{eq:tw-forget}.
Otherwise, we have $v \in S$ and let $(\delta_y, \beta_y)$ be the reduced characteristic of~$S \setminus \{v\}$ in~$(G_y, B_y)$.
By applying induction on~$y$ and $S\setminus\{v\}$, we see that $\DPtw_y[\delta_y, \beta_y] \le s - 1$.
Therefore, we must have set $\DPtw_x[\delta, \beta] \le \DPtw_y[\delta_y, \beta_y] + 1 \le s$ in~\eqref{eq:tw-forget} since $h^v_+(\delta_y, \beta_y) = (\delta, \beta)$ by \Cref{claim:tw-forget-correctness}.
\end{proofclaim}

It remains to argue about the time complexity.
The total number of reduced characteristic is $(n+1)^{2^w} \cdot 2^{2^w} = n^{O(2^w)}$.
It follows that the algorithm computes the table $\DPtw_x[\cdot, \cdot]$ in any node~$x$ in time $n^{O(2^w)}$ in all~\eqref{eq:tw-leaf}, \eqref{eq:tw-introduce}, \eqref{eq:tw-join} and~\eqref{eq:tw-forget}.
This makes the total running time over $O(w \cdot n)$ nodes of the nice tree decomposition still $n^{O(2^w)}$.
\end{proof}
\fi

Furthermore, we show that the same dynamic programming scheme yields an \FPT-algorithm when we additionally parameterize by the maximum degree of the input graph.

\begin{theorem} \label{thm:fpt-tw+Delta}
	The \MUVCShort problem can be solved by an \FPT-algorithm parameterized by the treewidth~$w$ of~$G$ plus the maximum degree~$\Delta$ in time $\Delta^{O(2^w)}\cdot n$.
\end{theorem}
\ifshort
\begin{sketch}
Crucially, in low degree graphs, the sizes of minimum vertex covers of different types cannot be too far away from each other.
Precisely, let $(G, X)$ be a terminal graph with maximum degree~$\Delta$ and let $S \subseteq V(G) \setminus X$ be a set of its vertices with reduced characteristic~$(\delta, \beta)$.
Then we have $0 \le \delta(D) \le \Delta \cdot |D|$ for every $D \subseteq X$.
It follows that the number of possible reduced characteristics in a terminal graph~$(G, X)$ is at most $(\Delta \cdot w)^{O(2^w)} = \Delta^{O(2^w)}$.
%This suffices to speed up the algorithm from \Cref{thm:xp-tw}.
\end{sketch}
\fi
\iflong
\begin{proof}
The crucial observation is that in low-degree graphs, the sizes of minimum vertex covers of different types cannot be too far away from each other.
	
\begin{claim}\label{claim:degree-tw}
	Let $(G, X)$ be a terminal graph and let $S \subseteq V(G) \setminus X$ be a set of its vertices with reduced characteristic~$(\delta, \beta)$.
	Then we have $0 \le \delta(D) \le \Delta \cdot |D|$ for every $D \subseteq X$ where $\Delta$ is the maximum degree of~$G$.
\end{claim}
\begin{proofclaim}
	Let $(\alpha, \beta)$ be the full characteristic of~$S$ in~$(G, X)$.
	We have already observed that $0 \le \alpha(\emptyset) - \alpha(D) = \delta(D)$.
	This holds since for a minimum vertex cover~$M$ of type~$\emptyset$, the set $M \cup D$ is a vertex cover of type~$D$ and the exact same reduced size.
	
	On the other hand, let $M$ be a minimum vertex cover of type~$D$ in~$(G, X)$, i.e., the reduced size of $M$ is $\alpha(D)$.
	Let $M' = M \setminus D \cup N_G(D)$ be the set obtained from $M$ by replacing vertices from~$D$ with all their neighbors.
	The set~$M'$ is a vertex cover of type~$\emptyset$ in~$(G, X)$ because we require~$X$ to be an independent set.
	Moreover, the reduced size of $M'$ is at most $|M| + \Delta \cdot |D|$ as every vertex in~$D$ has at most~$\Delta$ neighbors.
	It follows that $\alpha(\emptyset) \le \alpha(D) + \Delta \cdot |D|$ and the claim holds.
\end{proofclaim}

Due to \Cref{claim:degree-tw}, the total number of reduced characteristic in a terminal graph~$(G, X)$ is at most $(\Delta \cdot w)^{O(2^w)} = \Delta^{O(2^w)}$.
Therefore, the computation of every table $\DPtw_t[\cdot, \cdot]$ in \Cref{thm:xp-tw} finishes in~$\Delta^{O(2^w)}$ time.
The total running time is~$\Delta^{O(2^w)} \cdot n$ since the nice tree decomposition has size~$O(w \cdot n)$.
\end{proof}
\fi

\section{Clique-Width}

\begin{theorem} \label{thm:xp-cw}
	The \MUVCShort problem can be solved by an \XP-algorithm w.r.t the clique-width~$d$ of~$G$ in $n^{O(2^d)}$ time.
\end{theorem}

\ifshort
\begin{sketch}
The algorithm follows a scheme analogous to the algorithms in previous sections.
Due to the lack of space, we only describe what is the suitable definition of a characteristic of a set~$S$ in this case.

For a subset of vertices~$X$ of a $d$-labeled graph $G$, we denote by $\full_G(X)$ the set of labels from $[d]$ that are fully contained within~$X$\ifshort{}.\fi\iflong, i.e., $\full_G(X) = \{ i \in [d] \mid \lab_G^{-1}(i) \subseteq X\}$. \fi
%Observe that if there are no vertices in $G$ with some label~$i$ then $i \in \full_G(X)$ for every set $X$.
Conversely for $I \subseteq [d]$, we say that $T$ is a \emph{set of type $I$} if $\full_G(T) = I$.
Additionally for a type $I \subseteq [d]$ and a set~$T$ we say that $T$ \emph{extends type~$I$} if $I \subseteq \full_G(T)$.
	
A pair of functions $(\alpha, \beta)$ where $\alpha \colon 2^{[d]} \to \{0, \dots, n\}$ and $\beta \colon 2^{[d]} \to \{1,2\}$ is a \emph{characteristic} of a set $S \subseteq V(G)$ if for every $I \subseteq [d]$,
\begin{itemize}
	\item $\alpha(I)$ is the size of the smallest vertex cover that extends type~$I$ in~$G-S$, and
	\item $\beta(I) = 1$ if and only if the smallest vertex cover of size~$\alpha(I)$ that extends type~$I$ in $G - S$ is unique.
\end{itemize}
%A characteristic $(\alpha, \beta)$ is \emph{enforceable in~$G$} if there exists a set $S \subseteq V(G)$ with characteristic~$(\alpha, \beta)$.

The algorithm then proceeds by dynamic programming along a given clique-width $d$-expression~$\psi$ of~$G$.
Specifically, for each $d$-labeled graph~$H$ generated by a subexpression of~$\psi$, it stores a dynamic programming table $\DPcw_H[\alpha, \beta]$ such that $(\alpha, \beta)$ is a possible characteristic, and the value of $\DPcw_H[\alpha, \beta]$ contains the size of the smallest set with characteristic~$(\alpha, \beta)$ in~$H$, or~$\infty$ if no such set exists.
%	
%Observe that a set~$S$ with a characteristic~$(\alpha, \beta)$ is a feasible solution to \MUVCShort if and only if $\beta(\emptyset)=1$ because every vertex cover extends the empty type.
%Therefore after we compute $\DPcw_G[\cdot, \cdot]$ for the input graph, the algorithm simply returns the minimum value $\DPcw_G[\alpha, \beta]$ over all characteristics with $\beta(\emptyset) = 1$.
%	
%Now, we describe the computation separately for each of the operations allowed in a clique-width $d$-expression.
%A crucial part is to observe how the individual operations act on a characteristic of a fixed set.
\end{sketch}
\fi

\iflong
\begin{proof}
	The algorithm follows a scheme analogous to the algorithms in previous sections.
	Our first goal is, thus, to define a suitable definition of characteristics of sets.
	
	For a subset of vertices~$X$ of a $d$-labeled graph $G$, we denote by $\full_G(X)$ the set of labels from $[d]$ that are fully contained within~$X$,
	i.e., $\full_G(X) = \{ i \in [d] \mid \lab_G^{-1}(i) \subseteq X\}$.
	Observe that if there are no vertices in $G$ with some label~$i$ then $i \in \full_G(X)$ for every set $X$.
	Conversely for $I \subseteq [d]$, we say that $T$ is a \emph{set of type $I$} if $\full_G(T) = I$.
	Additionally for a type $I \subseteq [d]$ and a set~$T$ we say that $T$ \emph{extends type~$I$} if $I \subseteq \full_G(T)$.
	For any $d$-labeled graph~$H$ and $I \subseteq [d]$, let $\mu_H(I)$ be the size of a minimum vertex cover of type $I$ in~$H$ if it exists and $\infty$ otherwise.
	
	A pair of functions $(\alpha, \beta)$ where $\alpha \colon 2^{[d]} \to [0, n]$ and $\beta \colon 2^{[d]} \to [2]$ is a \emph{characteristic} of a set $S \subseteq V(G)$ if for every $I \subseteq [d]$,
	\begin{itemize}
		\item $\alpha(I) = \min\,\{\mu_{G - S}(J) \mid I \subseteq J \subseteq [d]\}$, i.e., the size of the smallest vertex cover that extends type~$I$ in~$G-S$, and
		\item $\beta(I) = 1$ if and only if there is a unique vertex cover of size~$\alpha(I)$ that extends type~$I$ in $G - S$.
	\end{itemize}
	A characteristic $(\alpha, \beta)$ is \emph{enforceable in~$G$} if there exists a set $S \subseteq V(G)$ with characteristic~$(\alpha, \beta)$.
	Let us observe a subtle difference from the characteristics we defined on trees and bounded-treewidth graphs that will play an important role later.
	In previous algorithms, all vertex covers of the graph~$G$ were partitioned into disjoint sets depending on their interaction with the root or bag respectively, and the function~$\alpha$ simply stored the size of a minimum vertex cover in each respective group.
	This is not true for the definition above, because a single vertex cover of type~$J$ might affect the value~$\alpha(I)$ for all types $I$ such that $I \subseteq J$.
		
	The algorithm proceeds by a dynamic programming along a given clique-width $d$-expression~$\psi$ of~$G$.
	Specifically, for each $d$-labeled graph~$H$ generated by a subexpression of~$\psi$, it stores a dynamic programming table $\DPcw_H[\alpha, \beta]$ such that $(\alpha, \beta)$ is a possible characteristic, and the value of $\DPcw_H[\alpha, \beta]$ contains the size of the smallest set of characteristic~$(\alpha, \beta)$ in~$H$, or~$\infty$ if no such set exists.
	
	Observe that a set~$S$ with a characteristic~$(\alpha, \beta)$ is a feasible solution to \MUVCShort if and only if $\beta(\emptyset)=1$ because every vertex cover extends the empty type.
	Therefore after we compute $\DPcw_G[\cdot, \cdot]$ for the input graph, the algorithm simply returns the minimum value $\DPcw_G[\alpha, \beta]$ over all characteristics with $\beta(\emptyset) = 1$.
	
	Now, we describe the computation separately for each of the operations allowed in a clique-width $d$-expression.
	%A crucial part is to observe how the individual operations act on a characteristic of a fixed set.

	\paragraph{Singleton $H = i(v)$.}
	There are only two possible choices of $S$ in the singleton graph $H$.
	These have the two characteristics $(\alpha^i_\mathsf{in}, \beta^i_\mathsf{in})$
	and  $(\alpha^i_\mathsf{out}, \beta^i_\mathsf{out})$ where $\beta^i_\mathsf{in}(I) = \beta^i_\mathsf{out}(I) = 1$ for every $I \subseteq [d]$ and $\alpha^i_\mathsf{in}, \alpha^i_\mathsf{out}$ are defined as follows
	\[\alpha^i_\mathsf{in}(I) = 0, \qquad
	\alpha^i_\mathsf{out}(I) = \begin{cases}
		1 &\text{if $i \in I$, and}\\
		0 &\text{otherwise.}
	\end{cases} \]
	The characteristic $(\alpha^i_\mathsf{in}, \beta^i_\mathsf{in})$ corresponds to the set $\{v\}$ while $(\alpha^i_\mathsf{out}, \beta^i_\mathsf{out})$ corresponds to the empty set. No other characteristic is enforceable in~$H$ and we set 
	\begin{equation}\label{eq:DPH-singleton}
		\DPcw_H[\alpha,\beta] = \begin{cases}
			1 &\text{if $(\alpha, \beta) = (\alpha^i_\mathsf{in}, \beta^i_\mathsf{in})$,}\\
			0 &\text{if $(\alpha, \beta) = (\alpha^i_\mathsf{out}, \beta^i_\mathsf{out})$, and}\\
			\infty &\text{otherwise.}
		\end{cases}
	\end{equation}
	
	\paragraph{Joining labels with edges, $H = \eta_{i,j}(H')$.}
	For each pair of pairwise different $i, j \in [d]$, we define a function $f_{i,j}$ acting on characteristics such that $f_{i,j}(\alpha', \beta') = (\alpha, \beta)$ where
	\begin{align*}
		\alpha(I) &= \begin{cases}
			\alpha'(I) &\text{if $I \cap \{i,j\} \neq \emptyset$}\\
			\min(\alpha'(I \cup \{i\}), \alpha'(I \cup \{j\})) &\text{otherwise}.
		\end{cases}\\
		\beta(I) &= \begin{cases}
			\beta'(I) &\text{if $I \cap \{i,j\} \neq \emptyset$, }\\
			\beta'(I \cup \{i\}) &\begin{gathered}
				\text{if $I \cap \{i,j\} = \emptyset$ and }\hfill\\\qquad\text{$\alpha'(I \cup \{i\}) < \alpha'(I \cup \{j\}),$}
			\end{gathered}\\
			\beta'(I \cup \{j\}) &\begin{gathered}
				\text{if $I \cap \{i,j\} = \emptyset$ and}\hfill\\\qquad \text{$\alpha'(I \cup \{i\}) > \alpha'(I \cup \{j\}),$}
			\end{gathered}\\
			1 &\begin{gathered}
				\text{if $I \cap \{i,j\} = \emptyset$, $\alpha'(I \cup \{i\}) =$}\\ \hfill\text{$\alpha'(I \cup \{j\}) = \alpha'(I \cup \{i,j\})$}\\\hfill\text{and $\beta'(I \cup \{i,j\}) = 1$, }
			\end{gathered}\\
			2 &\text{otherwise.}
		\end{cases}
	\end{align*}
	
	Let us show that the function $f_{i,j}$ describes precisely the effect of the operation~$\eta_{i,j}$ on a characteristic of a fixed set~$S$.
	
	\begin{claim}\label{claim:eta-correctness}
		Let $S \subseteq V(H')$ be an arbitrary set with characteristic~$(\alpha', \beta')$ in~$H'$ and let $(\alpha, \beta)$ be the image of $(\alpha', \beta')$ under $f_{i,j}$.
		Then $S$ has characteristic $(\alpha, \beta)$ in~$H$.
	\end{claim}
	\begin{proofclaim}
		Let $(\alpha'', \beta'')$ denote the characteristic of~$S$ in $H$.
		Our aim is to show that $(\alpha'', \beta'') = (\alpha, \beta)$.
		The labels do not change between $H'$ and $H$.
		Therefore a set $T \subseteq V(H')$ is a set of type~$I$ in~$H'-S$ if and only if it is a set of type~$I$ in~$H-S$.
		However, $H$ contains extra edges and thus, some vertex covers of~$H'-S$ are no longer vertex covers of~$H-S$.
		In particular, observe that any vertex cover of $H-S$ must contain all vertices with label~$i$ or all vertices with label~$j$.
		Note that this holds even when either one or both of the labels do not appear in $H-S$.
		
		First, let $I \subseteq [d]$ be a set such that $I \cap \{i,j\} \neq \emptyset$.
		We claim that a set $T \subseteq V(H')$ is a vertex cover of type~$I$ in $H'-S$ if and only if $T$ is a vertex cover of type~$I$ in $H-S$.
		On one hand, every vertex cover of type~$I$ in~$H-S$ is a vertex cover of type~$I$ in~$H'-S$ since $H-S$ is obtained by adding edges to $H'-S$.
		For the other direction, observe that every vertex cover~$T$ of type~$I$ in~$H'-S$ also covers the extra edges in $H-S$ because $T$ contains either all vertices with label $i$ or all vertices with label~$j$.
		Therefore, the set of vertex covers extending type~$I$ in $H-S$ is exactly the same as the set of vertex covers extending type~$I$ in $H'-S$ and we have $\alpha''(I) = \alpha'(I) = \alpha(I)$ and $\beta''(I)=\beta'(I)=\beta(I)$.
		
		Now, suppose that $I \subseteq [d]$ has an empty intersection with~$\{i,j\}$.
		First, we claim that a vertex cover~$T$ extends type~$I$ in~$H-S$ if and only if it extends type~$I\cup\{i\}$ or $I \cup \{j\}$ in~$H'-S$.
		This then immediately implies that $\alpha''(I) = \min(\alpha'(I \cup \{i\}), \alpha'(I \cup \{j\})) = \alpha(I)$.
		Let $T$ be a vertex cover of type $J \supseteq I$ in $H-S$.
		We know that it is a set of type~$J$ in $H'-S$.
		As we observed, we have necessarily $i \in J$ or $j \in J$.
		In the first case, we have $J \supseteq I \cup \{i\}$ and $T$ extends type $I \cup \{i\}$ in~$H'-S$ and analogously in the second case, $T$ extends type $I \cup \{j\}$ in~$H'-S$.
		For the other direction, notice that any vertex cover~$T$ extending type $I \cup \{i\}$ or $I \cup \{j\}$ in~$H'-S$ remains a vertex cover of~$H-S$ since it covers all the extra edges in~$H-S$.
		
		However, it remains to argue about the uniqueness.
		This is significantly more intricate because the set of vertex covers of~$H'-S$ extending $I \cup \{i\}$ is not disjoint from the set of vertex covers extending~$I \cup \{j\}$.
		In fact, their intersection contains exactly all the vertex covers extending the type~$I\cup\{i,j\}$.
		For $\ell \in \{i,j\}$, let $\mathcal{T}_\ell$ denote the set of all vertex covers of the size~$\alpha(I)$ extending~$I \cup \{\ell\}$ in~$H'-S$.
		%By definition, we have $\beta'(I \cup \{\ell\}) = 1$ if and only if $|\mathcal{T}_\ell| = 1$  for both $\ell \in \{i,j\}$.
		By definition, we have $\beta''(I) = 1$ if and only if $|\mathcal{T}_i \cup \mathcal{T}_j| = 1$.
		If $\alpha'(I \cup \{i\}) < \alpha'(I \cup \{j\})$, then the set $\mathcal{T}_j$ is empty and $\beta''(I) = \beta'(I \cup \{i\}) = \beta(I)$.
		Conversely if $\alpha'(I \cup \{i\}) > \alpha'(I \cup \{j\})$, then $\mathcal{T}_i$ is empty and $\beta''(I) = \beta'(I \cup \{j\}) = \beta(I)$.
		The most complicated situation occurs when $\alpha'(I \cup \{i\}) = \alpha'(I \cup \{j\})$ and both $\mathcal{T}_i$ and $\mathcal{T}_j$ are nonempty.
		In this case, we claim that $|\mathcal{T}_i \cup \mathcal{T}_j| = 1$ if and only if  $\alpha'(I \cup \{i,j\}) = \alpha'(I \cup \{i\}) = \alpha'(I \cup \{j\})$ and $\beta'(I \cup \{i,j\}) = 1$.
		It is then immediate that $\beta''(I) = \beta(I)$. 
		
		First suppose that $|\mathcal{T}_i \cup \mathcal{T}_j| = 1$ and let $T$ be the only vertex cover in $\mathcal{T}_i \cup \mathcal{T}_j$.
		Clearly, we have $\alpha'(I\cup\{i\}) = \alpha'(I\cup\{j\}) = |T|$ and moreover, also $ |T| = \alpha'(I\cup\{i,j\})$.
		Observe that any vertex cover extending~$I\cup \{i,j\}$ also extends both types $I \cup \{i\}$ and $I \cup \{j\}$.
		Consequently, $T$ must be a unique vertex cover of size~$\alpha(I)$ extending~$I \cup\{i,j\}$ in~$H'-S$ and $\beta'(I \cup \{i,j\}) = 1$.
		On the other hand if $\beta'(I \cup \{i,j\}) = 1$, then there exists a unique vertex cover~$T'$ of size~$\alpha'(I\cup\{i,j\})$ extending $I\cup\{i,j\}$ in~$H'-S$.
		But since we have $\alpha'(I \cup \{i,j\}) = \alpha'(I \cup \{i\}) = \alpha'(I \cup \{j\})$ it must also be the unique vertex cover extending types~$I \cup \{i\}$ and $I \cup \{j\}$.
		It follows that $\mathcal{T}_i = \mathcal{T}_j = \{T'\}$ and $|\mathcal{T}_i \cup \mathcal{T}_j| = 1$.
	\end{proofclaim}

	For a fixed characteristic~$(\alpha, \beta)$ the algorithm simply finds the preimage of $(\alpha, \beta)$ under~$f_{i,j}$ with the smallest value in~$\DPcw_{H'}[\cdot, \cdot]$.
	That is, we set
	\begin{equation}\label{eq:DPH-eta}
		\DPcw_H[\alpha, \beta] = \min_{(\alpha', \beta') \in f^{-1}_{i,j}(\alpha, \beta)} \DPcw_{H'}[\alpha', \beta']
	\end{equation}
	where we additionally define the minimum over empty set to be~$\infty$, i.e., we have $\DPcw_{H}[\alpha,\beta] = \infty$ whenever there is no preimage of $(\alpha, \beta)$ under~$f_{i,j}$.
	
	\paragraph{Relabeling, $H = \rho_{i\to j}(H')$.}
	For each pair of pairwise different $i, j \in [d]$, we define a function~$g_{i\to j}$ acting on characteristics such that $g_{i \to j}(\alpha', \beta') = (\alpha, \beta)$ where
	\begin{align*}
		(\alpha(I),\beta(I)) = \begin{cases}
			(\alpha'(I \cup \{i\}),        \beta'(I \cup \{i\}))       &\text{if $j \in I$,}\\
			(\alpha'(I \setminus \{i\}),   \beta'(I \setminus \{i\}))  &\text{otherwise.}
		\end{cases}
	\end{align*}
	
	Similar to before, we show that the function $g_{i \to j}$ describes exactly the effect of the operation~$\rho_{i \to j}$ on a characteristic of a fixed set~$S$.
	
	\begin{claim}\label{claim:rho-correctness}
		Let $S \subseteq V(H')$ be an arbitrary set with characteristic~$(\alpha', \beta')$ in~$H'$ and let $(\alpha, \beta)$ be the image of $(\alpha', \beta')$ under $g_{i\to j}$. Then $S$ has characteristic $(\alpha, \beta)$ in~$H$.
	\end{claim}
	\begin{proofclaim}
		Let $(\alpha'', \beta'')$ denote the characteristic of~$S$ in $H$.
		Our aim is to show that $(\alpha'', \beta'') = (\alpha, \beta)$.
		First, observe that $H'$ and $H$ differ only in their labels and therefore, any arbitrary set~$T \subseteq V(H)$ is a vertex cover of~$H'-S$ if and only if it is a vertex cover of~$H-S$.
		However, the type of~$T$ might differ between~$H'-S$ and~$H-S$.%, however only on the set~$\{i,j\}$ as the rest of the labels remain unchanged.
		
		Fix an arbitrary set~$T \subseteq V(H)$.
		First observe that for an index $\ell \in [d]\setminus \{i,j\}$, we have $\ell \in \full_{H-S}(T)$ if and only if $\ell \in \full_{H'-S}(T)$ because the vertices with label~$\ell$ remain unchanged.
		Moreover, we have $i \in \full_{H-S}(T)$ by definition since $H$ (and by extension~$H-S$) does not contain any vertices with label~$i$.
		Finally, we have $j \in \full_{H-S}(T)$ if and only if $\{i,j\} \subseteq \full_{H'-S}(T)$ since $\lab^{-1}_H(j) = \lab^{-1}_{H'}(i) \cup \lab^{-1}_{H'}(j)$.
		
		Now let $I \subseteq [d]$ be a type such that $j \in I$.
		Let $T$ be again a vertex cover of~$H-S$ (and thus also of~$H'-S$) of type $J$ in~$H-S$ and type~$J'$ in $H'-S$.
		We claim that $I \subseteq J$ if and only if $I \cup \{i\} \subseteq J'$.
		By previous arguments, we know that $J \setminus \{i,j\} = J' \setminus \{i,j\}$.
		If $I \subseteq J$, then $j \in J$ and it follows that  $\{i,j\} \subseteq J'$ and $I \cup \{i\} \subseteq J'$.
		For the other direction if $I \cup \{i\} \subseteq J'$, then $\{i,j\} \in J'$ and it follows that $j \in J$ and $I \subseteq J$.
		As a consequence, we see that $\alpha''(I) = \alpha'(I \cup \{i\}) = \alpha(I)$ since the set of  all vertex covers extending type~$I$ remains the same between~$H-S$ and~$H'-S$.
		The uniqueness is carried over for the same reason and we have $\beta''(I) = \beta'(I \cup \{i\}) = \beta(I)$.
		
		It remains to consider the case when $j \notin I$.
		Let $T$ be a vertex cover of~$H-S$ (and thus also of $H'-S$) of type $J$ in~$H-S$ and type~$J'$ in $H'-S$.
		We claim that $I \subseteq J$ if and only if $I \setminus \{i\} \subseteq J'$.
		The equalities of $\alpha''(I) = \alpha(I)$ and $\beta''(I) = \beta(I)$ then follow from this claim analogously to the previous case.
		%Recall again that $J \setminus \{i,j\} = J' \setminus \{i,j\}$.
		Assuming $I \subseteq J$, we get
		\[I \setminus \{i\} = I \setminus \{i,j\}  \subseteq J \setminus \{i,j\} = J'\setminus\{i,j\} \subseteq J'\]
		where the first equality holds since $j \notin I$, the first inclusion follows from the assumption $I \subseteq J$, the second equality was observed to hold in general and the final inclusion is trivial.
		For the other direction assume that $I \setminus \{i\} \subseteq J'$.
		We see that
		\begin{multline*}
			I \setminus \{i\} = I \setminus \{i,j\} \subseteq J'\setminus\{j\}\\ \subseteq J'\setminus\{i,j\} \cup \{i\} = J \setminus \{i,j\} \cup \{i\} \subseteq J \cup \{i\}
		\end{multline*}
		where the first equality holds again since $j \notin I$, the first inclusion is implied by $I \setminus \{i\} \subseteq J'$, the second inclusion is trivial, the second equality holds in general, and the final inclusion is trivial.
		Additionally, recall that $i \in J$ since $\lab^{-1}_H(i)$ is empty and $i$ belongs to the type of any set in~$H$.
		It follows that $I \subseteq J$ which finishes the proof.
		%We claim that any set~$T$ is of type~$I$ in~$H-S$ if and only if $T$ is a set of type  $J \in \Phi_{i,j}(I)$ in~$H'-S$.
		%On one hand, for a set~$T$ of type~$I$, $I \setminus \{i\}$ or $I \triangle \{i, j\}$ in~$H'-S$ there exists a vertex of label~$i$ or~$j$ not contained in~$T$ and thus, $T$ cannot contain all the vertices of label~$j$ after relabeling and $j \notin \full_{H-S}(T)$.
		%All the labels except $i$ and $j$ remain unchanged and the type of~$T$ in~$H-S$ is exactly~$I$.
		%For the other direction, let~$T$ be a set of type~$I$ in~$H-S$ and let $I'$ be its type in~$H'-S$.
		%As we observed, it must be that $\{i,j\} \not\subseteq I'$ and $I \setminus\{i,j\} = I'\setminus\{i,j\}$.
		%It follows that $I'$ is exactly one of the three sets~$I$, $I \setminus \{i\}$ or $I \triangle \{i, j\}$.
		%Thus, $\alpha''(I)$ is equal the minimum value~$\alpha'(J)$ for~$ J \in \Phi_{i,j}(I)$.
		%Furthemore, $\beta''(I) = 1$ if and only if that minimum is attained for a unique choice of~$J$ and there is a unique vertex cover of type~$J$ in $H'-S$, i.e., if $\beta'(J) = 1$.
		%It follows that $\alpha''(I) = \alpha(I)$ and $\beta''(I)=\beta(I)$ for such type~$I$.
		%
		%Finally, let $I \subseteq [d]$ be a set that contains both $i$ and~$j$.
		%By previous arguments, every set~$T$ of type~$I$ in~$H-S$ is a set of type~$I$ in $H'-S$ and vice versa.
		%Hence, a set $T$ is a vertex cover of type~$I$ in~$H-S$ if and only if it is a vertex cover of type~$I$ in~$H'-S$ and $\alpha''(I) = \alpha(I)$, $\beta''(I)=\beta(I)$ in this case as well.
	\end{proofclaim}
	
	Similar to the previous case, the algorithm simply finds the preimage of every characteristic $(\alpha, \beta)$ under~$g_{i\to j}$ with the smallest value in~$\DPcw_{H'}[\cdot, \cdot]$.
	That is, we set
	\begin{equation}\label{eq:DPH-rho}
		\DPcw_{H}[\alpha,\beta] = \min_{(\alpha', \beta') \in g^{-1}_{i\to j}(\alpha, \beta)} \DPcw_{H'}[\alpha', \beta']
	\end{equation}
	where the minimum over empty set is again~$\infty$.

	\paragraph{Disjoint union, $H = H_1 \oplus H_2$.}
	We define a function $h$ acting on pairs of characteristics $(\alpha_1, \beta_1), (\alpha_2, \beta_2)$ as $h((\alpha_1, \beta_1), (\alpha_2, \beta_2)) = (\alpha, \beta)$ where
	\begin{align*}
		\alpha(I) &=  \alpha_1(I) + \alpha_2(I) \text{, and}\\
		\beta(I) &= \min(2,\, \beta_1(I) \cdot \beta_2(I)).
	\end{align*}
	
	Let us show that when given the characteristics of two sets $S_1 \subseteq V(H_1)$ and $S_2 \subseteq V(H_2)$, the function~$h$ outputs exactly the characteristic of their union $S_1 \cup S_2$ in~$H$.
	
	\begin{claim}\label{claim:union-correctness}
		Let $S_1 \subseteq V(H_1)$ be an arbitrary set of characteristic~$(\alpha_1, \beta_1)$ in~$H_1$, let $S_2 \subseteq V(H_2)$ be an arbitrary set of characteristic~$(\alpha_2, \beta_2)$ in~$H_2$ and let $(\alpha, \beta)$ be the image of $(\alpha_1, \beta_1), (\alpha_2, \beta_2)$ under $h$. Then $S_1 \cup S_2$ has characteristic $(\alpha, \beta)$ in~$H$.
	\end{claim}
	\begin{proofclaim}
		Let $T \subseteq V(H-S)$ and set $T_1 = T \cap V(H_1)$ and $T_2 = T \cap V(H_2)$.
		Observe that $T$ is a vertex cover of~$H-S$ if and only if $T_1$ and $T_2$ are vertex covers of $H_1-S_1$ and $H_2-S_2$ respectively.
		Moreover, we have $I \subseteq \full_{H-S}(T)$ if and only if both $I \subseteq \full_{H_1-S_1}(T_1)$ and $I \subseteq \full_{H_2-S_2}(T_2)$.
		This holds since any label~$i \in [d]$ is fully covered in $H-S$ by~$T$ if and only if $T_1$ and $T_2$ contain all vertices with label~$i$ in $H_1-S_1$ and $H_2-S_2$ respectively.
		Therefore, every vertex cover extending type~$I$ in $H-S$ is obtained as a union of a vertex cover extending type~$I$ in~$H_1-S_1$ with a vertex cover extending type~$I$ in~$H_2-S_2$.
		In particular, the value~$\alpha''(I)$ is equal to the sum of the sizes of smallest vertex covers extending type~$I$ in~$H_1-S_1$ and in~$H_2-S_2$.
		Moreover, this smallest vertex cover is unique if and only if it is a combination of two unique such vertex covers in $H_1-S_1$ and $H_2-S_2$.
		That agrees exactly with how $\alpha(I)$ and $\beta(I)$ are computed from $\alpha_1(I)$, $\alpha_2(I)$, $\beta_1(I)$, and $\beta_2(I)$.
	\end{proofclaim}
	
	The computation of $\DPcw_{H}[\cdot, \cdot]$ is analogous to the previous two cases.
	That is, we set
	\begin{multline}\label{eq:DPH-union}
		\DPcw_{H}[\alpha,\beta] =\\ 	\min_{((\alpha_1, \beta_1),(\alpha_2,\beta_2)) \in h^{-1}(\alpha, \beta)} \DPcw_{H_1}[\alpha_1, \beta_1] + \DPcw_{H_2}[\alpha_2, \beta_2]
	\end{multline}
	where we take the minimum over empty set to be~$\infty$ as before.
	
	This finishes the description of the computation.
	Now, we show the correctness of the algorithm in two separate claims.
	First, we show that if there is a finite value stored in the dynamic programming table~$\DPcw_H[\cdot, \cdot]$, there is a set $S \subseteq V(H)$ with corresponding characteristic and size. 
	
	\begin{claim}\label{claim:cw-correct1}
		Let $H$ be a $d$-labeled graph corresponding to a subexpression of the clique-width $d$-expression of~$G$ and let $(\alpha, \beta)$ be arbitrary characteristic. If $\DPcw_H[\alpha, \beta] = s$ where $s \neq \infty$, then there exists a set $S \subseteq V(H)$ of size~$s$ with characteristic~$(\alpha,\beta)$.
	\end{claim}
	\begin{proofclaim}
		We prove the claim by a bottom-up induction on the clique-width $d$-expression of~$G$.
		First, let $H = i(v)$ be a singleton graph for some label~$i \in [d]$.
		There are only two finite entries in the table $\DPcw_H[\cdot, \cdot]$ in~\eqref{eq:DPH-singleton}, namely $\DPcw_H[\alpha^i_\mathsf{in}, \beta^i_\mathsf{in}]=1$ and~$\DPcw_H[\alpha^i_\mathsf{out}, \beta^i_\mathsf{out}] = 0$.
		There are also only two possible choices of $S \subseteq V(H)$.
		The only vertex cover of~$H-\{v\}$ is the empty set with type~$[d]$ that extends every type~$I$.
		It follows that the characteristic of~$S =\{v\}$ is exactly~$(\alpha^i_\mathsf{in}, \beta^i_\mathsf{in})$.
		On the other hand, there are two vertex covers in $H-\emptyset=H$, either the empty set or~$\{v\}$, and their types are~$[d]\setminus\{i\}$ and $[d]$ respectively.
		It follows that the smallest vertex cover extending a given type~$I$ is always unique and it has size~$1$ if $i \in I$ and size~$0$ otherwise.
		Therefore, the characteristic of the empty set in~$H$ is exactly~$(\alpha^i_\mathsf{out}, \beta^i_\mathsf{out})$ and the claim holds for singletons.
		
		Now assume that $H = \eta_{i,j}(H')$.
		The value $\DPcw_H[\alpha, \beta]$ was set to be $s$ in~\eqref{eq:DPH-eta} and hence, there exists a characteristic~$(\alpha', \beta')$ such that $f_{i,j}(\alpha', \beta') = (\alpha, \beta)$ and $\DPcw_{H'}[\alpha', \beta'] = s$.
		By applying induction on~$H'$, there exists a set $S \subseteq V(H')$ of size~$s$ and characteristic~$(\alpha', \beta')$.
		The characteristic of~$S$ in~$H$ is then precisely $(\alpha, \beta)$ due to \Cref{claim:eta-correctness} and we have found a set of the desired size and characteristic.
		
		The argument for $H = \rho_{i\to j}(H')$ is analogous to the previous case.
		This time, there is a set of characteristic~$(\alpha', \beta')$ in~$H'$ such that $g_{i \to j}(\alpha', \beta') = (\alpha, \beta)$ due to the computation in~\eqref{eq:DPH-rho}.
		The induction applied on~$H'$ together with~\Cref{claim:rho-correctness} then guarantees the existence of set with size~$s$ and characteristic~$(\alpha, \beta)$.
		
		Finally, suppose that $H = H_1 \oplus H_2$.
		Similar to before, $\DPcw_H[\alpha, \beta]$ was set to~$s$ in~\eqref{eq:DPH-union} and thus, there exist characteristics~$(\alpha_1, \beta_1)$ and $(\alpha_2, \beta_2)$ such that $h((\alpha_1, \beta_1), (\alpha_2, \beta_2)) = (\alpha, \beta)$ and $\DPcw_{H_1}[\alpha_1, \beta_1] = s_1$, $\DPcw_{H_2}[\alpha_2, \beta_2] = s_2$ with $s_1 + s_2 = s$.
		Applying induction on $H_\ell$ for both $\ell \in [2]$, we see that there exists a set $S_\ell \subseteq V(H_\ell)$ of size $s_\ell$ and characteristic $(\alpha_\ell, \beta_\ell)$.
		We conclude by \Cref{claim:union-correctness} that $S_1 \cup S_2$ is a set of size~$s$ and characteristic $(\alpha, \beta)$ in~$H$.
	\end{proofclaim}
	
	Next, we show the opposite implication, that is, if there is a set~$S$ of a given characteristic~$(\alpha, \beta)$, then the computed value $\DPcw_H[\alpha, \beta]$ is at most~$|S|$.
	
	\begin{claim}\label{claim:cw-correct2}
		Let $H$ be a $d$-labeled graph corresponding to a subexpression of the clique-width $d$-expression of~$G$ and let $(\alpha, \beta)$ be arbitrary characteristic. If there exists a set $S \subseteq V(H)$ of size~$s$ with characteristic~$(\alpha,\beta)$, then $\DPcw_H[\alpha, \beta] \leq s$.
	\end{claim}
	\begin{proofclaim}
		We prove the claim again by a bottom-up induction on the clique-width $d$-expression of~$G$.
		
		First, let $H = i(v)$ be a singleton graph for some label~$i \in [d]$.
		There are only two possible choices of~$S \subseteq V(H)$.
		We already showed in the proof of \Cref{claim:cw-correct1} that the characteristic of~$\{v\}$ is exactly~$(\alpha^i_\mathsf{in}, \beta^i_\mathsf{in})$ while the characteristic of~$\emptyset$ is exactly~$(\alpha^i_\mathsf{out}, \beta^i_\mathsf{out})$.
		For these, we have set $\DPcw_H[\alpha^i_\mathsf{in}, \beta^i_\mathsf{in}]=1$ and~$\DPcw_H[\alpha^i_\mathsf{out}, \beta^i_\mathsf{out}] = 0$ in~\eqref{eq:DPH-singleton} and it follows that $\DPcw_H[\alpha^i_\mathsf{in}, \beta^i_\mathsf{in}] = |\{v\}|$ and $\DPcw_H[\alpha^i_\mathsf{out}, \beta^i_\mathsf{out}] = |\emptyset|$.

		Now assume that $H = \eta_{i,j}(H')$ and let $(\alpha', \beta')$ be the characteristic of~$S$ in~$H'$.
		\Cref{claim:eta-correctness} implies that $f_{i,j}(\alpha', \beta') = (\alpha, \beta)$ and we see that $\DPcw_{H'}[\alpha', \beta'] \le s$ by applying induction on~$H'$ and~$S$.
		Thus, we can conclude that also $\DPcw_{H}[\alpha, \beta] \le s$ since $\DPcw_{H'}[\alpha', \beta']$ appears in the minimum on the right side of~\eqref{eq:DPH-eta}.
		
		The argument for $H = \rho_{i \to j}(H')$ is again almost identical.
		Let $(\alpha', \beta')$ be the characteristic of~$S$ in~$H'$.
		By a combination of induction and \Cref{claim:rho-correctness}, we see that the $\DPcw_{H}[\alpha, \beta] \le \DPcw_{H'}[\alpha', \beta'] \le s$.
		
		Finally, suppose that $H = H_1 \oplus H_2$ and for both $\ell \in [2]$, let $S_\ell$ be the restriction of $S$ to the vertices of~$H_\ell$ with characteristic $(\alpha_\ell, \beta_\ell)$ in~$H_\ell$.
		By applying induction on $H_\ell$ and~$S_\ell$ for both $\ell \in [2]$, we have $\DPcw_{H_\ell}[\alpha_\ell, \beta_\ell] \le |S_\ell|$.
		Moreover, \Cref{claim:union-correctness} implies that $h((\alpha_1, \beta_1),(\alpha_2, \beta_2)) = (\alpha, \beta)$.
		It follows that we have set $\DPcw_H[\alpha, \beta] \le \DPcw_{H_1}[\alpha_1, \beta_1] + \DPcw_{H_2}[\alpha_2, \beta_2] \le |S_1| + |S_2| = s$ in~\eqref{eq:DPH-union}.
	\end{proofclaim}
	
	It remains to argue about the runtime of the algorithm.
	First, we can compute a clique-width $(8^d-1)$-expression of a graph~$G$ of clique-width~$d$ in time $O(g(d) \cdot |V(G)|^3)$ for some function~$g$ using the algorithm by Oum~\cite{oum08}.
	Note that similar approximations of optimal clique-width expressions can be found much more efficiently in many specific graph classes of bounded clique-width.
	From now on, we assume that a clique-width $d$-expression of the graph~$G$ is given on input.
	
	\begin{claim}\label{claim:cw-runtime}
		Given a clique-width $d$-expression~$\phi$ of $G$ on input, the algorithm finishes in time $O(|\phi| \cdot 2^d \cdot C_d^2)$ where $C_d$ is the number of all possible characteristics in a $d$-labeled graph.
	\end{claim}
	\begin{proofclaim}
		We will show that for any $d$-labeled graph arising in~$\phi$, the table $\DPcw_H[\cdot, \cdot]$ can be filled in time~$O(2^d \cdot C_d^2)$.
		The claim follows since the number of such graphs is~$O(|\phi|)$.
		Also observe that the functions $f_{i,j}$, $g_{i \to j}$, and~$h$ can all be computed in time $O(2^d)$.
		
		When $H$ is a singleton, there are exactly two finite entries in $\DPcw_H[\cdot, \cdot]$ and we can fill the table in time $O(2^d \cdot C_d)$.
		Now, assume that $H = \eta_{i,j}(H')$.
		The algorithm simply computes $\DPcw_H[\alpha, \beta]$ by its definition in~\eqref{eq:DPH-eta}.
		That is, it enumerates over all possible characteristics and finds the minimum value of $\DPcw_{H'}[\alpha', \beta']$ over characteristic $(\alpha', \beta')$ such that $f_{i,j}(\alpha', \beta') = (\alpha, \beta)$.
		This enumeration takes $O(2^d)$ time per characteristic and the same is true for the computation of the function~$f_{i,j}$.
		Therefore, it takes $O(2^d \cdot C_d)$ time to fill a single entry~$\DPcw_H[\alpha, \beta]$ which adds up to $O(2^d \cdot C_d^2)$ time over the whole table.
		An analogous approach computes the table $\DPcw_H[\cdot, \cdot]$ when $H = \rho_{i \to j}(H')$.
		
		However, we can no longer use the same approach when $H = H_1 \oplus H_2$, as it would take $\Omega(2^d \cdot C^3_d)$ time.
		Instead, we speed up this computation using the same idea as before.
		We start by initially setting every entry to~$\infty$.
		Then we iterate over all possible pairs of characteristics $(\alpha_1, \beta_1)$, $(\alpha_2, \beta_2)$.
		For each pair, we first compute the value $h((\alpha_1, \beta_1), (\alpha_2, \beta_2))$, let us denote it $(\alpha, \beta)$.
		Afterwards, we update $\DPcw_H[\alpha, \beta]$ to $\DPcw_{H_1}[\alpha_1, \beta_1] + \DPcw_{H_2}[\alpha_2, \beta_2]$ but only if it is smaller than its current value.
		This takes $O(2^d)$ time per each pair of characteristics, for a total of $O(2^d \cdot C^2_d)$.
	\end{proofclaim}
	
	The total number of possible characteristics is $(n+1)^{2^d} \cdot 2^{2^d}$ since the domain of both $\alpha$ and $\beta$ is the set~$2^{[d]}$ and their ranges are $[0,n]$ and $[2]$ respectively.
	Plugging this into \Cref{claim:cw-runtime}, we see that the algorithm terminates in $n^{O(2^d)}$ time as promised.
\end{proof}
\fi

\begin{theorem} \label{thm:fpt-cw+k}
	The \MUVCShort problem can be solved by an \FPT-algorithm parameterized by the clique-width~$d$ of~$G$ plus the size of solution~$k$ in time $k^{O(2^d)}\cdot n$.
\end{theorem}

\ifshort
\begin{sketch}
	The result is obtained by truncating the dynamic programming table of the algorithm in \Cref{thm:xp-cw}.
	That is possible because the characteristic of a small set~$S$ cannot be too far away from the characteristic of the empty set.
	To be more precise, let $S \subseteq V(H)$  be a set of vertices in a $d$-labeled graph~$H$ with characteristic~$(\alpha_S, \beta_S)$ and let $(\alpha_\emptyset, \beta_\emptyset)$ be the characteristic of the empty set in~$H$. Then for arbitrary type $I \subseteq [d]$, we have  $0 \le \alpha_\emptyset(I) - \alpha_S(I) \le |S|$.
	
The algorithm proceeds in two passes over a clique-width $d$-expression~$\phi$ of the input graph~$G$.
In the first pass, it computes the characteristic of the empty set in every $d$-labeled graph generated by a subexpression of~$\phi$.
In the second pass, it follows the computation of the algorithm from \Cref{thm:xp-cw} restricted to characteristics with small distance to the characteristic of the empty set.
%By \Cref{claim:cw-k}, these still capture all solutions of size at most~$k$ and thus, the correctness of the algorithm carries over.
\end{sketch}
\fi

\iflong
\begin{proof}
	The result is obtained by truncating the dynamic programming table of the algorithm in \Cref{thm:xp-cw}.
	That is possible because the characteristic of a small set~$S$ cannot be too far away from the characteristic of the empty set.
	
	\begin{claim}\label{claim:cw-k}
		Let $H$ be a $d$-labeled graph, let $S \subseteq V(H)$ be a subset of its vertices with characteristic~$(\alpha_S, \beta_S)$ in~$H$ and let $(\alpha_\emptyset, \beta_\emptyset)$ be the characteristic of the empty set in~$H$. Then for arbitrary type $I \subseteq [d]$, we have  $0 \le \alpha_\emptyset(I) - \alpha_S(I) \le |S|$.
	\end{claim}
	\begin{proofclaim}
		Fix a type $I \subseteq [d]$.
		Let $T$ be the smallest vertex cover in~$H-S$ that extends type~$I$, i.e., the type of~$T$ in~$H-S$ is $J$ for some $J \supseteq I$.
		Clearly, $T \cup S$ is a vertex cover of~$H$ of size $|T| + |S| = \alpha_S(I) + |S|$.
		Moreover, the type of $T$ in~$H$ is still $J$ because $i \in \full_H(T \cup S)$ if and only if $i \in \full_{H-S}(T)$.
		We get
		\begin{multline*}
			\alpha_\emptyset(I) = \min\,\{\mu_{H}(I') \mid I \subseteq I' \subseteq [d]\}\\ \le \mu_{H}(J) \le |T| + |S| = \alpha_S(I) + |S|
		\end{multline*}
		where the first equality is the definition of $\alpha_\emptyset(I)$, the first inequality holds because $J \supseteq I$, the second inequality follows since $T \cup S$ is a vertex cover of type~$J$ in~$H$, and the final equality holds because $T$ is the smallest vertex cover extending~$I$ in~$H-S$.
		We obtain $\alpha_\emptyset(I) - \alpha_S(I) \le |S|$ by rearranging the inequality.
		
		On the other hand, let $T$ be the smallest vertex cover in~$H$ that extends type~$I$, i.e., the type of~$T$ in~$H$ is $J$ for some $J \supseteq I$.
		It is straightforward to see that $T \setminus S$ is a vertex cover of~$H-S$.
		Let $J'$ denote the type of $T \setminus S$ in~$H-S$.
		Observe that if we have $i \in \full_H(T)$ for some $i \in [d]$, then necessarily $i \in \full_{H-S}(T \setminus S)$.
		This implies $J' \supseteq J \supseteq I$.
		We get
		\begin{multline*}
			\alpha_S(I) = \min\,\{\mu_{H-S}(I') \mid I \subseteq I' \subseteq [d]\}\\ \le \mu_{H-S}(J')
			\le |T \setminus S| \le |T| = \alpha_\emptyset(I)
		\end{multline*}
		where the first equality is the definition of~$\alpha_S(I)$, the first inequality holds because $J' \supseteq I$, the second inequality holds since $T \setminus S$ is a vertex cover of type~$J'$ in~$H-S$, and the final equality holds by our choice of~$T$ as the smallest vertex cover extending~$I$ in~$H$.
		This wraps up the proof.
	\end{proofclaim}
	
	The algorithm proceeds in two passes over a clique-width $d$-expression~$\phi$ of the input graph~$G$.
	In the first pass, it computes the characteristic of the empty set in every $d$-labeled graph generated by a subexpression of~$\phi$.
	Let us denote $(\alpha^\emptyset_H, \beta^\emptyset_H)$ the characteristic of the empty set in such a $d$-labeled graph~$H$.
	This is done simply by setting the characteristic to $(\alpha^i_\mathsf{out}, \beta^i_\mathsf{out})$ whenever $H = i(v)$ and otherwise applying the functions $f_{i,j}$, $g_{i \to j}$, and $h$ to the characteristics of the empty set in subexpressions.
	The correctness of this computation is warranted by Claims~\ref{claim:eta-correctness}--\ref{claim:union-correctness}.
	Afterwards, we run the same dynamic programming algorithm as in \Cref{thm:xp-cw} but we restrict its computation within a graph~$H$ to characteristics~$(\alpha, \beta)$ such that $0 \le \alpha^\emptyset_H(I) - \alpha(I) \le k$.
	By \Cref{claim:cw-k}, these still capture all solutions of size at most~$k$ and thus, the correctness of the algorithm carries over.
	Moreover, the number of possible characteristics decreased to $(k+1)^{2^d} \cdot 2^{2^d}$ and thus the algorithm terminates in $k^{O(2^d)} \cdot n$ time by~\Cref{claim:cw-runtime}.
	Let us remark that for efficient implementation, we simply use as indices into the table the differences $\alpha^\emptyset_H(I) - \alpha(I)$ instead of the values $\alpha(I)$.
\end{proof}
\fi

\section{Hardness on Planar Graphs}

\begin{theorem}\label{thm:sigma_2P}
Both the \MUVCShort and \PAUVCShort problems are $\Sigma_2^P$-complete even when the input graph $G$ is planar and of maximum degree $5$.
\end{theorem}

\iflong
\begin{proof}
We first argue about \MUVCShort belonging in $\Sigma_2^p$. %Recall that a decision problem is in $\Sigma_2^P$ if and only if it can be decided by a non-deterministic Turing machine with the added use of an \NP-oracle. 
Given a graph $G=(V,E)$ and integer~$k$, assume we have guessed a set $S\subseteq V$ such that $|S|\leq k$ and $G'=G-S$ has a unique minimum vertex cover~$U$. Then, in order to verify that $U$ is indeed as required, it suffices to solve \PAUVCShort on $G'$ for $k=0$, which can be done in polynomial time with the help of an \NP-oracle~\cite{horiyama2024pauvc}. So, in what follows we focus on proving that \MUVCShort is $\Sigma_2^P$-hard for planar graphs of maximum degree $5$. Observe that slight modifications in our proof can lead to the same hardness result for the same family of graphs for \PAUVCShort.

We present a reduction from \UQPSAT~\cite{demaine18sigma2}. In that problem, we are given a 3CNF formula $\phi$ on the set of variables $\{x_1,\dots,x_{n_1},y_1,\dots,y_{n_2}\}$ and clauses $C=\{c_1,\dots,c_m\}$. We say that variables in $\{x_1,\dots,x_{n_1}\}$ ($\{y_1\dots,y_{n_2}\}$ resp.) are of \textit{type $x$} (\textit{type $y$} resp.). The task is to find a truth-assignment of the variables of type $x$ such that there exists a unique truth-assignment of the variables of type $y$ where each clause of $C$ is satisfied by exactly one of its literals (regardless of whether that literal contains a variable of type $x$ or $y$). We will construct a graph $G$ which has an \MUVCShort of order $n_1$ if and only if $\phi$ is a yes-instance of \UQPSAT.

\begin{figure}[!t]
\centering

\subfloat[The $c$-gadget. The three outer vertices correspond to the three literals of $c$.]{
\scalebox{0.6}{
\begin{tikzpicture}[inner sep=0.6mm]

	\node[draw, circle, line width=1pt, fill=white](x1) at (0,0)  [label=left: $\ell^1_c$]{};
    \node[draw, circle, line width=1pt, fill=white](x2) at (6,0)  [label=right: $\ell^2_c$]{};
    \node[draw, circle, line width=1pt, fill=white](x3) at (3,5)  [label=above: $\ell^3_c$]{};
    \node[draw, circle, line width=1pt, fill=white](x4) at (3,2.5)  [label=left: $w_c$]{};
    \node[draw, circle, line width=1pt, fill=white](x5) at (3,1.5)  [label=left: $z_c$]{};

	\draw[-, line width=1pt]  (x1) -- (x2);
    \draw[-, line width=1pt]  (x2) -- (x3);
    \draw[-, line width=1pt]  (x3) -- (x1);
    \draw[-, line width=1pt]  (x1) -- (x4);
    \draw[-, line width=1pt]  (x2) -- (x4);
    \draw[-, line width=1pt]  (x3) -- (x4);
    \draw[-, line width=1pt]  (x4) -- (x5);

\end{tikzpicture}
}
}\hspace{5pt}
\subfloat[The $y_i$-gadget. The vertices $y_i^2$ and $y_i^3$ are the inner colored vertices of the first appearance of the variable $y_i$.]{
\scalebox{0.6}{
\begin{tikzpicture}[inner sep=0.6mm]

    %outer
    \node[draw, circle, line width=1pt, fill=red](x1) at (2.75,0.00)  [label=right: $y_i^3$]{};
    \node[draw, circle, line width=1pt, fill=blue](x2) at (2.23,1.62)  [label=right: $y_i^4$]{};
    \node[draw, circle, line width=1pt, fill=red](x3) at (0.85,2.62)  [label=above: $y_i^5$]{};
    \node[draw, circle, line width=1pt, fill=blue](x4) at (-0.85,2.62)  [label=above: $y_i^6$]{};
    \node[draw, circle, line width=1pt, fill=red](x5) at (-2.23,1.62)  [label=left: $y_i^7$]{};
    %\node[draw, circle, line width=1pt, fill=blue](x6) at (-2.75,0.00)  []{};
    \node[draw, circle, line width=1pt, fill=red](x7) at (-2.23,-1.62)  []{};
    \node[draw, circle, line width=1pt, fill=blue](x8) at (-0.85,-2.62)  [label=below: $y_i^q$]{};
    \node[draw, circle, line width=1pt, fill=red](x9) at (0.85,-2.62)  [label=below: $y_i^1$]{};
    \node[draw, circle, line width=1pt, fill=blue](x10) at (2.23,-1.62)  [label=right: $y_i^2$]{};

    \draw[-, line width=1pt]  (x1) -- (x2);
    \draw[-, line width=1pt]  (x2) -- (x3);
    \draw[-, line width=1pt]  (x3) -- (x4);
    \draw[-, line width=1pt]  (x4) -- (x5);
    \draw[-, line width=1pt]  (x7) -- (x8);
    \draw[-, line width=1pt]  (x8) -- (x9);
    \draw[-, line width=1pt]  (x9) -- (x10);
    \draw[-, line width=1pt]  (x10) -- (x1);

    %middle
    \node[draw, circle, line width=1pt, fill=white](y1) at (1.90,0.61)  []{};
    \node[draw, circle, line width=1pt, fill=white](y2) at (1.23,1.64)  []{};
    \node[draw, circle, line width=1pt, fill=white](y3) at (0.00,2.00)  []{};
    \node[draw, circle, line width=1pt, fill=white](y4) at (-1.23,1.64)  [label=right: $w_i^6$]{};
    %\node[draw, circle, line width=1pt, fill=white](y5) at (-1.90,0.61)  []{};
    %\node[draw, circle, line width=1pt, fill=white](y6) at (-1.90,-0.61)  []{};
    \node[draw, circle, line width=1pt, fill=white](y7) at (-1.23,-1.64)  []{};
    \node[draw, circle, line width=1pt, fill=white](y8) at (0.00,-2.00)  [label=right: $w_i^q$]{};
    \node[draw, circle, line width=1pt, fill=white](y9) at (1.23,-1.64)  [label=above: $w_i^1$]{};
    \node[draw, circle, line width=1pt, fill=white](y10) at (1.90,-0.61)  []{};

    %inner
    \node[draw, circle, line width=1pt, fill=white](z1) at (1.19,0.38)  []{};
    \node[draw, circle, line width=1pt, fill=white](z2) at (0.77,1.02)  []{};
    \node[draw, circle, line width=1pt, fill=white](z3) at (0.00,1.25)  []{};
    \node[draw, circle, line width=1pt, fill=white](z4) at (-0.77,1.02)  [label=below: $z_i^6$]{};
    %\node[draw, circle, line width=1pt, fill=white](z5) at (-1.19,0.38)  []{};
    %\node[draw, circle, line width=1pt, fill=white](z6) at (-1.19,-0.38)  []{};
    \node[draw, circle, line width=1pt, fill=white](z7) at (-0.77,-1.02)  []{};
    \node[draw, circle, line width=1pt, fill=white](z8) at (0.00,-1.25)  [label=above: $z_i^q$]{};
    \node[draw, circle, line width=1pt, fill=white](z9) at (0.77,-1.02)  [label=above: $z_i^1$]{};
    \node[draw, circle, line width=1pt, fill=white](z10) at (1.19,-0.38)  []{};

    %inner edges
    \draw[-, line width=1pt]  (x1) -- (y1);
    \draw[-, line width=1pt]  (x2) -- (y1);
    \draw[-, line width=1pt]  (y1) -- (z1);
    
    \draw[-, line width=1pt]  (x2) -- (y2);
    \draw[-, line width=1pt]  (x3) -- (y2);
    \draw[-, line width=1pt]  (y2) -- (z2);
    
    \draw[-, line width=1pt]  (x3) -- (y3);
    \draw[-, line width=1pt]  (x4) -- (y3);
    \draw[-, line width=1pt]  (y3) -- (z3);

    \draw[-, line width=1pt]  (x4) -- (y4);
    \draw[-, line width=1pt]  (x5) -- (y4);
    \draw[-, line width=1pt]  (y4) -- (z4);

    \draw[-, line width=1pt]  (x7) -- (y7);
    \draw[-, line width=1pt]  (x8) -- (y7);
    \draw[-, line width=1pt]  (y7) -- (z7);

    \draw[-, line width=1pt]  (x8) -- (y8);
    \draw[-, line width=1pt]  (x9) -- (y8);
    \draw[-, line width=1pt]  (y8) -- (z8);

    \draw[-, line width=1pt]  (x9) -- (y9);
    \draw[-, line width=1pt]  (x10) -- (y9);
    \draw[-, line width=1pt]  (y9) -- (z9);

    \draw[-, line width=1pt]  (x10) -- (y10);
    \draw[-, line width=1pt]  (x1) -- (y10);
    \draw[-, line width=1pt]  (y10) -- (z10);

    \path (x5) -- (-2.75,0.00) node [black, midway, sloped] {\Large$\dots$};
    \path (-2.75,0.00) -- (-2.23,-1.62) node [black, midway, sloped] {\Large$\dots$};
    
\end{tikzpicture}
}
}
\subfloat[The $x_i$-gadget.]{
\scalebox{0.65}{
\begin{tikzpicture}[inner sep=0.6mm]

    %outer
    \node[draw, circle, line width=0.5pt, fill=blue](x1) at (2.75,0.00)  [label=right: $x_i^2$]{};
    \node[draw, circle, line width=0.5pt, fill=red](x2) at (2.23,1.62)  [label=right: $x_i^3$]{};
    \node[draw, circle, line width=0.5pt, fill=blue](x3) at (0.85,2.62)  [label=above right: $x_i^4$]{};
    \node[draw, circle, line width=0.5pt, fill=red](x4) at (-0.85,2.62)  [label=above left: $x_i^5$]{};
    \node[draw, circle, line width=0.5pt, fill=blue](x5) at (-2.23,1.62)  [label=left: $x_i^6$]{};
    %\node[draw, circle, line width=1pt, fill=blue](x6) at (-2.75,0.00)  []{};
    \node[draw, circle, line width=0.5pt, fill=blue](x7) at (-2.23,-1.62)  [label=left: $x_i^p$]{};
    \node[draw, circle, line width=2pt, fill=red](x8) at (-0.85,-2.62)  [label=left: $u_i^1$]{};
    \node[draw, circle, line width=2pt, fill=blue](x9) at (0.85,-2.62)  [label=right: $u_i^4$]{};
    \node[draw, circle, line width=0.5pt, fill=red](x10) at (2.23,-1.62)  [label=right: $x_i^1$]{};

    \draw[-, line width=0.5pt]  (x1) -- (x2);
    \draw[-, line width=0.5pt]  (x2) -- (x3);
    \draw[-, line width=0.5pt]  (x3) -- (x4);
    \draw[-, line width=0.5pt]  (x4) -- (x5);
    \draw[-, line width=2pt]  (x7) -- (x8);
    \draw[-, line width=2pt]  (x8) -- (x9);
    \draw[-, line width=2pt]  (x9) -- (x10);
    \draw[-, line width=0.5pt]  (x10) -- (x1);

    %middle
    \node[draw, circle, line width=0.5pt, fill=white](y1) at (1.90,0.61)  []{};
    \node[draw, circle, line width=0.5pt, fill=white](y2) at (1.23,1.64)  []{};
    \node[draw, circle, line width=0.5pt, fill=white](y3) at (0.00,2.00)  []{};
    \node[draw, circle, line width=0.5pt, fill=white](y4) at (-1.23,1.64)  [label=right: $w_i^5$]{};
    %\node[draw, circle, line width=1pt, fill=white](y5) at (-1.90,0.61)  []{};
    %\node[draw, circle, line width=1pt, fill=white](y6) at (-1.90,-0.61)  []{};
    \node[draw, circle, line width=0.5pt, fill=white](y7) at (-1.23,-1.64)  []{};
    \node[draw, circle, line width=2pt, fill=white](y8) at (0.00,-2.00)  [label=right: $w_i$]{};
    \node[draw, circle, line width=2pt, fill=white](y9) at (1.23,-1.64)  [label=above: $w_i^0$]{};
    \node[draw, circle, line width=0.5pt, fill=white](y10) at (1.90,-0.61)  []{};

    %inner
    \node[draw, circle, line width=0.5pt, fill=white](z1) at (1.19,0.38)  []{};
    \node[draw, circle, line width=0.5pt, fill=white](z2) at (0.77,1.02)  []{};
    \node[draw, circle, line width=0.5pt, fill=white](z3) at (0.00,1.25)  []{};
    \node[draw, circle, line width=0.5pt, fill=white](z4) at (-0.77,1.02)  [label=below: $z_i^5$]{};
    %\node[draw, circle, line width=1pt, fill=white](z5) at (-1.19,0.38)  []{};
    %\node[draw, circle, line width=1pt, fill=white](z6) at (-1.19,-0.38)  []{};
    \node[draw, circle, line width=0.5pt, fill=white](z7) at (-0.77,-1.02)  []{};
    \node[draw, circle, line width=2pt, fill=white](z8) at (0.00,-1.25)  [label=above: $z_i$]{};
    \node[draw, circle, line width=2pt, fill=white](z9) at (0.77,-1.02)  [label=above: $z_i^0$]{};
    \node[draw, circle, line width=0.5pt, fill=white](z10) at (1.19,-0.38)  []{};

    %inner edges
    \draw[-, line width=0.5pt]  (x1) -- (y1);
    \draw[-, line width=0.5pt]  (x2) -- (y1);
    \draw[-, line width=0.5pt]  (y1) -- (z1);
    
    \draw[-, line width=0.5pt]  (x2) -- (y2);
    \draw[-, line width=0.5pt]  (x3) -- (y2);
    \draw[-, line width=0.5pt]  (y2) -- (z2);
    
    \draw[-, line width=0.5pt]  (x3) -- (y3);
    \draw[-, line width=0.5pt]  (x4) -- (y3);
    \draw[-, line width=0.5pt]  (y3) -- (z3);

    \draw[-, line width=0.5pt]  (x4) -- (y4);
    \draw[-, line width=0.5pt]  (x5) -- (y4);
    \draw[-, line width=0.5pt]  (y4) -- (z4);

    \draw[-, line width=0.5pt]  (x7) -- (y7);
    \draw[-, line width=2pt]  (x8) -- (y7);
    \draw[-, line width=0.5pt]  (y7) -- (z7);

    \draw[-, line width=2pt]  (x8) -- (y8);
    \draw[-, line width=2pt]  (x9) -- (y8);
    \draw[-, line width=2pt]  (y8) -- (z8);

    \draw[-, line width=2pt]  (x9) -- (y9);
    \draw[-, line width=2pt]  (x10) -- (y9);
    \draw[-, line width=2pt]  (y9) -- (z9);

    \draw[-, line width=0.5pt]  (x10) -- (y10);
    \draw[-, line width=0.5pt]  (x1) -- (y10);
    \draw[-, line width=0.5pt]  (y10) -- (z10);

    \path (x5) -- (-2.75,0.00) node [black, midway, sloped] {\Large$\dots$};
    \path (-2.75,0.00) -- (-2.23,-1.62) node [black, midway, sloped] {\Large$\dots$};

    %assignment part
    \node[draw, circle, line width=2pt, fill=white](u2) at (-0.85,-3.62)  [label=left: $u_i^2$]{};
    \node[draw, circle, line width=2pt, fill=white](u3) at (-0.85,-4.62)  [label=left: $u_i^3$]{};
    \node[draw, circle, line width=2pt, fill=white](u5) at (0.85,-3.62)  [label=right: $u_i^5$]{};
    \node[draw, circle, line width=2pt, fill=white](u6) at (0.85,-4.62)  [label=right: $u_i^6$]{};

    \draw[-, line width=2pt]  (x8) -- (u2);
    \draw[-, line width=2pt]  (u2) -- (u3);
    \draw[-, line width=2pt]  (x9) -- (u5);
    \draw[-, line width=2pt]  (u5) -- (u6);
    
\end{tikzpicture}
}
}
\caption{The gadgets used in the proof of \Cref{thm:sigma_2P}.}\label{fig:hardness-gadgets}
\end{figure}
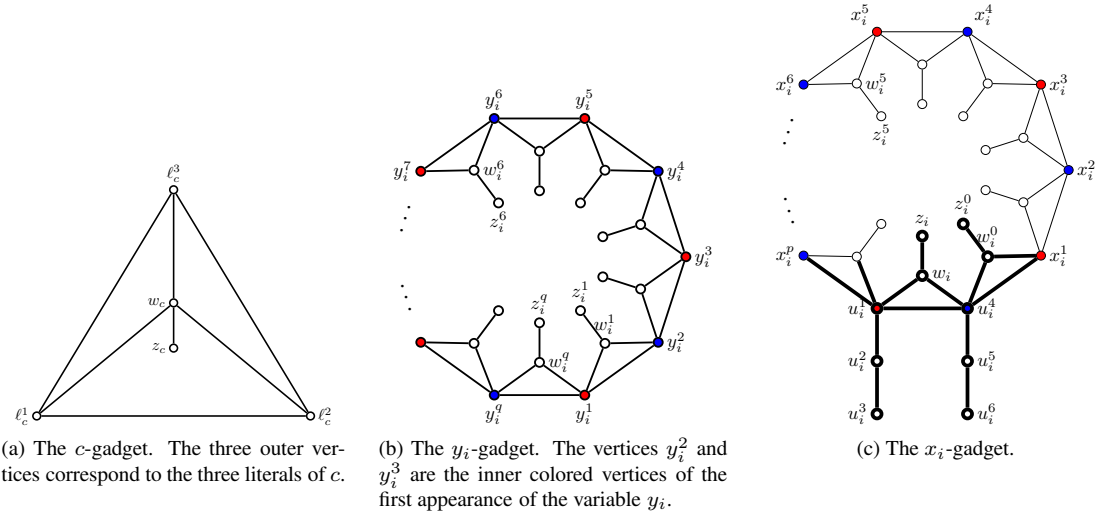

To construct the graph $G$, we first need to build an auxiliary graph $H$ as follows. First, we define a \textit{variable} (\textit{clause} resp.) vertex for each variable (clause resp.) in $\phi$. Then we add an edge between a variable and a clause vertex if the corresponding variable appears in the corresponding clause. Let $H$ be the resulting graph. Observe that $H$ is a planar graph (due to the ``planarity'' of $\phi$). We now start modifying $H$ to arrive to $G$. First, we replace each clause vertex $c$ of $H$ by a copy of the $c$-gadget illustrated in Figure~\ref{fig:hardness-gadgets}(a). Then, we replace each variable vertex by either a $y$-gadget or a $x$-gadget, illustrated in Figures~\ref{fig:hardness-gadgets}(b) and~\ref{fig:hardness-gadgets}(c) respectively, according to the type of the corresponding variable in $\phi$. Observe that the $x$- and $y$-gadgets are rather similar. Consider the $y_i$-gadget, i.e., the gadget that corresponds to the variable $y_i$ that appears in $\phi$. This gadget will have four colored vertices (see Figure~\ref{fig:hardness-gadgets}(b)) for each appearance of $y_i$ in $\phi$; two vertices colored red and two colored blue. That is, the index~$q$ that appears in Figure~\ref{fig:hardness-gadgets}(b) is equal to four times the number of appearances of~$y_i$ in $\phi$. For example, if $y_i$ appears in three clauses in a positive literal and in two clauses in a negative literal, then the $y_i$-gadget will have twenty colored vertices. Moreover, going in an anti-clockwise fashion, the colors of the vertices that correspond to each appearance of the $y_i$ variable will alternate, starting with a red and finishing with a blue; 
%\manolis{this is given an appearance of the variable. I suggest sth like: 
for the $j$-th appearance of variable $y_i$, we say that vertices $y^{4(j-1)+2}_i$ and $y^{4(j-1)+3}_i$ denote its \emph{inner colored vertices} 
%we denote as \textit{inner colored vertices} the first blue and second red vertex we meet going in the anti-clockwise fashion established earlier 
(see Figure~\ref{fig:hardness-gadgets}(b) for an example). Intuitively, the inner blue (red resp.) vertex included in the gadget for an appearance of the variable $y_i$ will be used to model that this variable is set to false (true resp.), while the other inner colored vertices will serve as points of additional connection between the gadgets. 
The same holds true for the $x_i$-gadget, i.e., the gadget that corresponds to the variable $x_i$ that appears in $\phi$. The only difference is that, in addition, the $x_i$-gadget contains an extra set of colored vertices together with two pending paths, illustrated by the bold vertices and edges in Figure~\ref{fig:hardness-gadgets}(c). 

At this stage, all the original edges of $H$ have been removed and we will now add the new edges between the gadgets. This edge-adding procedure happens in two steps. First we deal with the edges connecting the $c$-gadgets to the $x$- and/or $y$-gadgets.  
%We say that all the vertices of the $x$- and $y$-gadgets that are of degree $4$ at this stage, and the vertices $\ell_c^1,\ell_c^2$ and $\ell_c^3$ of the $c$-gadgets, are \textit{unmarked}. 
%Consider a clause $c$ and its corresponding $c$-gadget and assume that, in the initial graph $H$ there was an edge between the clause vertex $c$ and the variable vertex of type $x$. %, and $x$ appears in $c$ as a positive (negative resp.) literal. 
%Moreover, let $x$ be the $j$-th appearance of the $x$ variable in $\phi$ (according to a carefully chosen ordering of the variables %\manolis{ordering of clauses perhaps?}). 
%Then, going anti-clockwise, we locate the $j$-th quadruple $Q$ of colored vertices of the $x$-gadget. We then add an edge between any vertex of the $c$-gadget that is currently of degree three and the blue (red resp.) inner vertex of $Q$ if this appearance of $x$ is positive (negative resp.). 
%\manolis{
% Does it not suffice to go over the clause gadgets like $c_1, c_2,\ldots, c_m$? I think that any fixed ordering would do, no? I suggest to modify the construction with $j=1, \ldots, m$ going over the $c_j$ gadgets, and perhaps using another index, say $h$, to denote the $h$-th appearance of a variable in clauses.
%I think the previous sentence is slightly wrong and I propose the following:
Consider a clause $c$ and its corresponding $c$-gadget and assume that, in the initial graph $H$ there was an edge between the clause vertex $c$ and
the variable vertex $x_i$, which corresponds to a variable of type $x$.
Moreover, let this be the $j$-th appearance of the variable $x_i$ in $\phi$ (according to a carefully chosen ordering of the variables).
Then, going anti-clockwise, we locate the $j$-th quadruple $Q$ of colored vertices of the $x$-gadget, that is, $Q = \{ x^{4(j-1)+z}_i \mid z \in [4]\}$.
We then add an edge between any vertex among $\{\ell^1_c, \ell^2_c, \ell^3_c\}$ of the $c$-gadget that is currently of degree three and
the blue (red resp.) inner vertex of $Q$ if $x_i$ appears as a positive (negative resp.) literal in $c$.
%we add an edge between any unmarked vertex of the $c$-gadget and any blue unmarked vertex of the $x$-gadget; henceforth, the two vertices that were linked by this edge are considered as \textit{marked}. Moreover, let us assume that $x_i^j$ (for some $i\in[n_1]$ and $j\in[2,p]$) is the newly marked blue vertex. We also mark the red vertex $x_i^{j-1}$. If, instead, $x$ appears in $c$ as a negative literal, then we add an edge between any unmarked vertex of the $c$-gadget and any red unmarked vertex of the $x$-gadget; henceforth, the two vertices that were linked by this edge are considered as \textit{marked}. 
%Finally, let us assume that $x_i^j$ (for some $i\in[n_1]$ and $j\in[1,p-1]$) is the newly marked red vertex. We also mark the blue vertex $x_i^{j+1}$. 
We repeat the same procedure for all the edges of~$H$ that are between the clause vertex $c$ and any variable vertex of type $y$. Once we are done with the clause vertex $c$, we move on and repeat this procedure for every clause vertex of $H$. This completes the first step of adding the edges of $G$. 

In the second step, we connect the $x$- and/or $y$-gadgets whose corresponding variables appear in a common clause. To ease the exhibition, and since we treat these gadgets in the same way, we will assume we only have to deal with $x$-gadgets. So, consider a clause gadget $c$, with the corresponding clause being comprised of three literals on the variables $x_{i_1}$, $x_{i_2}$, and $x_{i_3}$ (for some $i_1,i_2,i_3 \in [n_1])$.
According to the construction of $G$ up to this point, there are 
\begin{itemize}
    \item a $c$-gadget, containing the vertices $\ell_c^1$, $\ell_c^2$, and $\ell_c^3$, and
    \item the $x_{i_1}$, $x_{i_2}$, and $x_{i_3}$-gadgets, containing some inner colored vertices $x_{i_1}^\alpha$, $x_{i_2}^\beta$, and $x_{i_3}^\gamma$ respectively
    such that $G$ contains the edges $\ell_c^1 x_{i_1}^\alpha, \ell_c^2 x_{i_2}^\beta$, and $\ell_c^3 x_{i_3}^\gamma$.
\end{itemize}
Note that since $x_{i_1}^\alpha$, $x_{i_2}^\beta$, and $x_{i_3}^\gamma$ are inner colored vertices, and according to the first step of the edge-adding procedure, the two colored neighbors of these vertices that lie in the $x_{i_1}$, $x_{i_2}$, and $x_{i_3}$-gadgets respectively are currently of degree $4$. The second step of the edge-adding procedure consists in adding the edges $x_{i_3}^{\gamma-1} x_{i_1}^{\alpha+1}$, $x_{i_1}^{\alpha-1} x_{i_2}^{\beta+1}$, and $x_{i_2}^{\beta-1} x_{i_3}^{\gamma+1}$. This step is repeated for every clause gadget $c$.  

The termination of the edge-adding procedure marks the end of the construction of $G$.
Observe that by carefully choosing the ordering used in the first step and bending the edges added in the second step of the edge-adding procedure, and due to the planarity of $H$, we can also ensure the planarity of~$G$.
On a high level, the ordering is chosen based on the planar embedding of graph $H$, while the edges added in the second step between variable gadgets
can be stretched to follow along the path dictated by the edge of each variable gadget towards their common clause gadget.
\Cref{fig:hardness-example} illustrates an example of the above construction. 

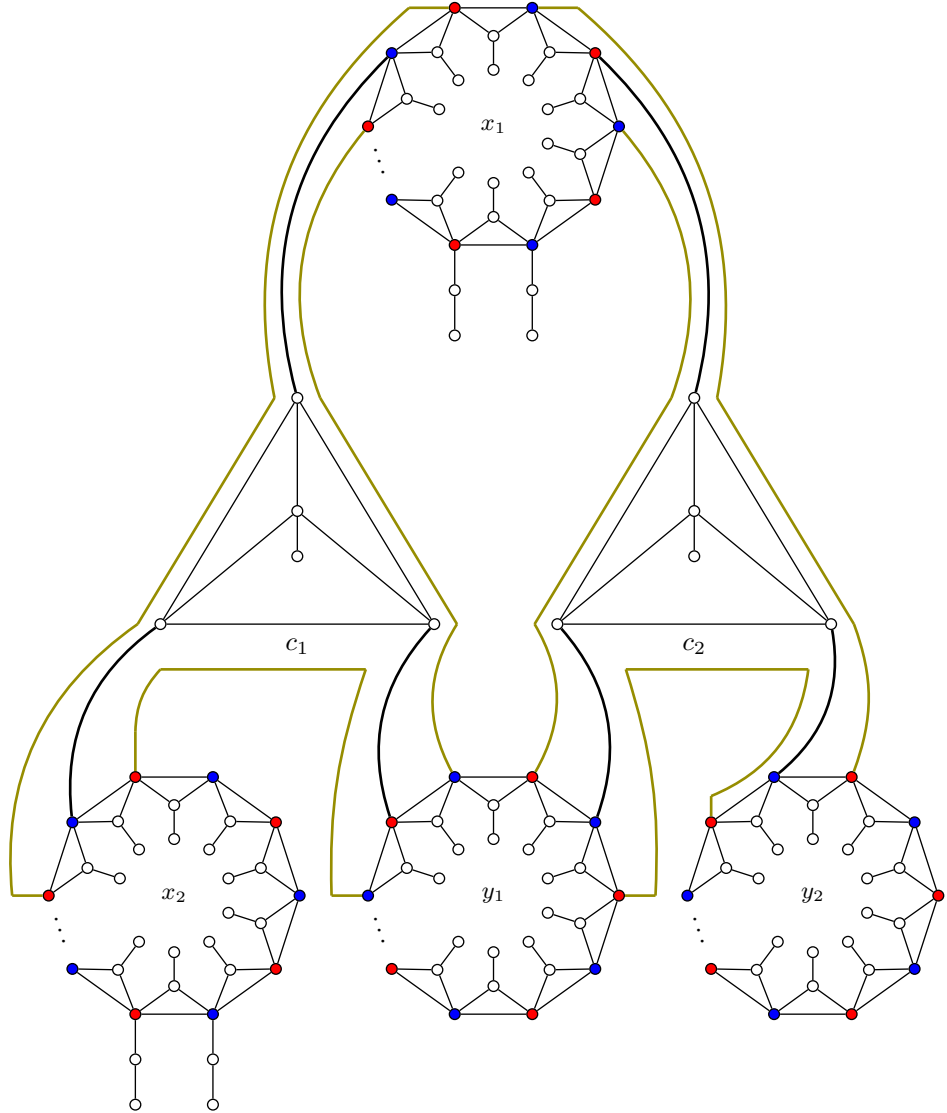
\begin{figure}[t!]

\centering
\begin{tikzpicture}[scale=0.4, inner sep=0.5mm]

%x2:
\begin{scope}
%outer
    \node[draw, circle, line width=0.5pt, fill=blue](x21) at (2.75,0.00)  []{};
    \node[draw, circle, line width=0.5pt, fill=red](x22) at (2.23,1.62)  []{};
    \node[draw, circle, line width=0.5pt, fill=blue](x23) at (0.85,2.62)  []{};
    \node[draw, circle, line width=0.5pt, fill=red](x24) at (-0.85,2.62)  []{};
    \node[draw, circle, line width=0.5pt, fill=blue](x25) at (-2.23,1.62)  []{};
    \node[draw, circle, line width=0.5pt, fill=red](x26) at (-2.75,0.00)  []{};
    \node[draw, circle, line width=0.5pt, fill=blue](x27) at (-2.23,-1.62)  []{};
    \node[draw, circle, line width=0.5pt, fill=red](x28) at (-0.85,-2.62)  []{};
    \node[draw, circle, line width=0.5pt, fill=blue](x29) at (0.85,-2.62)  []{};
    \node[draw, circle, line width=0.5pt, fill=red](x210) at (2.23,-1.62)  []{};

    %control points
    \coordinate (contx21) at (-3.55,0.00);
    \coordinate (contx22) at (-0.85,3.62);

    \draw[-, line width=0.5pt]  (x21) -- (x22);
    \draw[-, line width=0.5pt]  (x22) -- (x23);
    \draw[-, line width=0.5pt]  (x23) -- (x24);
    \draw[-, line width=0.5pt]  (x24) -- (x25);
    \draw[-, line width=0.5pt]  (x25) -- (x26);
    \draw[-, line width=0.5pt]  (x27) -- (x28);
    \draw[-, line width=0.5pt]  (x28) -- (x29);
    \draw[-, line width=0.5pt]  (x29) -- (x210);
    \draw[-, line width=0.5pt]  (x210) -- (x21);

    %middle
    \node[draw, circle, line width=0.5pt, fill=white](y21) at (1.90,0.61)  []{};
    \node[draw, circle, line width=0.5pt, fill=white](y22) at (1.23,1.64)  []{};
    \node[draw, circle, line width=0.5pt, fill=white](y23) at (0.00,2.00)  []{};
    \node[draw, circle, line width=0.5pt, fill=white](y24) at (-1.23,1.64)  []{};
    \node[draw, circle, line width=0.5pt, fill=white](y25) at (-1.90,0.61)  []{};
    %\node[draw, circle, line width=0.5pt, fill=white](y6) at (-1.90,-0.61)  []{};
    \node[draw, circle, line width=0.5pt, fill=white](y27) at (-1.23,-1.64)  []{};
    \node[draw, circle, line width=0.5pt, fill=white](y28) at (0.00,-2.00)  []{};
    \node[draw, circle, line width=0.5pt, fill=white](y29) at (1.23,-1.64)  []{};
    \node[draw, circle, line width=0.5pt, fill=white](y210) at (1.90,-0.61)  []{};

    %inner
    \node[draw, circle, line width=0.5pt, fill=white](z21) at (1.19,0.38)  []{};
    \node[draw, circle, line width=0.5pt, fill=white](z22) at (0.77,1.02)  []{};
    \node[draw, circle, line width=0.5pt, fill=white](z23) at (0.00,1.25)  []{};
    \node[draw, circle, line width=0.5pt, fill=white](z24) at (-0.77,1.02)  []{};
    \node[draw, circle, line width=0.5pt, fill=white](z25) at (-1.19,0.38)  []{};
    %\node[draw, circle, line width=0.5pt, fill=white](z6) at (-1.19,-0.38)  []{};
    \node[draw, circle, line width=0.5pt, fill=white](z27) at (-0.77,-1.02)  []{};
    \node[draw, circle, line width=0.5pt, fill=white](z28) at (0.00,-1.25)  []{};
    \node[draw, circle, line width=0.5pt, fill=white](z29) at (0.77,-1.02)  []{};
    \node[draw, circle, line width=0.5pt, fill=white](z210) at (1.19,-0.38)  []{};

    \node[] () at (0,0) []{\small$x_2$};

    %inner edges
    \draw[-, line width=0.5pt]  (x21) -- (y21);
    \draw[-, line width=0.5pt]  (x22) -- (y21);
    \draw[-, line width=0.5pt]  (y21) -- (z21);
    
    \draw[-, line width=0.5pt]  (x22) -- (y22);
    \draw[-, line width=0.5pt]  (x23) -- (y22);
    \draw[-, line width=0.5pt]  (y22) -- (z22);
    
    \draw[-, line width=0.5pt]  (x23) -- (y23);
    \draw[-, line width=0.5pt]  (x24) -- (y23);
    \draw[-, line width=0.5pt]  (y23) -- (z23);

    \draw[-, line width=0.5pt]  (x24) -- (y24);
    \draw[-, line width=0.5pt]  (x25) -- (y24);
    \draw[-, line width=0.5pt]  (y24) -- (z24);

    \draw[-, line width=0.5pt]  (x25) -- (y25);
    \draw[-, line width=0.5pt]  (x26) -- (y25);
    \draw[-, line width=0.5pt]  (y25) -- (z25);

    \draw[-, line width=0.5pt]  (x27) -- (y27);
    \draw[-, line width=0.5pt]  (x28) -- (y27);
    \draw[-, line width=0.5pt]  (y27) -- (z27);

    \draw[-, line width=0.5pt]  (x28) -- (y28);
    \draw[-, line width=0.5pt]  (x29) -- (y28);
    \draw[-, line width=0.5pt]  (y28) -- (z28);

    \draw[-, line width=0.5pt]  (x29) -- (y29);
    \draw[-, line width=0.5pt]  (x210) -- (y29);
    \draw[-, line width=0.5pt]  (y29) -- (z29);

    \draw[-, line width=0.5pt]  (x210) -- (y210);
    \draw[-, line width=0.5pt]  (x21) -- (y210);
    \draw[-, line width=0.5pt]  (y210) -- (z210);

    %\path (x25) -- (-2.75,0.00) node [black, midway, sloped] {\Large$\dots$};
    \path (-2.75,0.00) -- (-2.23,-1.62) node [black, midway, sloped] {$\dots$};

    %assignment part
    \node[draw, circle, line width=0.5pt, fill=white](u22) at (-0.85,-3.62)  []{};
    \node[draw, circle, line width=0.5pt, fill=white](u23) at (-0.85,-4.62)  []{};
    \node[draw, circle, line width=0.5pt, fill=white](u25) at (0.85,-3.62)  []{};
    \node[draw, circle, line width=0.5pt, fill=white](u26) at (0.85,-4.62)  []{};

    \draw[-, line width=0.5pt]  (x28) -- (u22);
    \draw[-, line width=0.5pt]  (u22) -- (u23);
    \draw[-, line width=0.5pt]  (x29) -- (u25);
    \draw[-, line width=0.5pt]  (u25) -- (u26);
\end{scope}

%y1
\begin{scope}[xshift=7cm]
%outer
    \node[draw, circle, line width=0.5pt, fill=red](x11) at (2.75,0.00)  []{};
    \node[draw, circle, line width=0.5pt, fill=blue](x12) at (2.23,1.62)  []{};
    \node[draw, circle, line width=0.5pt, fill=red](x13) at (0.85,2.62)  []{};
    \node[draw, circle, line width=0.5pt, fill=blue](x14) at (-0.85,2.62)  []{};
    \node[draw, circle, line width=0.5pt, fill=red](x15) at (-2.23,1.62)  []{};
    \node[draw, circle, line width=0.5pt, fill=blue](x16) at (-2.75,0.00)  []{};
    \node[draw, circle, line width=0.5pt, fill=red](x17) at (-2.23,-1.62)  []{};
    \node[draw, circle, line width=0.5pt, fill=blue](x18) at (-0.85,-2.62)  []{};
    \node[draw, circle, line width=0.5pt, fill=red](x19) at (0.85,-2.62)  []{};
    \node[draw, circle, line width=0.5pt, fill=blue](x110) at (2.23,-1.62)  []{};

    \draw[-, line width=0.5pt]  (x11) -- (x12);
    \draw[-, line width=0.5pt]  (x12) -- (x13);
    \draw[-, line width=0.5pt]  (x13) -- (x14);
    \draw[-, line width=0.5pt]  (x14) -- (x15);
    \draw[-, line width=0.5pt]  (x15) -- (x16);
    \draw[-, line width=0.5pt]  (x17) -- (x18);
    \draw[-, line width=0.5pt]  (x18) -- (x19);
    \draw[-, line width=0.5pt]  (x19) -- (x110);
    \draw[-, line width=0.5pt]  (x110) -- (x11);

    %middle
    \node[draw, circle, line width=0.5pt, fill=white](y11) at (1.90,0.61)  []{};
    \node[draw, circle, line width=0.5pt, fill=white](y12) at (1.23,1.64)  []{};
    \node[draw, circle, line width=0.5pt, fill=white](y13) at (0.00,2.00)  []{};
    \node[draw, circle, line width=0.5pt, fill=white](y14) at (-1.23,1.64)  []{};
    \node[draw, circle, line width=0.5pt, fill=white](y15) at (-1.90,0.61)  []{};
    %\node[draw, circle, line width=0.5pt, fill=white](y6) at (-1.90,-0.61)  []{};
    \node[draw, circle, line width=0.5pt, fill=white](y17) at (-1.23,-1.64)  []{};
    \node[draw, circle, line width=0.5pt, fill=white](y18) at (0.00,-2.00)  []{};
    \node[draw, circle, line width=0.5pt, fill=white](y19) at (1.23,-1.64)  []{};
    \node[draw, circle, line width=0.5pt, fill=white](y110) at (1.90,-0.61)  []{};

    %inner
    \node[draw, circle, line width=0.5pt, fill=white](z11) at (1.19,0.38)  []{};
    \node[draw, circle, line width=0.5pt, fill=white](z12) at (0.77,1.02)  []{};
    \node[draw, circle, line width=0.5pt, fill=white](z13) at (0.00,1.25)  []{};
    \node[draw, circle, line width=0.5pt, fill=white](z14) at (-0.77,1.02)  []{};
    \node[draw, circle, line width=0.5pt, fill=white](z15) at (-1.19,0.38)  []{};
    %\node[draw, circle, line width=0.5pt, fill=white](z6) at (-1.19,-0.38)  []{};
    \node[draw, circle, line width=0.5pt, fill=white](z17) at (-0.77,-1.02)  []{};
    \node[draw, circle, line width=0.5pt, fill=white](z18) at (0.00,-1.25)  []{};
    \node[draw, circle, line width=0.5pt, fill=white](z19) at (0.77,-1.02)  []{};
    \node[draw, circle, line width=0.5pt, fill=white](z110) at (1.19,-0.38)  []{};

    \node[] () at (0,0) []{\small $y_1$};

    %control points
    \coordinate (conty11) at (-3.55,0.00);
    \coordinate (conty12) at (3.55,0.00);

    %inner edges
    \draw[-, line width=0.5pt]  (x11) -- (y11);
    \draw[-, line width=0.5pt]  (x12) -- (y11);
    \draw[-, line width=0.5pt]  (y11) -- (z11);
    
    \draw[-, line width=0.5pt]  (x12) -- (y12);
    \draw[-, line width=0.5pt]  (x13) -- (y12);
    \draw[-, line width=0.5pt]  (y12) -- (z12);
    
    \draw[-, line width=0.5pt]  (x13) -- (y13);
    \draw[-, line width=0.5pt]  (x14) -- (y13);
    \draw[-, line width=0.5pt]  (y13) -- (z13);

    \draw[-, line width=0.5pt]  (x14) -- (y14);
    \draw[-, line width=0.5pt]  (x15) -- (y14);
    \draw[-, line width=0.5pt]  (y14) -- (z14);

    \draw[-, line width=0.5pt]  (x15) -- (y15);
    \draw[-, line width=0.5pt]  (x16) -- (y15);
    \draw[-, line width=0.5pt]  (y15) -- (z15);

    \draw[-, line width=0.5pt]  (x17) -- (y17);
    \draw[-, line width=0.5pt]  (x18) -- (y17);
    \draw[-, line width=0.5pt]  (y17) -- (z17);

    \draw[-, line width=0.5pt]  (x18) -- (y18);
    \draw[-, line width=0.5pt]  (x19) -- (y18);
    \draw[-, line width=0.5pt]  (y18) -- (z18);

    \draw[-, line width=0.5pt]  (x19) -- (y19);
    \draw[-, line width=0.5pt]  (x110) -- (y19);
    \draw[-, line width=0.5pt]  (y19) -- (z19);

    \draw[-, line width=0.5pt]  (x110) -- (y110);
    \draw[-, line width=0.5pt]  (x11) -- (y110);
    \draw[-, line width=0.5pt]  (y110) -- (z110);

    %\path (x15) -- (-2.75,0.00) node [black, midway, sloped] {$\dots$};
    \path (-2.75,0.00) -- (-2.23,-1.62) node [black, midway, sloped] {$\dots$};
\end{scope}

%y2
\begin{scope}[xshift=14cm]
%outer
    \node[draw, circle, line width=0.5pt, fill=red](x31) at (2.75,0.00)  []{};
    \node[draw, circle, line width=0.5pt, fill=blue](x32) at (2.23,1.62)  []{};
    \node[draw, circle, line width=0.5pt, fill=red](x33) at (0.85,2.62)  []{};
    \node[draw, circle, line width=0.5pt, fill=blue](x34) at (-0.85,2.62)  []{};
    \node[draw, circle, line width=0.5pt, fill=red](x35) at (-2.23,1.62)  []{};
    \node[draw, circle, line width=0.5pt, fill=blue](x36) at (-2.75,0.00)  []{};
    \node[draw, circle, line width=0.5pt, fill=red](x37) at (-2.23,-1.62)  []{};
    \node[draw, circle, line width=0.5pt, fill=blue](x38) at (-0.85,-2.62)  []{};
    \node[draw, circle, line width=0.5pt, fill=red](x39) at (0.85,-2.62)  []{};
    \node[draw, circle, line width=0.5pt, fill=blue](x310) at (2.23,-1.62)  []{};

    \draw[-, line width=0.5pt]  (x31) -- (x32);
    \draw[-, line width=0.5pt]  (x32) -- (x33);
    \draw[-, line width=0.5pt]  (x33) -- (x34);
    \draw[-, line width=0.5pt]  (x34) -- (x35);
    \draw[-, line width=0.5pt]  (x35) -- (x36);
    \draw[-, line width=0.5pt]  (x37) -- (x38);
    \draw[-, line width=0.5pt]  (x38) -- (x39);
    \draw[-, line width=0.5pt]  (x39) -- (x310);
    \draw[-, line width=0.5pt]  (x310) -- (x31);

    %middle
    \node[draw, circle, line width=0.5pt, fill=white](y31) at (1.90,0.61)  []{};
    \node[draw, circle, line width=0.5pt, fill=white](y32) at (1.23,1.64)  []{};
    \node[draw, circle, line width=0.5pt, fill=white](y33) at (0.00,2.00)  []{};
    \node[draw, circle, line width=0.5pt, fill=white](y34) at (-1.23,1.64)  []{};
    \node[draw, circle, line width=0.5pt, fill=white](y35) at (-1.90,0.61)  []{};
    %\node[draw, circle, line width=0.5pt, fill=white](y6) at (-1.90,-0.61)  []{};
    \node[draw, circle, line width=0.5pt, fill=white](y37) at (-1.23,-1.64)  []{};
    \node[draw, circle, line width=0.5pt, fill=white](y38) at (0.00,-2.00)  []{};
    \node[draw, circle, line width=0.5pt, fill=white](y39) at (1.23,-1.64)  []{};
    \node[draw, circle, line width=0.5pt, fill=white](y310) at (1.90,-0.61)  []{};

    %inner
    \node[draw, circle, line width=0.5pt, fill=white](z31) at (1.19,0.38)  []{};
    \node[draw, circle, line width=0.5pt, fill=white](z32) at (0.77,1.02)  []{};
    \node[draw, circle, line width=0.5pt, fill=white](z33) at (0.00,1.25)  []{};
    \node[draw, circle, line width=0.5pt, fill=white](z34) at (-0.77,1.02)  []{};
    \node[draw, circle, line width=0.5pt, fill=white](z35) at (-1.19,0.38)  []{};
    %\node[draw, circle, line width=0.5pt, fill=white](z6) at (-1.19,-0.38)  []{};
    \node[draw, circle, line width=0.5pt, fill=white](z37) at (-0.77,-1.02)  []{};
    \node[draw, circle, line width=0.5pt, fill=white](z38) at (0.00,-1.25)  []{};
    \node[draw, circle, line width=0.5pt, fill=white](z39) at (0.77,-1.02)  []{};
    \node[draw, circle, line width=0.5pt, fill=white](z310) at (1.19,-0.38)  []{};

    \node[] () at (0,0) []{\small $y_2$};

    %control points
    %\coordinate (conty21) at (0.85,3);
    \coordinate (conty22) at (-2.23,2.2);

    %inner edges
    \draw[-, line width=0.5pt]  (x31) -- (y31);
    \draw[-, line width=0.5pt]  (x32) -- (y31);
    \draw[-, line width=0.5pt]  (y31) -- (z31);
    
    \draw[-, line width=0.5pt]  (x32) -- (y32);
    \draw[-, line width=0.5pt]  (x33) -- (y32);
    \draw[-, line width=0.5pt]  (y32) -- (z32);
    
    \draw[-, line width=0.5pt]  (x33) -- (y33);
    \draw[-, line width=0.5pt]  (x34) -- (y33);
    \draw[-, line width=0.5pt]  (y33) -- (z33);

    \draw[-, line width=0.5pt]  (x34) -- (y34);
    \draw[-, line width=0.5pt]  (x35) -- (y34);
    \draw[-, line width=0.5pt]  (y34) -- (z34);

    \draw[-, line width=0.5pt]  (x35) -- (y35);
    \draw[-, line width=0.5pt]  (x36) -- (y35);
    \draw[-, line width=0.5pt]  (y35) -- (z35);

    \draw[-, line width=0.5pt]  (x37) -- (y37);
    \draw[-, line width=0.5pt]  (x38) -- (y37);
    \draw[-, line width=0.5pt]  (y37) -- (z37);

    \draw[-, line width=0.5pt]  (x38) -- (y38);
    \draw[-, line width=0.5pt]  (x39) -- (y38);
    \draw[-, line width=0.5pt]  (y38) -- (z38);

    \draw[-, line width=0.5pt]  (x39) -- (y39);
    \draw[-, line width=0.5pt]  (x310) -- (y39);
    \draw[-, line width=0.5pt]  (y39) -- (z39);

    \draw[-, line width=0.5pt]  (x310) -- (y310);
    \draw[-, line width=0.5pt]  (x31) -- (y310);
    \draw[-, line width=0.5pt]  (y310) -- (z310);

    %\path (x35) -- (-2.75,0.00) node [black, midway, sloped] {\Large$\dots$};
    \path (-2.75,0.00) -- (-2.23,-1.62) node [black, midway, sloped] {$\dots$};
\end{scope}

%c1
\begin{scope}[xshift=-0.3cm,yshift=6cm]
    \node[draw, circle, line width=0.5pt, fill=white](x41) at (0,0)  []{};
    \node[draw, circle, line width=0.5pt, fill=white](x42) at (6,0)  []{};
    \node[draw, circle, line width=0.5pt, fill=white](x43) at (3,5)  []{};
    \node[draw, circle, line width=0.5pt, fill=white](x44) at (3,2.5)  []{};
    \node[draw, circle, line width=0.5pt, fill=white](x45) at (3,1.5)  []{};

	\draw[-, line width=0.5pt]  (x41) -- (x42);
    \draw[-, line width=0.5pt]  (x42) -- (x43);
    \draw[-, line width=0.5pt]  (x43) -- (x41);
    \draw[-, line width=0.5pt]  (x41) -- (x44);
    \draw[-, line width=0.5pt]  (x42) -- (x44);
    \draw[-, line width=0.5pt]  (x43) -- (x44);
    \draw[-, line width=0.5pt]  (x44) -- (x45);

    %control points
    \coordinate (contc11) at (-0.5,0);
    \coordinate (contc13) at (0,-1);

    \coordinate (contc12) at (2.5,5);
    \coordinate (contc14) at (4.5,-1);

    \coordinate (contc16) at (6.5,0);
    \coordinate (contc15) at (3.5,5);

    \node[] () at (3,-0.5) []{$c_1$};
\end{scope}

%c2
\begin{scope}[xshift=8.4cm,yshift=6cm]
    \node[draw, circle, line width=0.5pt, fill=white](x51) at (0,0)  []{};
    \node[draw, circle, line width=0.5pt, fill=white](x52) at (6,0)  []{};
    \node[draw, circle, line width=0.5pt, fill=white](x53) at (3,5)  []{};
    \node[draw, circle, line width=0.5pt, fill=white](x54) at (3,2.5)  []{};
    \node[draw, circle, line width=0.5pt, fill=white](x55) at (3,1.5)  []{};

	\draw[-, line width=0.5pt]  (x51) -- (x52);
    \draw[-, line width=0.5pt]  (x52) -- (x53);
    \draw[-, line width=0.5pt]  (x53) -- (x51);
    \draw[-, line width=0.5pt]  (x51) -- (x54);
    \draw[-, line width=0.5pt]  (x52) -- (x54);
    \draw[-, line width=0.5pt]  (x53) -- (x54);
    \draw[-, line width=0.5pt]  (x54) -- (x55);

    \node[] () at (3,-0.5) []{$c_2$};

    %control points
    \coordinate (contc21) at (6.5,0);
    \coordinate (contc22) at (3.5,5);
    \coordinate (contc23) at (5.5,-1);
    \coordinate (contc24) at (1.5,-1);

    \coordinate (contc26) at (-0.5,0);
    \coordinate (contc25) at (2.5,5);
    
\end{scope}

%x1:
\begin{scope}[xshift=7cm,yshift=17cm]
%outer
    \node[draw, circle, line width=0.5pt, fill=blue](x61) at (2.75,0.00)  []{};
    \node[draw, circle, line width=0.5pt, fill=red](x62) at (2.23,1.62)  []{};
    \node[draw, circle, line width=0.5pt, fill=blue](x63) at (0.85,2.62)  []{};
    \node[draw, circle, line width=0.5pt, fill=red](x64) at (-0.85,2.62)  []{};
    \node[draw, circle, line width=0.5pt, fill=blue](x65) at (-2.23,1.62)  []{};
    \node[draw, circle, line width=0.5pt, fill=red](x66) at (-2.75,0.00)  []{};
    \node[draw, circle, line width=0.5pt, fill=blue](x67) at (-2.23,-1.62)  []{};
    \node[draw, circle, line width=0.5pt, fill=red](x68) at (-0.85,-2.62)  []{};
    \node[draw, circle, line width=0.5pt, fill=blue](x69) at (0.85,-2.62)  []{};
    \node[draw, circle, line width=0.5pt, fill=red](x610) at (2.23,-1.62)  []{};

    \draw[-, line width=0.5pt]  (x61) -- (x62);
    \draw[-, line width=0.5pt]  (x62) -- (x63);
    \draw[-, line width=0.5pt]  (x63) -- (x64);
    \draw[-, line width=0.5pt]  (x64) -- (x65);
    \draw[-, line width=0.5pt]  (x65) -- (x66);
    \draw[-, line width=0.5pt]  (x67) -- (x68);
    \draw[-, line width=0.5pt]  (x68) -- (x69);
    \draw[-, line width=0.5pt]  (x69) -- (x610);
    \draw[-, line width=0.5pt]  (x610) -- (x61);

    %middle
    \node[draw, circle, line width=0.5pt, fill=white](y61) at (1.90,0.61)  []{};
    \node[draw, circle, line width=0.5pt, fill=white](y62) at (1.23,1.64)  []{};
    \node[draw, circle, line width=0.5pt, fill=white](y63) at (0.00,2.00)  []{};
    \node[draw, circle, line width=0.5pt, fill=white](y64) at (-1.23,1.64)  []{};
    \node[draw, circle, line width=0.5pt, fill=white](y65) at (-1.90,0.61)  []{};
    %\node[draw, circle, line width=0.5pt, fill=white](y6) at (-1.90,-0.61)  []{};
    \node[draw, circle, line width=0.5pt, fill=white](y67) at (-1.23,-1.64)  []{};
    \node[draw, circle, line width=0.5pt, fill=white](y68) at (0.00,-2.00)  []{};
    \node[draw, circle, line width=0.5pt, fill=white](y69) at (1.23,-1.64)  []{};
    \node[draw, circle, line width=0.5pt, fill=white](y610) at (1.90,-0.61)  []{};

    %inner
    \node[draw, circle, line width=0.5pt, fill=white](z61) at (1.19,0.38)  []{};
    \node[draw, circle, line width=0.5pt, fill=white](z62) at (0.77,1.02)  []{};
    \node[draw, circle, line width=0.5pt, fill=white](z63) at (0.00,1.25)  []{};
    \node[draw, circle, line width=0.5pt, fill=white](z64) at (-0.77,1.02)  []{};
    \node[draw, circle, line width=0.5pt, fill=white](z65) at (-1.19,0.38)  []{};
    %\node[draw, circle, line width=0.5pt, fill=white](z6) at (-1.19,-0.38)  []{};
    \node[draw, circle, line width=0.5pt, fill=white](z67) at (-0.77,-1.02)  []{};
    \node[draw, circle, line width=0.5pt, fill=white](z68) at (0.00,-1.25)  []{};
    \node[draw, circle, line width=0.5pt, fill=white](z69) at (0.77,-1.02)  []{};
    \node[draw, circle, line width=0.5pt, fill=white](z610) at (1.19,-0.38)  []{};

    \node[] () at (0,0) []{\small $x_1$};

    %inner edges
    \draw[-, line width=0.5pt]  (x61) -- (y61);
    \draw[-, line width=0.5pt]  (x62) -- (y61);
    \draw[-, line width=0.5pt]  (y61) -- (z61);
    
    \draw[-, line width=0.5pt]  (x62) -- (y62);
    \draw[-, line width=0.5pt]  (x63) -- (y62);
    \draw[-, line width=0.5pt]  (y62) -- (z62);
    
    \draw[-, line width=0.5pt]  (x63) -- (y63);
    \draw[-, line width=0.5pt]  (x64) -- (y63);
    \draw[-, line width=0.5pt]  (y63) -- (z63);

    \draw[-, line width=0.5pt]  (x64) -- (y64);
    \draw[-, line width=0.5pt]  (x65) -- (y64);
    \draw[-, line width=0.5pt]  (y64) -- (z64);
    
    \draw[-, line width=0.5pt]  (x65) -- (y65);
    \draw[-, line width=0.5pt]  (x66) -- (y65);
    \draw[-, line width=0.5pt]  (y65) -- (z65);

    \draw[-, line width=0.5pt]  (x67) -- (y67);
    \draw[-, line width=0.5pt]  (x68) -- (y67);
    \draw[-, line width=0.5pt]  (y67) -- (z67);

    \draw[-, line width=0.5pt]  (x68) -- (y68);
    \draw[-, line width=0.5pt]  (x69) -- (y68);
    \draw[-, line width=0.5pt]  (y68) -- (z68);

    \draw[-, line width=0.5pt]  (x69) -- (y69);
    \draw[-, line width=0.5pt]  (x610) -- (y69);
    \draw[-, line width=0.5pt]  (y69) -- (z69);

    \draw[-, line width=0.5pt]  (x610) -- (y610);
    \draw[-, line width=0.5pt]  (x61) -- (y610);
    \draw[-, line width=0.5pt]  (y610) -- (z610);

    %\path (x65) -- (-2.75,0.00) node [black, midway, sloped] {\Large$\dots$};
    \path (-2.75,0.00) -- (-2.23,-1.62) node [black, midway, sloped] {$\dots$};

    %assignment part
    \node[draw, circle, line width=0.5pt, fill=white](u62) at (-0.85,-3.62)  []{};
    \node[draw, circle, line width=0.5pt, fill=white](u63) at (-0.85,-4.62)  []{};
    \node[draw, circle, line width=0.5pt, fill=white](u65) at (0.85,-3.62)  []{};
    \node[draw, circle, line width=0.5pt, fill=white](u66) at (0.85,-4.62)  []{};

    \draw[-, line width=0.5pt]  (x68) -- (u62);
    \draw[-, line width=0.5pt]  (u62) -- (u63);
    \draw[-, line width=0.5pt]  (x69) -- (u65);
    \draw[-, line width=0.5pt]  (u65) -- (u66);

    %control points
    \coordinate (contx11) at (-1.85,2.62);
    \coordinate (contx12) at (1.85,2.62);
\end{scope}

%edges between gadgets 
\draw[-, line width=1pt]  (x41)  edge   [bend right=30] (x25);
\draw[-, line width=1pt]  (x42)  edge   [bend right=30] (x15);
\draw[-, line width=1pt]  (x51)  edge   [bend left=30] (x12);
\draw[-, line width=1pt]  (x52)  edge   [bend left=30] (x34);
\draw[-, line width=1pt]  (x43)  edge   [bend left=30] (x65);
\draw[-, line width=1pt]  (x53)  edge   [bend right=30] (x62);

%bad edges between gadgets
%\draw[-, line width=1pt, color=olive]  (contx21)  edge   [bend left=30] (contc11);

\draw[-, line width=1pt, color=olive]  (x26)  edge   [] (contx21);
\draw[-, line width=1pt, color=olive]  (contx21)  edge   [bend left=30] (contc11);
\draw[-, line width=1pt, color=olive]  (contc11)  edge   (contc12);
\draw[-, line width=1pt, color=olive]  (contc12)  edge   [bend left=30] (contx11);
\draw[-, line width=1pt, color=olive]  (contx11)  edge   [] (x64);

\draw[-, line width=1pt, color=olive]  (x24)  edge   [] (contx22);
\draw[-, line width=1pt, color=olive]  (contx22)  edge   [bend left=20] (contc13);
\draw[-, line width=1pt, color=olive]  (contc13)  edge   (contc14);
\draw[-, line width=1pt, color=olive]  (contc14)  edge   [bend right=10] (conty11);
\draw[-, line width=1pt, color=olive]  (conty11)  edge   [] (x16);

%\draw[-, line width=1pt, color=olive]  (x33)  edge   [] (conty21);
\draw[-, line width=1pt, color=olive]  (x33)  edge   [bend right=20] (contc21);
\draw[-, line width=1pt, color=olive]  (contc21)  edge   (contc22);
\draw[-, line width=1pt, color=olive]  (contc22)  edge   [bend right=30] (contx12);
\draw[-, line width=1pt, color=olive]  (contx12)  edge   [] (x63);

\draw[-, line width=1pt, color=olive]  (x35)  edge   [] (conty22);
\draw[-, line width=1pt, color=olive]  (conty22)  edge   [bend right=30] (contc23);
\draw[-, line width=1pt, color=olive]  (contc23)  edge   (contc24);
\draw[-, line width=1pt, color=olive]  (contc24)  edge   [bend left=10] (conty12);
\draw[-, line width=1pt, color=olive]  (conty12)  edge   [] (x11);

\draw[-, line width=1pt, color=olive]  (x13)  edge   [bend right=30] (contc26);
\draw[-, line width=1pt, color=olive]  (contc26)  edge   (contc25);
\draw[-, line width=1pt, color=olive]  (contc25)  edge   [bend right=30] (x61);

\draw[-, line width=1pt, color=olive]  (x14)  edge   [bend left=30] (contc16);
\draw[-, line width=1pt, color=olive]  (contc16)  edge   (contc15);
\draw[-, line width=1pt, color=olive]  (contc15)  edge   [bend left=30] (x66);

\end{tikzpicture}

\caption{An example of the construction for a formula $\phi$ containing the two clauses $c_1=(x_1\lor x_2\lor \lnot y_1)$ and $c_2=(\lnot x_1 \lor y_1 \lor y_2)$. The bold (olive resp.) edges represent the edges added in the first (second resp.) step of the edge-adding procedure.}\label{fig:hardness-example}
\end{figure}

We now present some properties of the constructed graph $G$ that will be used in the reduction.
In the following we follow the naming conventions depicted in Figure~\ref{fig:hardness-gadgets}(c).

\begin{claim}\label{cl:hardness-n1-deletions}
    Let $S \subseteq V(G)$ with $|S| \leq n_1$ and $G-S$ having a unique minimum vertex cover.
    It holds that $|S \cap \{u_i^2,u_i^3,u_i^5,u_i^6\}| = 1$ for all $i \in [n_1]$.
\end{claim}

\begin{proofclaim}
    Recall that by construction the graph $G$ contains $n_1$ copies of the $x$-gadget.
    In the following, let $U$ denote the unique minimum vertex cover of $G-S$.

    We first show that for all $i \in [n_1]$, $S$ contains exactly one vertex among vertices $\{u_i^1,u_i^2,u_i^3,u_i^4,u_i^5,u_i^6\}$ in $x_i$-gadget.
    Towards a contradiction, assume that there exists a variable~$x_i$ such that $S$ contains no such vertices from the $x_i$-gadget.
    Then, it holds that set $U$ must contain at least one of the vertices in $\{u_i^1,u_i^4\}$.
    Assume that $u_i^1 \in U$. In that case, $U$ must also contain any one vertex among $\{u_i^2,u_i^3\}$, and both options are valid,
    contradicting the uniqueness of $U$. The case $u_i^4 \in U$ is analogous.

    Now assume that there exists $i \in [n_1]$ such that $u_i^1 \in S$.
    In that case due to the previous paragraph it holds that $u_i^2, u_i^3 \notin S$.
    Similarly to before, $U$ must contain any one vertex among $\{u_i^2,u_i^3\}$, and both options are valid,
    contradicting the uniqueness of $U$. The case $u_i^4 \in U$ is analogous.
\end{proofclaim}

Next, we have the following observation.

\begin{observation}\label{obs:hardness-minimum}
    Let $S \subseteq V(G)$ such that $S$ contains a single vertex per $x$-gadget and let $U$ be a vertex cover of $G-S$.
    It holds that $U$ contains at least:
    \begin{itemize}
        \item  $3$ vertices from each $c$-gadget,
        \item  $\frac{3q}{2}$ vertices from each $y$-gadget,
        \item  $\frac{3p}{2}+4$ vertices from each $x$-gadget. 
    \end{itemize}
    Moreover, any vertex cover of $G-S$ that contains exactly as many vertices from each gadget as specified above is minimum.
\end{observation}

In the following claim we focus on the $x_i$-gadget, for any $i\in[n_1]$, and specify which of its edges 
%In the following claim we specify the behavior of the $x_i$-gadget, for each $i\in[n_1]$, in regards to $D_i$. %relates the deletions we must have with vertices that must belong to the unique vertex cover of $G'$.

\begin{claim}\label{cl:hardness-u1-or-u4}
    Let $S \subseteq V(G)$ with $|S| \leq n_1$ and $G-S$ having a unique minimum vertex cover $U \subseteq V(G-S)$.
    For every $i \in [n_1]$, if $S \cap \{u_i^2,u_i^3\} \neq \emptyset$ then $u_i^1 \in U$, else if $S \cap \{u_i^5,u_i^6\} \neq \emptyset$ then $u_i^4 \in U$.
\end{claim}
\begin{proofclaim}
Fix an $i \in [n_1]$ for which we prove the statement.
Recall that by \Cref{cl:hardness-n1-deletions} it holds that $|S \cap \{u_i^2,u_i^3,u_i^5,u_i^6\}| = 1$.
Consider the case where $S \cap \{u_i^2,u_i^3\} \neq \emptyset$.
For the sake of contradiction, assume that $u_i^1 \notin U$, which implies that $u_i^4 \in U$.
Then, $U$ must also contain any one vertex among $\{u_i^5,u_i^6\}$, and both options are valid,
contradicting the uniqueness of $U$.
The case $S \cap \{u_i^5,u_i^6\} \neq \emptyset$ is analogous.
\end{proofclaim}

We now focus our attention on the $c$-gadgets. In the following we follow the naming conventions depicted in Figure~\ref{fig:hardness-gadgets}(a).

\begin{claim}\label{cl:hardness-c-gadgets-3vc}
    Let $S \subseteq V(G)$ with $|S| \leq n_1$ and $G-S$ having a unique minimum vertex cover $U \subseteq V(G-S)$.
    It holds that $|U \cap \{\ell_c^1,\ell_c^2,\ell_c^3\}| = 2$ for any $c$-gadget in $G$.
\end{claim}

\begin{proofclaim}
    Fix a $c$-gadget for which we prove the statement.
    By \Cref{cl:hardness-n1-deletions}, it holds that $S$ does not contain any of its vertices.
    Observe that $|U \cap \{\ell_c^1,\ell_c^2,\ell_c^3\}| \ge 2$, as those three vertices induce a $K_3$.
    Assume that $|U \cap \{\ell_c^1,\ell_c^2,\ell_c^3\}| > 2$.
    In this case $U$ must also contain at least one vertex in $\{w_c,z_c\}$,
    and both options are valid.
    This contradicts the uniqueness of $U$.
\end{proofclaim}

The next property we need concerns the $x$- and $y$-gadgets. Since these gadgets behave in the same way in what follows, we will henceforth denote them as $v$-gadgets. That is, a $v$-gadget can be either a $x$- or a $y$-gadget. We will also use the superscript and subscript notations introduced in Figures~\ref{fig:hardness-gadgets}(b) and~(c) for the $v$-gadgets. 

\begin{claim}\label{cl:hardness-xy-gadgets-alternate}
    Let $S \subseteq V(G)$ with $|S| \leq n_1$ and $G-S$ having a unique minimum vertex cover $U \subseteq V(G-S)$.
    Let $\mathcal{V}$ be any $v$-gadget and let $W$ be the intersection of $U$ and $V(\mathcal{V})$. Then either 
    \begin{itemize}
        \item $W$ contains all the red and no blue vertices of $\mathcal{V}$, or
        \item $W$ contains all the blue and no red vertices of $\mathcal{V}$. 
    \end{itemize}
\end{claim}
\begin{proofclaim}
    By \Cref{cl:hardness-n1-deletions}, it holds that $S$ does not contain any red or blue vertex.
    %Let $\mathcal{V}$ be the $x_i$-gadget, for any $i\in[n_1]$ (the case of $Y$ being a $y_i$-gadget is treated in the same way). 
    Clearly, $W$ must contain at least one vertex from each pair of adjacent red and blue vertices of $\mathcal{V}$, as otherwise $U$ is not a vertex cover.
    Thus, it suffices to show that $W$ does not contain both the red and blue vertices of any pair of adjacent colored vertices of $\mathcal{V}$.
    Assume that it does, that is, assume that $W$ contains both the vertices $v_i^\alpha$ and $v_i^{\alpha+1}$ (for some $\alpha \in [p-1]$ or $\alpha \in [q-1]$).
    Then, $W$ must also contain at least one vertex from $\{w_i^\alpha,z_i^\alpha\}$, and both options are valid.
    The case where $\mathcal{V}$ is the $x_i$-gadget and $W$ contains both the vertices in $\{u_i^1,u_i^4\}$ is treated in the same way.
    In any case, we obtain a contradiction to the uniqueness of $U$. 
\end{proofclaim}

%\nikos{I am not sure if it will help but I would have given a name in the "behavior". I.e. if contains all blue then its a Yes behavior otherwise a No behavior. Maybe this will help in order to relate the VCs to the assignments.}

%\nikos{maybe it is worth proving that: select a clause gadget and let $z_1$, $z_2$, $z_3$ be the indices related to the the x and y gadgets that have vertices incident to vertices of the clause gadget. "Then, only one of them can use the color that satisfies the the clause".}

The final property we observe concerns the interplay between the clause and the variable gadgets.
To ease the exposition of what follows, we introduce some additional notation.
Consider a $c$-gadget of $G$, and let $v_{i_1}^\alpha$, $v_{i_2}^\beta$, and $v_{i_3}^\gamma$ be vertices belonging in
the $v_{i_1}$-, $v_{i_2}$-, and $v_{i_3}$-gadgets respectively,
such that $\ell_c^1 v_{i_1}^\alpha \in E(G)$, $\ell_c^2 v_{i_2}^\beta \in E(G)$, and $\ell_c^3 v_{i_3}^\gamma \in E(G)$.
The subgraph of $G$ that is induced by this $c$-gadget and the vertices $v_{i_1}^\alpha$, $v_{i_2}^\beta$, and $v_{i_3}^\gamma$ will henceforth be named the \emph{important $c$-subgraph} of $G$.
Notice that due to \Cref{cl:hardness-n1-deletions}, any set $S \subseteq V(G)$ with $|S| \leq n_1$ and $G-S$ having a unique minimum vertex cover
does not include any vertices appearing in the important $c$-subgraph.
Consequently, in this case the important $c$-subgraph of $G-S$ refers to the same graph.

\begin{claim}\label{cl:hardness-x-c-gadgets}
    Let $S \subseteq V(G)$ with $|S| \leq n_1$ and $G-S$ having a unique minimum vertex cover $U \subseteq V(G-S)$.
    Consider the important $c$-subgraph of $G-S$ for some $c$-gadget.
    It holds that $|U \cap \{v_{i_1}^\alpha,v_{i_2}^\beta,v_{i_3}^\gamma\}| = 1$.  
\end{claim}
\begin{proofclaim}
    Let $W = U \cap \{v_{i_1}^\alpha,v_{i_2}^\beta,v_{i_3}^\gamma\}$ and observe that $|W| > 0$.
    Indeed, let us assume that $|W|=0$.
    We have from \Cref{cl:hardness-c-gadgets-3vc} that at least one among $\ell_c^1,\ell_c^2,\ell_c^3$, say $\ell_c^1$, does not belong to $U$.
    Then, we have that neither $\ell_c^1$ nor its adjacent vertex $v_{i_1}^\alpha$ belong to $U$,
    contradicting the fact that $U$ is a vertex cover of $G-S$.
    It thus suffices to show that $|W|<2$.
    Towards a contradiction, assume that $|W|\geq 2$ and let, w.l.o.g., $\{v_{i_1}^\alpha,v_{i_2}^\beta\} \subseteq W$.
    It follows by \Cref{cl:hardness-xy-gadgets-alternate} that $v_i^{\alpha-1}, v_j^{\beta+1} \notin U$.
    But $v_i^{\alpha-1} v_j^{\beta+1} \in E(G-S)$, contradicting the fact that $U$ is a vertex cover of $G-S$. 
\end{proofclaim}

%\manolis{Maybe omit this paragraph and just mention that the two claims lead to the following? I find it rather confusing.}
By considering both Claims~\ref{cl:hardness-c-gadgets-3vc} and~\ref{cl:hardness-x-c-gadgets},
we get that $U$ will contain exactly two out of the three ``outer'' vertices for each $c$-gadget, and exactly one of the adjacent vertices of the $v$-gadgets that are linked to this $c$-gadget. Let us denote this one vertex by $\nu$. The subtle detail that is implied by the proofs of these two claims is that the neighbor of $\nu$ in the $c$-gadget we are considering is exactly the ``outer'' vertex of this gadget that is not included in $U$. Allow us to formally state this. 

\begin{corollary}\label{cor:hardness-bij}
    Let $S \subseteq V(G)$ with $|S| \leq n_1$ and $G-S$ having a unique minimum vertex cover $U \subseteq V(G-S)$.
    Consider any $c$-gadget and the corresponding important $c$-subgraph of $G$.
    Let $W = U\cap \{\ell_c^1,\ell_c^2,\ell_c^3,v_{i_1}^\alpha,v_{i_2}^\beta,v_{i_3}^\gamma\}$.
    Then either
    \begin{itemize}
        \item $W=\{\ell_c^1,\ell_c^2, v_{i_3}^\gamma\}$, or
        \item $W=\{\ell_c^1,\ell_c^3,v_{i_2}^\beta\}$, or
        \item $W=\{\ell_c^2,\ell_c^3,v_{i_1}^\alpha\}$.
    \end{itemize}
\end{corollary} 

We are now ready to present our reduction. We first assume that we have a yes-instance of the \UQPSAT problem. That is, we have a truth-assignment $\tau$ of the variables of type $x$ which can be uniquely extended into a truth-assignment $\sigma$ of the variables of type $x$ and $y$ that results in $\phi$ being $1$-in-$3$ satisfied. We will show that the constructed graph $G=(V,E)$ has a \MUVCShort of order exactly~$n_1$. By \Cref{cl:hardness-n1-deletions}, it suffices to provide a set $S\subseteq V$ such that $G'=G-S$ has a unique minimum vertex cover~$U$, and $|S|=n_1$. We proceed as follows: for each $i\in[n_1]$, if $\tau(x_i)=\true$ ($\tau(x_i)=\false$ resp.), we include $u_i^5$ ($u_i^2$ resp.) in $S$. The set $S$ is comprised by only these vertices. It follows that $|S|=n_1$. It remains to show that $G'$ has a unique minimum vertex cover. Before proceeding, observe that every vertex denoted by a $w$, with any possible subscript and/or superscript, is contained in the unique minimum vertex cover of $G'$ (if it exists). Indeed, every edge $wz$ (for every possible subscript and/or superscript of $z$) must be covered, and any vertex cover that contains $z$ instead of $w$ is of order at least as big as a vertex cover that contains $w$. Thus, in what follows we will assume that these vertices are included in the unique minimum vertex cover and ignore them. 
Consider now the set $U$ built in the following fashion: 
\begin{enumerate}
    \item For every $c$-gadget, include in $U$ %the vertex $w_c$ and 
    the two vertices from $\{\ell_c^1,\ell_c^2,\ell_c^3\}$ that correspond to literals that do not satisfy the clause $c$ according to $\sigma$.
    \item For every $i\in[n_1]$, include in $U$ 
    %the vertices $\{w_i,w_i^0,\dots,w_i^p\}$ and 
    all the blue (red resp.) vertices and the vertex $u_i^2$ ($u_i^5$ resp.) if $\sigma(x_i)=\true$ ($\sigma (x_i)=\false$ resp.). Also, for every $i\in[n_2]$, include in $U$ 
    %the vertices $\{w_i^1,\dots,w_i^q\}$ and 
    all the blue (red resp.) vertices if $\sigma(y_i)=\true$ ($\sigma (y_i)=\false$ resp.).
\end{enumerate}
We claim that the set $U$ is indeed the unique minimum vertex cover of $G'$. Since it is straightforward to check that $U$ is indeed a vertex cover, we will focus on its minimality and its uniqueness. 
%\nikos{is it really straightforward? I prefer it this way but I just wonder}
%\nikos{(in comment the original and replaced by the blue)} 
% Towards a contradiction, assume there exists a set $U'$ that is also a vertex cover of $G'$, with $|U'|<U$. Then, either:
% \begin{itemize}
%     \item $U'$ does not contain at least one vertex included in $U$ at the item $1.$ above. Then we obtain a contradiction from Claim~\ref{cl:hardness-c-gadgets-3vc}.
%     \item $U'$ does not contain at least one vertex included in $U$ at the item $2.$ above. Then we obtain a contradiction from Claim~\ref{cl:hardness-xy-gadgets-alternate}.
% \end{itemize}
In order to prove that $U$ is minimum it suffices to observe that each gadget includes a minimum number of vertices, according to the values provided in \Cref{obs:hardness-minimum}. 
%Furthermore, since we have constructed a vertex cover that includes the minimum number of vertices from each gadget we can also conclude that any other minimum vertex cover must achieve the minimum from each gadget, otherwise it will have a bigger size than $U$ and it will not be minimum any more.

%Thus, $U$ is indeed a minimum vertex cover of $G'$. 
We finally argue about the uniqueness of $U$. Towards a contradiction, assume that there exists a set $U'$ that is also a minimum vertex cover of $G'$. Recall that by $v$-gadget we mean either a $x$- or a $y$-gadget. Observe first that any minimum vertex cover of $G'$ should include, for any $v$-gadget, either all its blue or all its red vertices. This holds true for both $U$ and $U'$.
%Because we want $U'$ to be minimal we know that for any $x$ and $y$-gadget, $U'$ includes either all their blue or all their red vertices \nikos{I think that this should be lemma somewhere. Notice this is not related to the case where we have unique vc}. 
Now, consider the assignment $\sigma'$ such that $\sigma'(v)=\true$ if $U'$ contains all the blue vertices of the $v$-gadget and $\sigma'(v)=\false$, if $U'$ contains all the red vertices of the $v$-gadget (where $v$ is a variable either of type $x$ or $y$). 
We claim that if $U'$ exists, then:
\begin{itemize}
    \item $\sigma(x_i)=\sigma'(x_i)$, for all $i \in [n_1]$,
    \item there exists $j \in [n_2]$ such that $\sigma(y_j)\neq \sigma'(y_j)$, and 
    \item $\sigma'$ sets exactly one literal of each clause to $\true$. 
\end{itemize}
If this holds it is a contradiction to the uniqueness of $\sigma$.

We begin by proving the first item. 
%First we prove that $\sigma(x_i)=\sigma'(x_i)$, for all $i \in [n_1]$. 
Assume that it is not true. W.l.o.g., assume that $\sigma (x_i)=\true$ and $\sigma' (x_i)=\false$, for some $i \in [n_1]$. Then $U$ includes all the blue and $U'$ includes all the red vertices of the $x_i$-gadget. In this case we claim that $|U'|>|U|$. Indeed, since $\sigma (x_i)=\true$, we have that $u_i^5 \notin V(G')$. Therefore, the intersection of $U'$ with the vertices of the $x_i$-gadget must include all its red vertices, all its $w$ vertices, one vertex from the set $\{u_i^2,u_i^3\}$; %to cover the edge $u_i^2u_i^3$ 
and one vertex from the set $\{u_i^4,u_i^5\}$. 
%to cover the edge $u_i^4u_i^5$ (as $u_i^5$ is blue, thus its not directly included in $U'$). 
Therefore, there exists a gadget where $U'$ includes more than the minimum number of vertices specified in \Cref{obs:hardness-minimum}, which means that $|U|<|U'|$.

Next we prove the second item.
%Now, we will prove that there exists $j \in [n_2]$ such that $\sigma(y_j)=\sigma'(y_j)$. 
Let us assume it is not true, i.e., $\sigma(v)=\sigma'(v)$, for all variables. Since $U$ is assumed different from $U'$, there exists a $c$-gadget $C$ such that $U$ and $U'$ have different intersections with its vertices. 
%W.l.o.g., let $v_{i_1}^\alpha$, $v_{i_2}^\beta$ and $v_{i_3}^\gamma$ be vertices belonging in the $v_i$, $v_j$ and $v_k$-gadgets respectively, such that $\ell_c^1v_{i_1}^\alpha\in E(G')$, $\ell_c^2v_{i_2}^\beta\in E(G')$ and $\ell_c^3v_{i_3}^\gamma\in E(G')$. 
Let us consider the important $c$-subgraph that corresponds to $C$. 
Since $U$ and $U'$ include the same vertices from all variable gadgets, we have that $U \cap \{ v_{i_1}^\alpha,v_{i_2}^\beta,v_{i_3}^\gamma\} = U' \cap \{ v_{i_1}^\alpha,v_{i_2}^\beta,v_{i_3}^\gamma\} $. Additionally, by the construction of $U$, we know that exactly one of $\{ v_{i_1}^\alpha,v_{i_2}^\beta,v_{i_3}^\gamma\}$ belongs in $U$. W.l.o.g., assume that $U' \cap \{ v_{i_1}^\alpha,v_{i_2}^\beta,v_{i_3}^\gamma\} = U \cap \{ v_{i_1}^\alpha,v_{i_2}^\beta,v_{i_3}^\gamma\} = \{ v_{i_1}^\alpha \}$. Since $U$ and $U'$ are minimum vertex covers, we have that they include both $\ell_c^2$ and $\ell_c^3$. In order for $U'$ to be minimum, it can include at most one extra vertex from $C$. This must be the vertex $w_c$ (as otherwise it needs to include both $\ell_c^1$ and $z_c$). However, this is exactly the same as $U$. This contradicts the assumption that $U$ and $U'$ do not include the same vertices from $C$. Therefore, it must exists $j \in [n_2]$ such that $\sigma(y_j)\neq\sigma'(y_j)$.

It remains to show the third item.
%that $\sigma'$ sets exactly one literal of each clause to true. 
Consider a clause $c$ along with its $c$-gadget and the corresponding important $c$-subgraph. 
%We will again assume that $v_{i_1}^\alpha$, $v_{i_2}^\beta$ and $v_{i_3}^\gamma$ are the vertices belonging in the $v_i$, $v_j$ and $v_k$-gadgets respectively, such that $\ell_c^1v_{i_1}^\alpha\in E(G')$, $\ell_c^2v_{i_2}^\beta\in E(G')$ and $\ell_c^3v_{i_3}^\gamma\in E(G')$. 
Assume that $\sigma'$ does not satisfy $c$. That is, all literals of $c$ are false by $\sigma'$. Then, $U' \cap \{ v_{i_1}^\alpha,v_{i_2}^\beta,v_{i_3}^\gamma\}=\emptyset$ (by the definition of $\sigma'$). Then $U'$ must include all of $\ell_c^1$, $\ell_c^1$, $\ell_c^1$ and one of the $\{w_c,z_c\}$, contradicting the minimality of~$U'$. 
It remains to show that $|U' \cap \{ v_{i_1}^\alpha,v_{i_2}^\beta,v_{i_3}^\gamma\}| < 2$. Assume otherwise and, w.l.o.g., let $v_{i_1}^\alpha,v_{i_2}^\beta \in U'$. Recall that $U'$ cannot contain two consecutive vertices from any variable gadget as otherwise it will be bigger than $U$. Therefore, $v_i^{\alpha-1},v_j^{\beta+1} \notin U'$. Thus, $U'$ does not cover the edge $v_i^{\alpha-1} v_j^{\beta+1}$, a contradiction to $U'$ being a vertex cover. 

To sum up, we have managed to construct two truth-assignments $\sigma$ and $\sigma'$, which are both extensions of $\tau$. This contradicts the uniqueness of $\sigma$. Therefore, $U$ is unique, completing the first direction of the reduction.

%\nikos{comment the original}
% It follows by Corollary~\ref{cor:hardness-bij} that there exists an important $c$-subgraph and the sets $S=U\cap \{\ell_c^1,\ell_c^2,\ell_c^3,v_{i_1}^\alpha,v_{i_2}^\beta,v_{i_3}^\gamma\}$ and $S'=U'\cap \{\ell_c^1,\ell_c^2,\ell_c^3,v_{i_1}^\alpha,v_{i_2}^\beta,v_{i_3}^\gamma\}$, such that $S\neq S'$. Assume, w.l.o.g., that $S=\{\ell_c^1,\ell_c^3,v_{i_2}^\beta\}$ and $S'=\{\ell_c^1,\ell_c^2,v_{i_3}^\gamma\}$. By the construction of $U$, we have that according to $\sigma$, the clause $c$ is satisfied only by its literal $\ell_c^2$. Consider now the assignment $\sigma'$ that differs from $\sigma$ only in the assignments of the variables of the clause $c$, and is such that $\sigma'(v_j)\neq \sigma(v_j)$ and $\sigma'(v_k)\neq \sigma (v_k)$. Since $U'$ is a minimum vertex cover of $G'$, we obtain that $\sigma'$ also $1$-in-$3$ satisfies $c$. That is a contradiction to the uniqueness of $\sigma$. \nikos{this is incorrect... $\sigma$ and $\sigma'$ may defer in more than one clause... Also, this must be related to the $x_i$s i.e. we need to prove that if such $\sigma'$ exists then it "agree" with $\sigma$ in the $x_i$s (otherwise there is no problem with the uniqueness).}

For the opposite direction, assume that we have a solution~$S$ of \MUVCShort of order $n_1$. That is, the graph $G'=G-S$ has a unique minimum vertex cover $U$. We will construct a truth-assignment $\tau$ of the variables of type $x$ that is uniquely extended into a truth assignment $\sigma$ of the variables of type $x$ and $y$ that results in $\phi$ being $1$-in-$3$ satisfied. 
First, it follows from \Cref{cl:hardness-n1-deletions} that for every $i \in [n_1]$, the set $S$ contains exactly one vertex $u$ of $x_i$;
in particular, this vertex is either one of $\{u_i^2, u_i^3\}$ or one of $\{u_i^5, u_i^6\}$.
Then, from \Cref{cl:hardness-u1-or-u4} we have that either $u_i^1 \in U$ and, by \Cref{cl:hardness-xy-gadgets-alternate},
the same holds true for all the red vertices of the $x_i$-gadget, or $u_i^4 \in U$ and the same holds true for all the blue vertices of the $x_i$-gadget. 

We consider the truth-assignment $\tau$ such that for each $i\in[n_1]$ we have $\tau(x_i)=\false$ if $\{u_i^2,u_i^3\}\cap S\neq \emptyset$ and $\tau(x_i)=\true$ if $\{u_i^4,u_i^5\}\cap S\neq \emptyset$. We claim that $\tau$ is uniquely extended into the desired truth-assignment $\sigma$. 

Towards a contradiction, assume there exist two different truth-assignments $\sigma$ and $\sigma'$ of the variables of $\phi$, both extending $\tau$, and such that $\phi$ is $1$-in-$3$ satisfied by both of them.
We will prove that this results into two different minimum vertex covers of $G'$. 
Consider now the set $U$ built in the following fashion: 
\begin{enumerate}
    \item For every $c$-gadget, include in $U$ %the vertex $w_c$ and 
    the two vertices from $\{\ell_c^1,\ell_c^2,\ell_c^3\}$ that correspond to literals that do not satisfy the clause $c$ according to $\sigma$.
    \item For every $i\in[n_1]$, include in $U$ 
    all the blue (red resp.) vertices and the vertex $u_i^2$ ($u_i^5$ resp.) if $\sigma(x_i)=\true$ ($\sigma (x_i)=\false$ resp.). Also, for every $i\in[n_2]$, include in $U$ 
    all the blue (red resp.) vertices if $\sigma(y_i)=\true$ ($\sigma (y_i)=\false$ resp.).
    \item Finally we include all $w$ vertices.
\end{enumerate}
The same way, we define $U'$ by replacing $\sigma$ by $\sigma'$.

Since $\sigma$ and $\sigma'$ are assumed different, we have that $U \neq U'$. It now suffices to prove that both $U$ and $U'$ are minimum vertex cover sets of $G'$. It is straightforward to observe that both $U$ and $U'$ are indeed vertex covers of $G'$. Finally, both $U$ and~$U'$ are using the minimum possible number of vertices from each gadget, according to \Cref{obs:hardness-minimum}. Therefore, $G'$ has two different minimum vertex covers which is a contradiction.

It remains to prove that there exists an extension $\sigma$ of the assignment $\tau$. 
Assume that we have the vertex cover $U$ of $G-S$. We set $\sigma(v) = \true$ if $U$ contains all blue vertices from the $v$-gadget and $\sigma(v) = \false$ otherwise. We claim that $\phi$ is $1$-in-$3$ satisfied by $\sigma$. Consider a clause $c$ in $\phi$ and let $\{ v_{i_1}^\alpha,v_{i_2}^\beta,v_{i_3}^\gamma\}$ be the vertices that are adjacent to the vertices of the $c$-gadget.
Recall that, by construction, any  $v\in \{v_{i_1}^\alpha,v_{i_2}^\beta,v_{i_3}^\gamma\}$ is blue if the corresponding variable appears in $c$ positively and red if the corresponding variable appears in $c$ negatively. Therefore, $v \in  \{v_{i_1}^\alpha,v_{i_2}^\beta,v_{i_3}^\gamma\} \cap U$ if and only if the literal that corresponds to $v$ in $c$ has been set true by $\sigma$. By \Cref{cor:hardness-bij} we know that exactly one of the vertices $\{v_{i_1}^\alpha,v_{i_2}^\beta,v_{i_3}^\gamma\}$ will be included in $U$. Thus $\phi$ is $1$-in-$3$ satisfied by $\sigma$. This completes the $\Sigma_2^P$-hardness for
\MUVCShort.

For \PAUVCShort, the proof is almost the same.
We just need to argue that the solution $S$ we constructed (starting from the solution of \UQPSAT) together with the minimum vertex cover $U$ of $G-S$ can give us a minimum vertex cover of $G$. That is, $S\cup U$ is a minimum vertex cover of $G$. This is indeed true as we need to contain the vertices of $S$ in order to cover the $u_i^2u_i^3$ (if $S\cap \{ u_i^2,u_i^3\} \neq \emptyset$) and $u_i^5u_i^6$ (if $S\cap \{ u_i^5,u_i^6\} \neq \emptyset$) edges, and it is always minimum because, for the rest of the gadgets, we are using the minimum number of vertices. 
\end{proof}

\fi

\ifshort
\begin{sketch}
We first argue about \MUVCShort belonging in $\Sigma_2^p$. 
%Recall that a decision problem is in $\Sigma_2^P$ if and only if it can be decided by a non-deterministic Turing machine with the added use of an \NP-oracle. 
Given a graph $G$ and integer $k$, assume we have a set $S\subseteq V(G)$ such that $|S|\leq k$ and $G'=G-S$ has a unique minimum vertex cover $U$. To verify that $U$ is indeed as required, it suffices to solve \PAUVCShort on $G'$ for $k=0$, which can be done in polynomial time with the help of an \NP-oracle~\cite{horiyama2024pauvc}. We then focus on proving that \MUVCShort is $\Sigma_2^P$-hard for planar graphs of maximum degree $5$. Slight modifications in our proof can lead to the same hardness result for the same family of graphs for \PAUVCShort.

The reduction is from \UQPSAT~\cite{demaine18sigma2}: we are given a 3CNF formula $\phi$ on the set of variables $\{x_1,\dots,x_{n_1},y_1,\dots,y_{n_2}\}$ and clauses $C=\{c_1,\dots,c_m\}$. Variables in $\{x_1,\dots,x_{n_1}\}$ ($\{y_1\dots,y_{n_2}\}$ resp.) are of \textit{type $x$} (\textit{type $y$} resp.). The task is to find a truth-assignment of the variables of type $x$ such that there exists a unique truth-assignment of the variables of type $y$ where each clause of $C$ is satisfied by exactly one literal. We will construct a graph $G$ which has an \MUVCShort of order $n_1$ if and only if $\phi$ is a yes-instance of \UQPSAT.

\begin{figure}[!t]
\centering
\subfloat[The $c$-gadget. %The three outer vertices represent the literals of $c$.
]{
\scalebox{0.7}{

\begin{tikzpicture}[inner sep=0.6mm]

	\node[draw, circle, line width=1pt, fill=white](x1) at (0,0)  [label=left: $\ell^1_c$]{};
    \node[draw, circle, line width=1pt, fill=white](x2) at (4,0)  [label=right: $\ell^2_c$]{};
    \node[draw, circle, line width=1pt, fill=white](x3) at (2,4)  [label=above: $\ell^3_c$]{};
    \node[draw, circle, line width=1pt, fill=white](x4) at (2,2)  [label=left: $w_c$]{};
    \node[draw, circle, line width=1pt, fill=white](x5) at (2,1)  [label=left: $z_c$]{};

	\draw[-, line width=1pt]  (x1) -- (x2);
    \draw[-, line width=1pt]  (x2) -- (x3);
    \draw[-, line width=1pt]  (x3) -- (x1);
    \draw[-, line width=1pt]  (x1) -- (x4);
    \draw[-, line width=1pt]  (x2) -- (x4);
    \draw[-, line width=1pt]  (x3) -- (x4);
    \draw[-, line width=1pt]  (x4) -- (x5);

\end{tikzpicture}
}
}
\subfloat[The $y_i$-gadget.]{
\scalebox{0.7}{
\begin{tikzpicture}[inner sep=0.6mm]

    %outer
    \node[draw, circle, line width=1pt, fill=red](x1) at (2.75,0.00)  [label=right: $y_i^3$]{};
    \node[draw, circle, line width=1pt, fill=blue](x2) at (2.23,1.62)  [label=right: $y_i^4$]{};
    \node[draw, circle, line width=1pt, fill=red](x3) at (0.85,2.62)  [label=above: $y_i^5$]{};
    \node[draw, circle, line width=1pt, fill=blue](x4) at (-0.85,2.62)  [label=above: $y_i^6$]{};
    \node[draw, circle, line width=1pt, fill=red](x5) at (-2.23,1.62)  [label=left: $y_i^7$]{};
    %\node[draw, circle, line width=1pt, fill=blue](x6) at (-2.75,0.00)  []{};
    \node[draw, circle, line width=1pt, fill=red](x7) at (-2.23,-1.62)  []{};
    \node[draw, circle, line width=1pt, fill=blue](x8) at (-0.85,-2.62)  [label=below: $y_i^q$]{};
    \node[draw, circle, line width=1pt, fill=red](x9) at (0.85,-2.62)  [label=below: $y_i^1$]{};
    \node[draw, circle, line width=1pt, fill=blue](x10) at (2.23,-1.62)  [label=right: $y_i^2$]{};

    \draw[-, line width=1pt]  (x1) -- (x2);
    \draw[-, line width=1pt]  (x2) -- (x3);
    \draw[-, line width=1pt]  (x3) -- (x4);
    \draw[-, line width=1pt]  (x4) -- (x5);
    \draw[-, line width=1pt]  (x7) -- (x8);
    \draw[-, line width=1pt]  (x8) -- (x9);
    \draw[-, line width=1pt]  (x9) -- (x10);
    \draw[-, line width=1pt]  (x10) -- (x1);

    %middle
    \node[draw, circle, line width=1pt, fill=white](y1) at (1.90,0.61)  []{};
    \node[draw, circle, line width=1pt, fill=white](y2) at (1.23,1.64)  []{};
    \node[draw, circle, line width=1pt, fill=white](y3) at (0.00,2.00)  []{};
    \node[draw, circle, line width=1pt, fill=white](y4) at (-1.23,1.64)  [label=right: $w_i^6$]{};
    %\node[draw, circle, line width=1pt, fill=white](y5) at (-1.90,0.61)  []{};
    %\node[draw, circle, line width=1pt, fill=white](y6) at (-1.90,-0.61)  []{};
    \node[draw, circle, line width=1pt, fill=white](y7) at (-1.23,-1.64)  []{};
    \node[draw, circle, line width=1pt, fill=white](y8) at (0.00,-2.00)  [label=right: $w_i^q$]{};
    \node[draw, circle, line width=1pt, fill=white](y9) at (1.23,-1.64)  [label=above: $w_i^1$]{};
    \node[draw, circle, line width=1pt, fill=white](y10) at (1.90,-0.61)  []{};

    %inner
    \node[draw, circle, line width=1pt, fill=white](z1) at (1.19,0.38)  []{};
    \node[draw, circle, line width=1pt, fill=white](z2) at (0.77,1.02)  []{};
    \node[draw, circle, line width=1pt, fill=white](z3) at (0.00,1.25)  []{};
    \node[draw, circle, line width=1pt, fill=white](z4) at (-0.77,1.02)  [label=below: $z_i^6$]{};
    %\node[draw, circle, line width=1pt, fill=white](z5) at (-1.19,0.38)  []{};
    %\node[draw, circle, line width=1pt, fill=white](z6) at (-1.19,-0.38)  []{};
    \node[draw, circle, line width=1pt, fill=white](z7) at (-0.77,-1.02)  []{};
    \node[draw, circle, line width=1pt, fill=white](z8) at (0.00,-1.25)  [label=above: $z_i^q$]{};
    \node[draw, circle, line width=1pt, fill=white](z9) at (0.77,-1.02)  [label=above: $z_i^1$]{};
    \node[draw, circle, line width=1pt, fill=white](z10) at (1.19,-0.38)  []{};

    %inner edges
    \draw[-, line width=1pt]  (x1) -- (y1);
    \draw[-, line width=1pt]  (x2) -- (y1);
    \draw[-, line width=1pt]  (y1) -- (z1);
    
    \draw[-, line width=1pt]  (x2) -- (y2);
    \draw[-, line width=1pt]  (x3) -- (y2);
    \draw[-, line width=1pt]  (y2) -- (z2);
    
    \draw[-, line width=1pt]  (x3) -- (y3);
    \draw[-, line width=1pt]  (x4) -- (y3);
    \draw[-, line width=1pt]  (y3) -- (z3);

    \draw[-, line width=1pt]  (x4) -- (y4);
    \draw[-, line width=1pt]  (x5) -- (y4);
    \draw[-, line width=1pt]  (y4) -- (z4);

    \draw[-, line width=1pt]  (x7) -- (y7);
    \draw[-, line width=1pt]  (x8) -- (y7);
    \draw[-, line width=1pt]  (y7) -- (z7);

    \draw[-, line width=1pt]  (x8) -- (y8);
    \draw[-, line width=1pt]  (x9) -- (y8);
    \draw[-, line width=1pt]  (y8) -- (z8);

    \draw[-, line width=1pt]  (x9) -- (y9);
    \draw[-, line width=1pt]  (x10) -- (y9);
    \draw[-, line width=1pt]  (y9) -- (z9);

    \draw[-, line width=1pt]  (x10) -- (y10);
    \draw[-, line width=1pt]  (x1) -- (y10);
    \draw[-, line width=1pt]  (y10) -- (z10);

    \path (x5) -- (-2.75,0.00) node [black, midway, sloped] {\Large$\dots$};
    \path (-2.75,0.00) -- (-2.23,-1.62) node [black, midway, sloped] {\Large$\dots$};
    
\end{tikzpicture}
}
}

\subfloat[The $x_i$-gadget.]{
\scalebox{0.7}{
\begin{tikzpicture}[inner sep=0.6mm]

    %outer
    \node[draw, circle, line width=0.5pt, fill=blue](x1) at (2.75,0.00)  [label=right: $x_i^2$]{};
    \node[draw, circle, line width=0.5pt, fill=red](x2) at (2.23,1.62)  [label=right: $x_i^3$]{};
    \node[draw, circle, line width=0.5pt, fill=blue](x3) at (0.85,2.62)  [label=above right: $x_i^4$]{};
    \node[draw, circle, line width=0.5pt, fill=red](x4) at (-0.85,2.62)  [label=above left: $x_i^5$]{};
    \node[draw, circle, line width=0.5pt, fill=blue](x5) at (-2.23,1.62)  [label=left: $x_i^6$]{};
    %\node[draw, circle, line width=1pt, fill=blue](x6) at (-2.75,0.00)  []{};
    \node[draw, circle, line width=0.5pt, fill=blue](x7) at (-2.23,-1.62)  [label=left: $x_i^p$]{};
    \node[draw, circle, line width=2pt, fill=red](x8) at (-0.85,-2.62)  [label=below: $u_i^1$]{};
    \node[draw, circle, line width=2pt, fill=blue](x9) at (0.85,-2.62)  [label=below: $u_i^4$]{};
    \node[draw, circle, line width=0.5pt, fill=red](x10) at (2.23,-1.62)  [label=right: $x_i^1$]{};

    \draw[-, line width=0.5pt]  (x1) -- (x2);
    \draw[-, line width=0.5pt]  (x2) -- (x3);
    \draw[-, line width=0.5pt]  (x3) -- (x4);
    \draw[-, line width=0.5pt]  (x4) -- (x5);
    \draw[-, line width=2pt]  (x7) -- (x8);
    \draw[-, line width=2pt]  (x8) -- (x9);
    \draw[-, line width=2pt]  (x9) -- (x10);
    \draw[-, line width=0.5pt]  (x10) -- (x1);

    %middle
    \node[draw, circle, line width=0.5pt, fill=white](y1) at (1.90,0.61)  []{};
    \node[draw, circle, line width=0.5pt, fill=white](y2) at (1.23,1.64)  []{};
    \node[draw, circle, line width=0.5pt, fill=white](y3) at (0.00,2.00)  []{};
    \node[draw, circle, line width=0.5pt, fill=white](y4) at (-1.23,1.64)  [label=right: $w_i^5$]{};
    %\node[draw, circle, line width=1pt, fill=white](y5) at (-1.90,0.61)  []{};
    %\node[draw, circle, line width=1pt, fill=white](y6) at (-1.90,-0.61)  []{};
    \node[draw, circle, line width=0.5pt, fill=white](y7) at (-1.23,-1.64)  []{};
    \node[draw, circle, line width=2pt, fill=white](y8) at (0.00,-2.00)  [label=right: $w_i$]{};
    \node[draw, circle, line width=2pt, fill=white](y9) at (1.23,-1.64)  [label=above: $w_i^0$]{};
    \node[draw, circle, line width=0.5pt, fill=white](y10) at (1.90,-0.61)  []{};

    %inner
    \node[draw, circle, line width=0.5pt, fill=white](z1) at (1.19,0.38)  []{};
    \node[draw, circle, line width=0.5pt, fill=white](z2) at (0.77,1.02)  []{};
    \node[draw, circle, line width=0.5pt, fill=white](z3) at (0.00,1.25)  []{};
    \node[draw, circle, line width=0.5pt, fill=white](z4) at (-0.77,1.02)  [label=below: $z_i^5$]{};
    %\node[draw, circle, line width=1pt, fill=white](z5) at (-1.19,0.38)  []{};
    %\node[draw, circle, line width=1pt, fill=white](z6) at (-1.19,-0.38)  []{};
    \node[draw, circle, line width=0.5pt, fill=white](z7) at (-0.77,-1.02)  []{};
    \node[draw, circle, line width=2pt, fill=white](z8) at (0.00,-1.25)  [label=above: $z_i$]{};
    \node[draw, circle, line width=2pt, fill=white](z9) at (0.77,-1.02)  [label=above: $z_i^0$]{};
    \node[draw, circle, line width=0.5pt, fill=white](z10) at (1.19,-0.38)  []{};

    %inner edges
    \draw[-, line width=0.5pt]  (x1) -- (y1);
    \draw[-, line width=0.5pt]  (x2) -- (y1);
    \draw[-, line width=0.5pt]  (y1) -- (z1);
    
    \draw[-, line width=0.5pt]  (x2) -- (y2);
    \draw[-, line width=0.5pt]  (x3) -- (y2);
    \draw[-, line width=0.5pt]  (y2) -- (z2);
    
    \draw[-, line width=0.5pt]  (x3) -- (y3);
    \draw[-, line width=0.5pt]  (x4) -- (y3);
    \draw[-, line width=0.5pt]  (y3) -- (z3);

    \draw[-, line width=0.5pt]  (x4) -- (y4);
    \draw[-, line width=0.5pt]  (x5) -- (y4);
    \draw[-, line width=0.5pt]  (y4) -- (z4);

    \draw[-, line width=0.5pt]  (x7) -- (y7);
    \draw[-, line width=2pt]  (x8) -- (y7);
    \draw[-, line width=0.5pt]  (y7) -- (z7);

    \draw[-, line width=2pt]  (x8) -- (y8);
    \draw[-, line width=2pt]  (x9) -- (y8);
    \draw[-, line width=2pt]  (y8) -- (z8);

    \draw[-, line width=2pt]  (x9) -- (y9);
    \draw[-, line width=2pt]  (x10) -- (y9);
    \draw[-, line width=2pt]  (y9) -- (z9);

    \draw[-, line width=0.5pt]  (x10) -- (y10);
    \draw[-, line width=0.5pt]  (x1) -- (y10);
    \draw[-, line width=0.5pt]  (y10) -- (z10);

    \path (x5) -- (-2.75,0.00) node [black, midway, sloped] {\Large$\dots$};
    \path (-2.75,0.00) -- (-2.23,-1.62) node [black, midway, sloped] {\Large$\dots$};

    %assignment part
    \node[draw, circle, line width=2pt, fill=white](u2) at (-1.85,-2.62)  [label=below: $u_i^2$]{};
    \node[draw, circle, line width=2pt, fill=white](u3) at (-2.85,-2.62)  [label=below: $u_i^3$]{};
    \node[draw, circle, line width=2pt, fill=white](u5) at (1.85,-2.62)  [label=below: $u_i^5$]{};
    \node[draw, circle, line width=2pt, fill=white](u6) at (2.85,-2.62)  [label=below: $u_i^6$]{};

    \draw[-, line width=2pt]  (x8) -- (u2);
    \draw[-, line width=2pt]  (u2) -- (u3);
    \draw[-, line width=2pt]  (x9) -- (u5);
    \draw[-, line width=2pt]  (u5) -- (u6);
    
\end{tikzpicture}
}
}
\caption{The gadgets used  in the proof of \Cref{thm:sigma_2P}. Inner colored vertices of the first appearance of $y_i$ are $y_i^2$ and $y_i^3$.}\label{fig:hardness-gadgets}
\end{figure}

To construct the graph $G$, we first build an auxiliary graph $H$ as follows. We define a \textit{variable} (\textit{clause} resp.) vertex for each variable (clause resp.) in $\phi$. We then add an edge between a variable and a clause vertex if the corresponding variable appears in the corresponding clause; let $H$ be the resulting graph. Observe that $H$ is a planar graph (due to the ``planarity'' of $\phi$). We now go from $H$ to $G$. First, replace each clause vertex $c$ of $H$ by a copy of the $c$-gadget shown in Figure~\ref{fig:hardness-gadgets}(a). Then, replace each variable vertex by either a $y$-gadget or a $x$-gadget, shown in Figures~\ref{fig:hardness-gadgets}(b) and~\ref{fig:hardness-gadgets}(c) respectively, according to the type of the corresponding variable in $\phi$. The $x$- and $y$-gadgets are similar. Consider the $y_i$-gadget, i.e., the gadget corresponding to the variable $y_i$ that appears in $\phi$. It will have four colored vertices (see Figure~\ref{fig:hardness-gadgets}(b)) for each appearance of $y_i$ in $\phi$; two red vertices and two blue. That is, the index $q$ that appears in Figure~\ref{fig:hardness-gadgets}(b) is equal to four times the number of appearances of $y_i$ in $\phi$. %For example, if $y_i$ appears in three clauses in a positive literal and in two clauses in a negative literal, then the $y_i$-gadget will have twenty colored vertices. 
Moreover, going in an anti-clockwise fashion, the colored vertices that correspond to each appearance of the $y_i$ variable will alternate, starting with a red and finishing with a blue; 
for the $j$-th appearance of variable $y_i$, we say that vertices $y^{4(j-1)+2}_i$ and $y^{4(j-1)+3}_i$ denote its \emph{inner colored vertices} 
%we denote as \textit{inner colored vertices} the first blue and second red vertex we meet going in the anti-clockwise fashion established earlier 
(see Figure~\ref{fig:hardness-gadgets}(b) for an example). The inner blue (red resp.) vertex included in the gadget for an appearance of the variable $y_i$ will model that this variable is set to false (true resp.), while the other inner colored vertices will serve as points of additional connection between the gadgets. 
The same holds true for the $x_i$-gadget which, in addition, contains an extra set of colored vertices together with two pending paths, illustrated by the bold vertices and edges in Figure~\ref{fig:hardness-gadgets}(c). 

At this stage, all the original edges of $H$ are removed. We add the new edges between the gadgets through a two-step \textit{edge-adding procedure}. First we deal with the edges connecting the $c$-gadgets to the $x$- and/or $y$-gadgets.  
%We say that all the vertices of the $x$- and $y$-gadgets that are of degree $4$ at this stage, and the vertices $\ell_c^1,\ell_c^2$ and $\ell_c^3$ of the $c$-gadgets, are \textit{unmarked}. 
Consider a clause $c$ and its corresponding $c$-gadget and assume that, in the graph $H$ there was an edge between the clause vertex $c$ and the variable vertex of type $x$. %\manolis{Here do we want to say that $x$ is a variable of type $x$? If so I suggest to rename $x$ to $x_i$.}
%, and $x$ appears in $c$ as a positive (negative resp.) literal. 
Also, let $x$ be the $j$-th appearance of the $x$ variable in $\phi$ (according to a carefully chosen ordering of the variables. %\manolis{Do we mean to say ordering of clauses? The appearance of the variables is dictated by an ordering of the clauses.}). 
Then, going anti-clockwise, we locate the $j$-th quadruple $Q$ of colored vertices of the $x$-gadget. We then add an edge between any vertex of the $c$-gadget that is currently of degree three and the blue (red resp.) inner vertex of $Q$ if this appearance of $x$ is positive (negative resp.). 
We repeat the same procedure for all the edges of $H$ that are between the clause vertex $c$ and any variable vertex of type $y$. Once we are done, we move on and repeat this procedure for every clause vertex of $H$. This completes the first step of the edge-adding procedure. 

In the second step, we connect the $x$- and/or $y$-gadgets whose corresponding variables appear in a common clause. To ease the exhibition, and since we treat these gadgets in the same way, we will assume we only have to deal with $x$-gadgets. So, consider a clause gadget $c$, with the corresponding clause being comprised of three literals on the variables $x_{i_1}$, $x_{i_2}$, and $x_{i_3}$ (for some $i_1,i_2,i_3 \in [n_1])$.
According to the construction of $G$ up to this point, there are 
\begin{itemize}
    \item a $c$-gadget, containing the vertices $\ell_c^1$, $\ell_c^2$, and $\ell_c^3$, and
    \item the $x_{i_1}$, $x_{i_2}$, and $x_{i_3}$-gadgets, containing some inner colored vertices $x_{i_1}^\alpha$, $x_{i_2}^\beta$, and $x_{i_3}^\gamma$ respectively
    such that $G$ contains the edges $\ell_c^1 x_{i_1}^\alpha, \ell_c^2 x_{i_2}^\beta$, and $\ell_c^3 x_{i_3}^\gamma$.
\end{itemize}
Note that since $x_{i_1}^\alpha$, $x_{i_2}^\beta$, and $x_{i_3}^\gamma$ are inner colored vertices, and according to the first step of the edge-adding procedure, the two colored neighbors of these vertices that lie in the $x_{i_1}$, $x_{i_2}$, and $x_{i_3}$-gadgets respectively are currently of degree $4$. The second step of the edge-adding procedure consists in adding the edges $x_{i_3}^{\gamma-1} x_{i_1}^{\alpha+1}$, $x_{i_1}^{\alpha-1} x_{i_2}^{\beta+1}$, and $x_{i_2}^{\beta-1} x_{i_3}^{\gamma+1}$. This step is repeated for every clause gadget $c$.  

This marks the end of the construction of $G$. We ensure the planarity of~$G$ based on the planarity of $H$ and by carefully choosing the ordering used in the first step and by bending the edges added in the second step of the edge-adding procedure to closely shadow the preexisting edges of $G$.

We are now ready to sketch our reduction. Assume that we have a yes-instance of \UQPSAT testified by the truth-assignment $\sigma$. We will show that the constructed graph $G=(V,E)$ has a \MUVCShort of order exactly~$n_1$. Due to the $x$-gadgets, it suffices to provide a set $S\subseteq V$ such that $G'=G-S$ has a unique minimum vertex cover~$U$, and $|S|=n_1$. We proceed as follows: for each $i\in[n_1]$, if $\tau(x_i)=\true$ ($\tau(x_i)=\false$ resp.), we include $u_i^5$ ($u_i^2$ resp.) in $S$. It follows that $|S|=n_1$. It remains to show that $G'=G-S$ has a unique minimum vertex cover $U$. Observe that for every variable gadget, if $U$ contains a blue (red resp.) vertex, then it should also contain all other blue (red resp.) vertices, and no red (blue resp.) vertex of the same gadget. In particular, $U$ will contain exactly all the $w$s, the blue (red resp.) vertices of the $x$-gadgets and $y$-gadgets that are set to true (false resp.) and two out the three outer vertices of each $c$-gadget whose corresponding literals do not satisfy $c$. The uniqueness and minimiality of this $U$ is guaranteed by the formula $\phi$ being $1$-in-$3$ satisfied by $\sigma$ and vice versa. 
\end{sketch}
\fi

\section*{Acknowledgments}
The first, second, and fourth authors were supported by the European Union under the project Robotics and advanced industrial production (reg. no. CZ.02.01.01/00/22\_008/0004590). The first author was additionally supported by the International Mobility of Researchers MSCA-F-CZ-III at CTU in Prague, $\text{CZ}.02.01.01/00/22\_010/0008601$ Programme Johannes Amos Comenius. The third author was partially supported by Charles University projects UNCE 24/SCI/008 and
PRIMUS 24/SCI/012, and by the project 25-17221S of GAČR.
The fifth author was supported by the ANR project ANR-21-CE48-0022 (S-EX-AP-PE-AL) and the Barrande Fellowship programme.
\bibliography{aaai2026}

@inproceedings{horiyama2024pauvc,
  title={Theoretical Aspects of Generating Instances with Unique Solutions: Pre-assignment Models for Unique Vertex Cover},
  author={Horiyama, Takashi and Kobayashi, Yasuaki and Ono, Hirotaka and Seto, Kazuhisa and Suzuki, Ryu},
  booktitle={Proceedings of the AAAI Conference on Artificial Intelligence},
  volume={38},
  number={18},
  pages={20726--20734},
  year={2024}
}

@inproceedings{an2024pre,
  author       = {Shinwoo An and
                  Yeonsu Chang and
                  Kyungjin Cho and
                  O{-}joung Kwon and
                  Myounghwan Lee and
                  Eunjin Oh and
                  Hyeonjun Shin},
  title        = {Pre-Assignment Problem for Unique Minimum Vertex Cover on Bounded Clique-Width Graphs},
  booktitle    = {AAAI-25, Sponsored by the Association for the Advancement of Artificial Intelligence},
  pages        = {26886--26894},
  publisher    = {{AAAI} Press},
  year         = {2025},
  doi          = {10.1609/AAAI.V39I25.34893}
}

@article{oum08,
	author       = {Sang{-}il Oum},
	title        = {Approximating rank-width and clique-width quickly},
	journal      = {{ACM} Trans. Algorithms},
	volume       = {5},
	number       = {1},
	pages        = {10:1--10:20},
	year         = {2008},
	url          = {https://doi.org/10.1145/1435375.1435385},
	doi          = {10.1145/1435375.1435385}
}

@article{demaine18sigma2,
	author       = {Erik D. Demaine and
	Fermi Ma and
	Ariel Schvartzman and
	Erik Waingarten and
	Scott Aaronson},
	title        = {The fewest clues problem},
	journal      = {Theor. Comput. Sci.},
	volume       = {748},
	pages        = {28--39},
	year         = {2018},
	url          = {https://doi.org/10.1016/j.tcs.2018.01.020},
	doi          = {10.1016/J.TCS.2018.01.020}
}

@misc{benchmarkmaximumcut,
      title={A Benchmark for {Maximum Cut}: Towards Standardization of the Evaluation of Learned Heuristics for Combinatorial Optimization}, 
      author={Ankur Nath and Alan Kuhnle},
      year={2024},
      eprint={2406.11897},
      archivePrefix={arXiv},
      primaryClass={cs.AI},
      url={https://arxiv.org/abs/2406.11897}, 
}

@misc{nphardeval4v,
    title={{NPHardEval4V}: A Dynamic Reasoning Benchmark of Multimodal Large Language Models},
    author={Lizhou Fan and Wenyue Hua and Xiang Li and Kaijie Zhu and Mingyu Jin and Lingyao Li and Haoyang Ling and Jinkui Chi and Jindong Wang and Xin Ma and Yongfeng Zhang},
    year={2024},
    eprint={2403.01777},
    archivePrefix={arXiv},
    primaryClass={cs.CL}
}

@misc{nphardeval,
    title={{NPHardEval}: Dynamic Benchmark on Reasoning Ability of Large Language Models via Complexity Classes},
    author={Lizhou Fan and Wenyue Hua and Lingyao Li and Haoyang Ling and Yongfeng Zhang},
    year={2023},
    eprint={2312.14890},
    archivePrefix={arXiv},
    primaryClass={cs.AI}
}

@misc{grapharena,
    title={GraphArena: Benchmarking Large Language Models on Graph Computational Problems},
    author={Jianheng Tang and Qifan Zhang and Yuhan Li and Jia Li},
    year={2024},
    eprint={2407.00379},
    archivePrefix={arXiv},
    primaryClass={cs.AI}
}

@article{tsplib,
  title={{TSPLIB}—A traveling salesman problem library},
  author={Reinelt, Gerhard},
  journal={ORSA journal on computing},
  volume={3},
  number={4},
  pages={376--384},
  year={1991},
  publisher={INFORMS}
}

@misc{asuncion2007uci,
  title={UCI machine learning repository},
  author={Asuncion, Arthur and Newman, David and others},
  year={2007},
  publisher={Irvine, CA, USA}
}

@article{hoos2000satlib,
  title={{SATLIB}: An online resource for research on SAT},
  author={Hoos, Holger H and St{\"u}tzle, Thomas},
  journal={Sat},
  volume={2000},
  pages={283--292},
  year={2000}
}

@article{koch2011miplib,
  title={{MIPLIB} 2010: Mixed integer programming library version 5},
  author={Koch, Thorsten and Achterberg, Tobias and Andersen, Erling and Bastert, Oliver and Berthold, Timo and Bixby, Robert E and Danna, Emilie and Gamrath, Gerald and Gleixner, Ambros M and Heinz, Stefan and others},
  journal={Mathematical Programming Computation},
  volume={3},
  pages={103--163},
  year={2011},
  publisher={Springer}
}

@article{chang2011libsvm,
  title={{LIBSVM}: a library for support vector machines},
  author={Chang, Chih-Chung and Lin, Chih-Jen},
  journal={ACM transactions on intelligent systems and technology (TIST)},
  volume={2},
  number={3},
  pages={1--27},
  year={2011},
  publisher={Acm New York, NY, USA}
}

@techreport{hagberg2008exploring,
  title={Exploring network structure, dynamics, and function using NetworkX},
  author={Hagberg, Aric and Swart, Pieter J and Schult, Daniel A},
  year={2008},
  institution={Los Alamos National Laboratory (LANL), Los Alamos, NM (United States)}
}

@book{Diestel17,
  author    = {Reinhard Diestel},
  title     = {Graph Theory},
  publisher = {Springer},
  year      = {2025},
  series    = {Graduate texts in mathematics},
  volume    = {173},
  isbn      = {978-3-662-70106-5}
}

@article{HudryL19,
  title     = {Complexity of unique (optimal) solutions in graphs: Vertex Cover and Domination},
  author    = {Olivier Hudry and Antoine Lobstein},
  journal   = {Journal of Combinatorial Mathematics and Combinatorial Computing},
  volume    = {110},
  pages     = {217--240},
  year      = {2019}
}

@book{CyganFKLMPPS15,
  author       = {Marek Cygan and
                  Fedor V. Fomin and
                  \L{}ukasz Kowalik and
                  Daniel Lokshtanov and
                  D{\'{a}}niel Marx and
                  Marcin Pilipczuk and
                  Micha\l{} Pilipczuk and
                  Saket Saurabh},
  title        = {Parameterized Algorithms},
  publisher    = {Springer},
  year         = {2015},
  url          = {https://doi.org/10.1007/978-3-319-21275-3},
  doi          = {10.1007/978-3-319-21275-3},
  isbn         = {978-3-319-21274-6},
}

@book{Niedermeier06,
  author       = {Rolf Niedermeier},
  title        = {Invitation to Fixed-Parameter Algorithms},
  publisher    = {Oxford University Press},
  year         = {2006},
  url          = {https://doi.org/10.1093/ACPROF:OSO/9780198566076.001.0001},
  doi          = {10.1093/ACPROF:OSO/9780198566076.001.0001},
  isbn         = {9780198566076}
}

@book{FlumG06,
  author       = {J{\"{o}}rg Flum and
                  Martin Grohe},
  title        = {Parameterized Complexity Theory},
  series       = {Texts in Theoretical Computer Science. An {EATCS} Series},
  publisher    = {Springer},
  year         = {2006},
  url          = {https://doi.org/10.1007/3-540-29953-X},
  doi          = {10.1007/3-540-29953-X},
  isbn         = {978-3-540-29952-3}
}

@book{downey2012parameterized,
  title={Parameterized complexity},
  author={Downey, Rodney G and Fellows, Michael Ralph},
  year={2012},
  publisher={Springer Science \& Business Media}
}

@article{CourcelleO00,
  author       = {Bruno Courcelle and
                  Stephan Olariu},
  title        = {Upper bounds to the clique width of graphs},
  journal      = {Discret. Appl. Math.},
  volume       = {101},
  number       = {1-3},
  pages        = {77--114},
  year         = {2000},
  url          = {https://doi.org/10.1016/S0166-218X(99)00184-5},
  doi          = {10.1016/S0166-218X(99)00184-5}
}

@book{Kloks94,
  author       = {Ton Kloks},
  title        = {Treewidth, Computations and Approximations},
  series       = {Lecture Notes in Computer Science},
  volume       = {842},
  publisher    = {Springer},
  year         = {1994},
  url          = {https://doi.org/10.1007/BFb0045375},
  doi          = {10.1007/BFb0045375},
  isbn         = {3-540-58356-4}
}

@inproceedings{KorhonenL23,
  author       = {Tuukka Korhonen and
                  Daniel Lokshtanov},
  editor       = {Barna Saha and
                  Rocco A. Servedio},
  title        = {An Improved Parameterized Algorithm for Treewidth},
  booktitle    = {Proceedings of the 55th Annual {ACM} Symposium on Theory of Computing,
                  {STOC} 2023, Orlando, FL, USA, June 20-23, 2023},
  pages        = {528--541},
  publisher    = {{ACM}},
  year         = {2023},
  url          = {https://doi.org/10.1145/3564246.3585245},
  doi          = {10.1145/3564246.3585245}
}

@article{Bodlaender96,
  author       = {Hans L. Bodlaender},
  title        = {A Linear-Time Algorithm for Finding Tree-Decompositions of Small Treewidth},
  journal      = {SIAM Journal on Computing},
  volume       = {25},
  number       = {6},
  pages        = {1305--1317},
  year         = {1996},
  doi          = {10.1137/S0097539793251219},
}

@article{B98,
title = "A partial $k$-arboretum of graphs with bounded treewidth",
journal = "Theoretical Computer Science",
volume = "209",
number = "1",
pages = "1-45",
year = "1998",
author = "Hans L. Bodlaender",
doi = {10.1016/S0304-3975(97)00228-4}
}

\end{document}